\documentclass[11pt,reqno]{amsart}

\usepackage{amsmath}
\usepackage{latexsym}
\usepackage{amssymb}
\usepackage{mathrsfs}
\usepackage{framed}
\usepackage{stmaryrd}
\usepackage{graphicx}
\usepackage[utf8]{inputenc}
\usepackage{dsfont}
\usepackage{enumerate}
\usepackage{enumitem}
\usepackage{hyperref}
\usepackage{color}
\usepackage{tikz}

\usepackage{multicol}

\usepackage{amsmath}
\usepackage{pgfplots}
\pgfplotsset{compat=1.18}

\usepackage[left=2.5cm,right=2.5cm,top=3cm,bottom=3cm]{geometry}

\usepackage{xcolor}
\usepackage{comment}
  \usepackage{soul}

\newcommand {\clr}{\color{red}}

\newcommand{\one}{\mathds{1}}

\newtheorem{theorem}{Theorem}[section]
\newtheorem{definition}[theorem]{Definition}%[section]
\newtheorem{prop}[theorem]{Proposition}%[section]
\newtheorem{lemma}[theorem]{Lemma}%[section]
\newtheorem{remark}[theorem]{Remark}%[section]
\newtheorem{defiprop}[theorem]{Definition-Proposition}
\newtheorem{hypo}{Hypothesis}

\newtheorem{exem}[theorem]{Example}

\newcommand{\R}{\mathbb{R}}             % REAL
\newcommand{\N}{\mathbb{N}}             % INTEGER
\newcommand{\Z}{\mathbb{Z}}             %
\newcommand{\C}{\mathbb{C}}             % COMPLEX

\numberwithin{equation}{section}

\newcommand{\Hi}{\mathcal{H}}
\newcommand{\Res}{\mathcal{R}}

\DeclareMathOperator{\Det}{Det}
\DeclareMathOperator{\Tr}{Tr}

\DeclareMathOperator{\Id}{Id}
\DeclareMathOperator{\re}{Re}
\DeclareMathOperator{\im}{Im}
\DeclareMathOperator{\Ran}{Ran}
\DeclareMathOperator{\Ker}{Ker}
\DeclareMathOperator{\supp}{supp}

\renewcommand {\clr} {\color{red}}

\begin{document}

	\title{The spectral shift function for  non-self-adjoint perturbations}
	
	\author[V.~Bruneau]{Vincent Bruneau}
	\author[N.~Frantz]{Nicolas Frantz}
	\author[F.~Nicoleau]{Fran\c cois Nicoleau}
	
	\maketitle
	
	\begin{abstract}
		This paper is devoted to the definition and analysis of the spectral shift function (SSF) associated with non-self-adjoint perturbations of self-adjoint operators. Motivated by applications in scattering theory, we consider both trace-class and relatively trace-class perturbations. We extend the Lifshits–Kreĭn trace formula to non-self-adjoint operators under suitable assumptions on the spectrum and the behavior of the resolvent. The role of spectral singularities is carefully analyzed, and we provide a generalization of the SSF using functional calculus. Finally, we apply our results to Schrödinger operators with complex-valued short-range potentials in dimension three. Toy models illustrate properties that one might hope to extend to general cases.  In particular, they suggest that the SSF carries information on the presence of complex eigenvalues.
	\end{abstract}
	
	%%%%%%%%%%%%%%%%%%%%%%%%%%%%%%%%%  Commentaires %%%%%%%%%%%%%%%%%%%%%%%%%%
	
	%%%%%%%%%%%%%%%%%%%%%%%%%%%%%%%%%%%%%%%%%%%%%%%%%%%%%%%%%%%%%%%%%%%%

	\tableofcontents
	
	\section{Introduction}

	The Spectral Shift Function (SSF), initially introduced by Lifshits and Kreĭn \cite{Li52, Kr53}, provides a framework for the spectral analysis of trace-class (or relatively trace-class) perturbations of a reference operator. It yields a spectral invariant that is often related to scattering quantities such as the \textit{scattering phase} and the \textit{average time delay} (see references \cite{Pu02, Ya92}). More precisely, in the self-adjoint setting, the derivative of the SSF with respect to energy coincides (up to a factor of $1/2\pi$) with the scattering phase shift, as established by the Birman–Kreĭn formula. This connection allows one to interpret the SSF as encoding the cumulative spectral effect of the scattering process. Furthermore, the Eisenbud–Wigner formula shows that the time delay operator (measuring the difference in the sojourn time due to the interaction), is also expressed in terms of the energy derivative of the scattering matrix, and thus related to the SSF.
	
	\vspace{0.1cm}
	Originally defined for pairs of self-adjoint or unitary operators, the SSF has been extended to contraction operators viewed as perturbations of unitary operators (see \cite{Ry86, Ry95}), and to dissipative (or accumulative) operators interpreted as boundary perturbations of self-adjoint operators (see \cite{MaNe15, MaNePe19}). 
	A Levinson's formula has recently been obtained in \cite{AlFaRi25_09} for dissipative operators.
In these works, at least one side of the complex plane (the upper or lower half-plane, or the exterior of the unit disk) lies in the resolvent set. The SSF for a general pair of operators in a Banach space is considered in \cite{Mi19}, where it is defined on $(0, +\infty)$, which is assumed to be a subset of the resolvent set.

	\vspace{0.1cm}
	Our goal here is to consider a reference operator $H_0$ which remains self-adjoint, and to define a \textit{Spectral Shift Function} (SSF for short) for non-self-adjoint perturbations of $H_0$. With the objective of establishing connections with scattering theory, we aim to define the SSF on the real line $\mathbb{R}$, which contains the essential spectrum of the operators under consideration. That is, for a general relatively compact perturbation $H$ of a self-adjoint operator $H_0$ (in particular, $H$ may have eigenvalues on both sides of the real axis), we seek to define and analyze a function $\xi := \xi(\cdot; H, H_0)$ such that for any $f\in  D(\R)$, the space of smooth functions with compact support, 
	\begin{equation}\label{ssf}
		\Tr(f(H) - f(H_0)) = \int_{\mathbb{R}} \xi(\lambda)\, f'(\lambda)\, d\lambda.
	\end{equation}
	% where $f(H)$ will be defined using the Helffer--Sj\"ostrand formula \eqref{eq:HSformula}, with an almost analytic extension of $f$.
	
	A particularly relevant model in this context is the Schrödinger operator $H = H_0 + V$, where $H_0 = -\Delta$ on $L^2(\mathbb{R}^d)$ and $V$ is a complex-valued electric potential vanishing at infinity.

	\vspace{0.1cm}
	
	As in \cite{Da95_08}, the operator $f(H)$ for $f \in D(\mathbb{R})$ is defined via the Helffer--Sjöstrand formula \eqref{eq:HSformula}, using an almost analytic extension of $f$. This construction is extend to our non-self-adjoint setting in section \ref{sec:functional calculus}. We then study the left-hand side of \eqref{ssf}, first for trace-class perturbations and subsequently for relatively trace-class perturbations. We show that it defines the derivative of the SSF in the sense of distributions, and that for lower-bounded operators, the SSF may be chosen to vanish near $- \infty$. %\footnote{Parler du principe d'invariance.}
	
	In the original works of Lifshits and Kreĭn \cite{Li52, Kr53}, in the self-adjoint setting, the existence of the SSF and its integrability properties were rigorously justified via its relation to the perturbation determinant:
	\begin{equation}\label{det}
		\xi(\lambda) = \frac{1}{\pi} \lim_{\varepsilon \to 0^+} \arg D(\lambda + i\varepsilon); \qquad 
		D(z):=
		\Det \left( I + V (H_0 - z)^{-1} \right).
	\end{equation}
	For non-self-adjoint operators $H$, this formula no longer holds since
the perturbation determinant may vanish and the property $D(\overline{z})= \overline{D(z)}$ is no longer true. However, using factorization theorems for holomorphic functions on $\mathbb{C}^+$, where $\mathbb{C}^\pm=\{z\in\mathbb{C}\mid \pm \im z>0\}$ (see for instance \cite{MaNe15}), one can still derive representation results for the SSF in terms of measures,  but as we will see in the toy models, these measures are no longer necessarily real and nor absolutely continuous.
	
	\vspace{0.3cm}\noindent
	Throughout the paper, we will use the following notations:
	
	\vspace{0.1cm}\noindent
	Given a Banach space $\mathcal{A}$, we denote by $\mathcal{B}(\mathcal{A})$ the set of bounded linear operators on $\mathcal{A}$. For a Hilbert space $\mathcal{H}$, the Schatten class of order $p$ is denoted by $\mathcal{L}_p(\mathcal{H})$. If $A$ is a closed operator, its spectrum is denoted by $\sigma(A)$ and its resolvent set by $\rho(A)$. The kernel and range of $A$ are denoted respectively by $\Ker(A)$ and $\Ran(A)$. Given a self-adjoint operator $H_0$, we denote by $\Res_0(z) := (H_0 - z)^{-1}$ its resolvent at $z \in \rho(H_0)$, and by $\Res(z) := (H - z)^{-1}$ the resolvent of another (possibly non-self-adjoint) operator $H$. For an subset $I \subset \mathbb{R}$, we denote by $D(I)$ the space of smooth functions with compact support in $I$ and values in $\mathbb{C}$. Its topological dual is denoted by $D'(I)$. For intervals $I\subset\R$ and  $a\in(0,\infty]$, we introduce the following subset of the complex plane:	
	\begin{equation*}
		S_{a}(I):= \lbrace z\in\C \,\big|\, \re(z)\in I\text{ and } 0<|\im(z)|<a\rbrace,\quad 
		S_a^\pm(I):= S_{a}(I) \cap \C^\pm.
		%\lbrace z\in\C^\pm \,\big|\,  \re(z)\in I\text{ and } \pm \im(z)<a\rbrace.
	\end{equation*}

		\section{Assumptions and main results}\label{Sec1}

\subsection{Abstract setting}\label{sub:setting}

	We consider an operator $H$ acting on a complex separable Hilbert space $\mathcal{H}$. We assume that $H$ is of the form
	\begin{equation*}
		H := H_0 + V,
	\end{equation*}
	where $(H_0, \mathcal{D}(H_0))$ is a self-adjoint operator bounded from below, and $V$ is bounded and relatively compact with respect to $H_0$. In particular, $H$ is closed with domain $\mathcal{D}(H) = \mathcal{D}(H_0)$, and its adjoint is given by
	\begin{equation*}
		H^\ast = H_0 + V^\ast, \quad \mathcal{D}(H^\ast) = \mathcal{D}(H_0).
	\end{equation*}
	
	%The spectrum of $H$, denoted by $\sigma(H)$, is the set of complex numbers $z$ such that the operator $H - z$ does not admit a bounded inverse. Its complement in $\C$ is called the resolvent set and is denoted by $\rho(H)$.
	We define the point spectrum of $H$ as
	\begin{equation*}
		\sigma_\mathrm{p}(H) := \left\lbrace \lambda \in \C \mid \Ker(H - \lambda) \neq \{0\} \right\rbrace,
	\end{equation*}
	that is, the set of eigenvalues of $H$. For an isolated eigenvalue $\lambda \in \sigma_\mathrm{p}(H)$, the corresponding Riesz projection is given by
	\begin{equation}\label{eq:RieszProjection}
		\Pi_\lambda(H) := \frac{1}{2\pi i} \oint_{\Gamma} (z - H)^{-1} \, \mathrm{d}z,
	\end{equation}
	where $\Gamma = C(\lambda, r)$ is a positively oriented circle centered at $\lambda$ with radius $r > 0$ small enough so that $\lambda$ is the only point of the spectrum $\sigma(H)$ lying inside the disk $D(\lambda, r)$.
	
	%The operator $\Pi_\lambda(H)$ is the projection onto the generalized eigenspace associated with $\lambda$. In other words, the range $\Ran(\Pi_\lambda(H))$ consists of all vectors $u \in \mathcal{D}(H^k)$ such that $(H - \lambda)^k u = 0$ for some $k \in \N$.
	
	An isolated eigenvalue is called \emph{discrete} if the range of the corresponding Riesz projection is of finite dimension. We denote by $\sigma_\mathrm{disc}(H)$ the set of discrete eigenvalues of $H$. The \emph{essential spectrum} of $H$, denoted by $\sigma_\mathrm{ess}(H)$, is defined as the complement of the discrete spectrum in the full spectrum of $H$:
	$$
	\sigma_\mathrm{ess}(H) := \sigma(H) \setminus \sigma_\mathrm{disc}(H).
	$$
	
	With this definition, and since $V$ is relatively compact with respect to $H_0$, it follows from Weyl’s theorem (see e.g. \cite[Theorem XIII.14]{ReSi80_01}) that
	\begin{equation}\label{eq:ess spec}
	\sigma_\mathrm{ess}(H) = \sigma_\mathrm{ess}(H_0).
	\end{equation}

Moreover, using a variational argument and the boundedness of $V$, one sees that $\sigma(H)$ lies in a horizontal strip and that its real part is bounded from below, since $H=H_0+V$ is a bounded perturbation of $H_0$.

\subsection{Assumptions}

Throughout this paper, some hypotheses are understood to be stated relative to a fixed open interval $I\subset\R$ which is chosen according to the spectral region under consideration. Whenever a hypothesis is invoked, the corresponding set $I$ is assumed to be fixed.
	
\begin{hypo}\label{hyp:eigenvalues}
		Let $I\subset\mathbb{R}$ be an open interval. The perturbed operator $H$ has only finitely many non-real eigenvalues whose real parts lie in $I$. %There exists $a_I>0$ such that $S_{a_I}(I)$ does not contains any nonreal eigenvalues of $H$. 
		
\end{hypo}
	
The following assumption requires the resolvent of $H$ blow up at most polynomially in a tubular neighborhood of the  real interval $I$. %line.

\begin{hypo}\label{hyp:estimate resolvente}
	Let $I\subset \R$ be an open interval. There exists $a_I>0$, $n_I>0$ and $c_I>0$, such that  $S_{a_I}(I)\subset \rho(H)$ and for all $z\in S_{a_I}(I)$,
	\begin{equation}\label{eq:hypestimeresol}
	%\footnote{N:Ce n'est pas impossible qu'il faille fixer un compact pour la partie réelle.}
		\|\Res_H(z)\|_{\mathcal{B}(\Hi)}\leq c_I |\im(z)|^{-1}\left(\frac{\left\langle \re(z)\right\rangle }{|\im(z)|}\right)^{n_I}.
	\end{equation} 
\end{hypo}

%\footnote{\clr F : les hypothèses précédentes ainsi écrites ne sont pas claires. Est-ce pour un intervalle $I$ fixé ou bien pour tout $I$ ? N: J'ai fais une remarque au début de la section. Est-ce que cela te semble suffisant ?}	
	
The next assumption require that the difference between  the resolvents (or of the operators) is a trace class operator. 
\begin{hypo}\label{hyp:Trace}
	There exists $c \in \R$ and $m\in\Z$, such that 
	\begin{equation*}
		(H-c)^{-m}-(H_0-c)^{-m}\in\mathcal{L}_1(\Hi)
	\end{equation*}
\end{hypo}
	
\begin{remark}\label{rqHyp} Let us comment these hypotheses in particular cases.
	\begin{itemize}
		\item If Hypothesis~\ref{hyp:eigenvalues} holds on $I=\R$, then $H$ has a finite number of non-real eigenvalues.
		\item If $I$ is a bounded interval, Hypothesis~\ref{hyp:eigenvalues} means there is no accumulation of discrete eigenvalues to $I\times \{0\}$.% (outside the real axis the spectrum is discrete with bounded imaginary part).
		\item If $I$ is a bounded interval, the estimate \eqref{eq:hypestimeresol} in Hypothesis~\ref{hyp:estimate resolvente} is equivalent to 
		\begin{equation*}
			\|\Res_H(z)\|_{\mathcal{B}(\Hi)}\leq c_I \  |\im(z)|^{-n_I-1}, \quad \forall z\in S_{a_I}(I). 
		\end{equation*}
		\item If $I \subset (- \infty, c )$ with $c < \inf \sigma_\mathrm{ess}(H_0)$ then both  Hypothesis~\ref{hyp:eigenvalues} and \ref{hyp:estimate resolvente} hold on $I$. % (only a finite number of eigenvalues are in the half-plane $\{ \re(z) \} c\}$).}
		\item The case $m=-1$ in Hypothesis~\ref{hyp:Trace} means that $V=H-H_0$ belongs to $\mathcal{L}_1(\Hi)$. 
	\end{itemize}
\end{remark}

 \subsection{Main results}

We now summarize the main results of the paper. Our first result concerns the construction of a functional calculus for the
non-self-adjoint operator $H = H_0 + V$ on the real line. Since $H$ may possess
non-real eigenvalues, we separate the spectral contribution associated with
these eigenvalues and apply the Helffer--Sj\"ostrand functional calculus,
following the approach of Davies~\cite{Da95_08, Da07}, to the part of the operator
whose spectrum lies on the real axis.

\medskip\noindent
Under Hypotheses~1 and~2 on an open interval $I \subset \R$, this construction
allows us to define the operator $f(H)$ for suitable functions $f$ supported in $I$
through a Helffer-Sj\"ostrand type formula. Moreover, the map
\[
f \longmapsto f(H)
\]
defines a continuous algebra morphism for the pointwise multiplication, (see Section \ref{sec:functional calculus}).

\medskip\noindent
Our second result establishes the existence of the spectral shift function
for trace-class perturbations and is proved in Section 4.

\begin{theorem}[Existence of the SSF]
	\label{thm:ssf-traceMR}
	Assume that Hypotheses~1, 2, and~3 hold on an open interval $I$ with $m=-1$.
	Then for every $f \in \mathcal D(I)$ the operator difference
	$f(H)-f(H_0)$ belongs to $\mathcal L_1(\mathcal H)$ and the map
	\[
	f \mapsto \mathrm{Tr}(f(H)-f(H_0))
	\]
	defines a distribution on $I$. 
	The spectral shift function $\xi(\cdot; H, H_0)$ is defined, up to an additive constant, by
	\[
	\langle \xi', f \rangle = \mathrm{Tr}(f(H)-f(H_0)).
	\]
	Moreover, it satisfies
	\begin{equation}\label{ssf=sigmaMR}
		\xi^\prime(\cdot; H, H_0) =
		\frac{1}{2\pi i}
		\lim_{\varepsilon \rightarrow 0^+}
		\left(
		\sigma(\cdot + i\varepsilon) - \sigma(\cdot - i\varepsilon)
		\right)
		\quad \text{in } \mathcal{D}'(I),
	\end{equation}
	where for $z \in \rho(H) \cap \rho(H_0)$,
	\[
	\sigma(z) := \Tr\big(\Res_H(z) - \Res_{H_0}(z)\big).
	\]
\end{theorem}

\medskip\noindent
Note that unlike the self-adjointness case where $\sigma(\cdot - i\varepsilon)= \overline{\sigma(\cdot + i\varepsilon)}$, the right hand side of \eqref{ssf=sigmaMR} is not necessarily real.
However, as in the self-adjoint setting, in Section \ref{Sec:relTr}, we then extend this construction to relatively trace-class perturbations
by using a spectral change of variables.

\begin{theorem}[Relatively trace-class perturbations]
\label{thm:ssf-relativeMR}
Suppose that Hypotheses 1, 2 and 3 hold on $I$ for some $m \in \N^*$.
Then the spectral shift function for the pair $(H,H_0)$ can be defined by
\[
\xi(\lambda;H,H_0)
:=
\xi\!\left((\lambda+c)^{-m}; (H+c)^{-m}, (H_0+c)^{-m}\right),
\]
where $\xi(\cdot; (H+c)^{-m},(H_0+c)^{-m})$ denotes the SSF associated with
the trace-class pair $((H+c)^{-m},(H_0+c)^{-m})$, with some $c\gg 1$.
\end{theorem}

\noindent
Under the assumptions of the previous theorem, together with an additional technical condition, one can derive a formula similar to \eqref{ssf=sigmaMR}. This representation will be the key tool in the analysis of the spectral shift function for Schr\"odinger operators.

\medskip\noindent
In the end, we focus on Schr\"odinger operators with (smooth) complex-valued compactly supported potentials in dimension three.
In this framework we introduce the notions of outgoing and incoming
spectral singularities $\lambda_0 >0$ (see Section \ref{Sec:spsing}), corresponding to real resonances
$\pm \sqrt{\lambda_0}$.
We prove that the spectral shift function is regular away from these 
singularities. More precisely, outside the set of spectral singularities,
the derivative of the SSF is a smooth function of the energy.

\medskip\noindent
Near a spectral singularity $\lambda_0>0$ of finite order $\nu_0$, we obtain a
precise distributional asymptotic expansion. In the outgoing case (for instance), one has (see Section \ref{sec:Schro}),
\[
\xi'(\lambda;H,H_0)
=
\sum_{k=1}^{\nu_0}\frac{\alpha_k (\lambda_0)}{(\lambda-\lambda_0 + i0)^k}
+ H(\lambda) \quad {\rm in} \ \mathcal{D}'(I), 
\]
where $H(\lambda)$ is smooth near $\lambda_0$. Equivalently, the singular part
of $\xi'(\lambda;H,H_0)$ is a finite linear combination of principal value
distributions $\operatorname{p.v.}(\lambda-\lambda_0)^{-k}$ and derivatives of
the Dirac mass at $\lambda_0$.

\medskip\noindent
Finally,  in Section \ref{Sec:Toys}, we conclude the paper with the study of several simple and explicit examples illustrating the general results obtained above.

%	\subsection{Strategy of the proofs}
	
%	{\clr F : le papier est assez long. Fait-on sauter cette section car j ai deja tout rappelé  dans la section precedente ?}

	\section{Functional Calculus}\label{sec:functional calculus}
	
\subsection{Helffer-Sjöstrand formula}\label{subsec:HS}
    In what follows, we identify $\mathbb{R}^2$ with $\mathbb{C}$, and for $z\in\mathbb{C}$ we denote by $(x,y)$ its coordinates in $\mathbb{R}^2$. If $f:\mathbb{R}^2\to\mathbb{C}$ is differentiable, we set
\[
\partial_{\bar z} f := \tfrac12\bigl(\partial_x f + i\partial_y f\bigr).
\]

For any $\beta\in\R$, we define $\mathscr{S}^\beta$ the space of smooth functions $f:\R\to \C$ such that
	\begin{equation*}
	    f^{(r)}(x)=\mathcal{O}\left(\langle x \rangle^{\beta-r}\right),\quad |x|\to\infty,
	\end{equation*}
	for all $r\ge 0$  and where $\langle x \rangle = (1+x^2)^\frac12$. We define
\[
\mathscr{A}:=\bigcup_{\beta<0}\mathscr{S}^\beta.
\]
Then $\mathscr{A}$ is an algebra under pointwise multiplication, it contains the
smooth compactly supported functions from $\mathbb{R}$ to $\mathbb{C}$, and it is
stable under pointwise multiplication by functions in $\mathscr{S}^0$. On $\mathscr{A}$ we consider the family of norms 
	\begin{equation*}
	    \|f\|_n:=\sum_{k=0}^n\int_\R|f^{(k)}(x)|\langle x\rangle^{k-1}\mathrm{d}x.
	\end{equation*}
In particular, $D(\R)$ is dense in $\mathscr{A}$ for each norm $\|\cdot\|_n$.

In \cite{Da95_08}, Davies proves that for any closed operator $L$ acting on a Banach space $\mathcal{B}$ with \emph{real} spectrum satisfying an assumption similar to \eqref{eq:hypestimeresol}, the operator $f(L)$ acting on $\mathcal{B}$ is well defined by the Helffer--Sjöstrand formula for any $f\in\mathscr{A}$:
\begin{equation}\label{eq:HSformula}
    f(L)=\frac{1}{\pi}\int_\C \partial_{\bar z}\tilde{f}(z)\,(L-z)^{-1}\mathrm{d}x\mathrm{d}y,
\end{equation}
where $\tilde{f}$ is a  suitable  almost-analytic extension of $f$. Moreover, the map
\begin{align*}
    \Psi:\ &\mathscr{A}\longrightarrow \mathcal{B}(\mathscr{B})\\
          & f\mapsto f(L)
\end{align*}
is linear, continuous, and a morphism with respect to pointwise multiplication. By following the proof of \cite{Da95_08}, we cheek easily that for $f$ compactly supported in an open interval $I \subset \R$, $f(L)$ is still well defined if  the spectrum of $L$ is real only in $I \times \R$ and $f(L^*)=(\bar{f}(L))^*$.

Let us recall the definition of a such almost-analytic extension of a function $f\in\mathscr{C}^\infty(\R,\C)$ :

\begin{definition}[Almost-analytic extension]\label{defi:almostanalyticextension}
Let $f \in \mathscr{C}^\infty(\mathbb{R})$ and let $N \in \mathbb{N}$. An \emph{almost-analytic extension of order $N$} of $f$ is a function $\tilde f : \mathbb{C} \to \mathbb{C}$ of the form
\begin{equation}\label{defaaf}
    \tilde f(z)=\left( \sum_{k=0}^N \frac{f^{(k)}(x)}{k!}(iy)^k \right)\chi(x,y), \quad \text{with} \quad
\chi(x,y)=\tau\left(\frac{y}{\langle x\rangle}\right),
\end{equation}
%with 
%$$
%\chi(x,y)=\tau\left(\frac{y}{\langle x\rangle}\right),
%$$
and $\tau$ being a cut-off function defined on $\R$. The almost-analytic extension $\tilde f$ satisfies
$$
\partial_{\bar{z}} \tilde f(z) = \mathcal{O}(|\im z|^N)
\quad \text{as } \im z \to 0.
$$
\end{definition}

In this section, we aim to give a meaning to \eqref{eq:HSformula} by replacing $L$ with $H$, and $\mathscr{B}$ by $\Hi$, taking into account that $H$ may have non-real eigenvalues. The idea is to decompose $H$ as a direct sum of an operator with real spectrum, denoted $H_{\mathrm{r}}$, and another one with non-real spectrum, and then to define $f(H)$ in terms of $f(H_{\mathrm{r}})$. Since our assumptions are made locally in nature, we introduce the following notations : for any open interval $I\subset\R$, we set 
$$\mathscr{A}(I):=\lbrace f\in\mathscr{A}\,|\supp(f)\subset I\rbrace. $$
In particular, \begin{itemize}[label=$\bullet$]
    \item if $I=\R$, then $\mathscr{A}(I)=\mathscr{A}$;
    \item if $I$ is bounded, then $\mathscr{A}(I)=D(I)$; 
    \item for all $I\subset\R$ open interval, $D(I)$ is dense in $\mathscr{A}(I)$ for each norm $\|\cdot\|_n$. 
\end{itemize}

For the rest of the Subsection~\ref{subsec:HS}, we fix $I\subset\R$ to be an open interval where Hypotheses \ref{hyp:eigenvalues} and \ref{hyp:estimate resolvente} hold. 

\subsubsection{Decomposition}
In order to use the functional calculus developed in \cite{Da95_08}, let us decompose $H$ into a "real part" (i.e. an operator with real spectrum) and the "non-real part" (the restriction of $H$ to the eigenspaces associated with non-real eigenvalues). This decomposition is based on the following lemma.

\begin{lemma}\label{lem:decomposition_commuting_projection}
Let $P \in \mathcal{B}(\mathcal{H})$ be a bounded projection (i.e. $P^2=P$)
%{\clr with closed range.} {\clr F : a mon avis, l image est toujours fermee pour une projection continue.}
%\begin{equation*}
 %   \text{i.e.,} \quad P\in\mathcal{B}(\Hi),\quad P^2=P, 
%\end{equation*}
such that
\begin{equation}\label{eq:commutation}
    P(\mathcal{D}(H)) \subset \mathcal{D}(H) \quad \text{and} \quad [P,H]u = 0 \quad \forall u \in \mathcal{D}(H).
\end{equation}
Define
$$
F := \operatorname{Ran}(P), \qquad G := \operatorname{Ran}(\Id_{\mathcal{H}} - P).
$$
Then:
\begin{enumerate}
    \item $F$ and $G$ are invariant under $H$, i.e.,
    $$
        H(F \cap \mathcal{D}(H)) \subset F, \quad H(G \cap \mathcal{D}(H)) \subset G.
    $$
    \item there exists a continuous isomorphism 
    \begin{align*}
    \Psi :  F \oplus G &\to \mathcal{H} \\
     (u_F,u_G)&\mapsto u_F+u_G.
    \end{align*}
    \item Denoting by $H_{|F}$ and $H_{|G}$ the restrictions of $H$ to $F$ and $G$ respectively, then for all $u\in D(H)$,
    \begin{equation*}
        Hu=\Psi \left( \begin{array}{c|c}
   H_{|F} & 0 \\
   \hline
   0 & H_{|G} \\
\end{array}\right)\Psi^{-1}u=
   H_{|F}u_F+ H_{|G}u_G 
    \end{equation*}
    \item The spectrum satisfies
    $$
        \sigma(H) = \sigma(H_{|F})\cup\sigma(H_{|G})
    $$
    \item For all $z\in\rho(H)$, for all $u\in\Hi$, 
    \begin{equation*}
        \Res_H(z)u=\Psi \left( \begin{array}{c|c}
   \Res_{H_{|F}}(z) & 0 \\
   \hline
   0 & \Res_{H_{|G}}(z) \\
\end{array}\right)\Psi^{-1}u, \qquad \Res_\bullet(z):= (\bullet - z )^{-1}
    \end{equation*}
\end{enumerate}
\end{lemma}

\begin{proof}
 Since $P$ has closed range, we have the topological direct sum
$$
\mathcal{H} = \operatorname{Ran}(P) \oplus \operatorname{Ran}(\Id_{\mathcal{H}} - P) = F \oplus G.
$$
Let $u \in \mathcal{D}(H)$. By \eqref{eq:commutation}, $Pu \in \mathcal{D}(H)$ and $HPu = PHu$, which shows that $F$ is invariant under $H$. Similarly, $G$ is invariant under $H$. Thus, for any $u = u_F + u_G \in \mathcal{D}(H)$ with $u_F \in F$, $u_G \in G$, we have
$$
Hu = H u_F + H u_G = H_{|F} u_F + H_{|G} u_G.
$$
Finally, to relate the spectra, it suffices to observe that for any $z \in \mathbb{C}$, for any $u\in\mathcal{D}(H)$,
$$
\left(H- z \Id_\mathcal{H}\right)u = \Psi\left( \begin{array}{c|c}
   H_{|F}-z\Id_{|F} & 0 \\
   \hline
   0 & H_{|G}-z\Id_{|G} \\
\end{array}\right)\Psi^{-1}u.
$$
\end{proof}

In the following and for sake of conciseness, we will write $H=H_{|F}\oplus H_{|G}$ instead of using the matrix representation.  

\begin{prop}\label{prop:decomp}
Assume that Hypotheses \ref{hyp:eigenvalues} and \ref{hyp:estimate resolvente} hold on an open interval $I \subset \R$. Then there exist a bounded projection $\Pi_I\in\mathcal{B}(\Hi)$ %with closed range in $\Hi$ 
such that, 
\begin{equation*}
    \Pi_I(\mathcal{D}(H))\subset \mathcal{D}(H),\quad \text{and}\quad  [\Pi_I,H]u=0,\quad \forall u\in\mathcal{D}(H),
\end{equation*} 
and two operators, $H_{I,\mathrm{c}}$ acting on $\Ran(\Pi_I)$, and $H_{I,\mathrm{r}}$ acting on $\Ran(\Id_\Hi-\Pi_I)$ such that one has
\begin{enumerate}
    \item there exists an isomorphism $\Psi:\Ran(\Pi_I(H))\oplus\Ran(\Id_\Hi-\Pi_I(H))\rightarrow \Hi$, such that
    \begin{equation*}
        \Psi(u,v)=u+v
    \end{equation*}
    \item for all $u\in\mathcal{D}(H)$ 
    $$Hu=\Psi \left( \begin{array}{c|c}
   H_{I,\mathrm{c}} & 0 \\
   \hline
   0 & H_{I,\mathrm{r}} \\
\end{array}\right)\Psi^{-1}u$$
    \item $\sigma(H_{I,\mathrm{r}})\cap I\times \R \subset I\times \{0\}$ and $\sigma(H_{I,\mathrm{c}})\cap I \times \{0\}=\emptyset$,
    \item there exists $a_I>0$ such that 
    \begin{equation*}
        \overline{S_{a_I}}\ni z\mapsto \Res_{H_{I,\mathrm{c}}}(z) \in\mathcal{B}(\Ran(\Pi_I))
    \end{equation*}
    is holomorphic, and the following resolvent estimates hold : 
    \begin{equation}\label{eq:resolrestR}
        \forall z\in S_{a_I}(I),\quad \|\Res_{H_{I,\mathrm{r}}}(z)\|_{\mathcal{B}(\Ran(\Id_\Hi-\Pi_I))}\leq c_I|\im(z)|^{-1}\left( \frac{\langle \re(z)\rangle }{|\im(z)|}\right)^{n_I} ; 
    \end{equation}
    \begin{equation}\label{eq:resolrestC}
        \sup_{z\in\overline{S_{a_I}}} \|\Res_{H_{I,\mathrm{c}}}(z)\|_{\mathcal{B}(\Ran(\Pi_I)} <\infty.
    \end{equation}
\end{enumerate}
\end{prop}

\begin{proof}
By Hypothesis~\ref{hyp:eigenvalues}, $H$ has only finitely many non-real eigenvalues whose real parts lie in $I$. Therefore, by setting $\Omega=\lbrace \lambda\in\sigma_\mathrm{p}(H)\,|\, \re(\lambda)\in I\rbrace$, the operator
\begin{equation*}
    \Pi_I(H) := \sum_{\lambda \in \Omega} \Pi_\lambda(H),
\end{equation*}
where $\Pi_\lambda$ is the Riesz projection associated to $\lambda$ (see \eqref{eq:RieszProjection}), is a finite sum of finite-rank projections. Moreover as $\Pi_\lambda\Pi_\mu=0$ for $\lambda\neq \mu$, $\Pi_I$ is a projection and has finite rank. %Hence its range is closed in $\mathcal{H}$ with respect to the norm topology. 
Setting
$$F := \Ran(\Pi_I(H)), \quad G := \Ran(\Id - \Pi_I(H)),$$
and using the notations of Lemma \ref{lem:decomposition_commuting_projection}, we denote
$$H_{I,\mathrm{c}} := H_{|F}, \quad H_{I,\mathrm{r}} := H_{|G} ;$$
this provides the decomposition of $H$. Moreover, by construction,
$$
\sigma(H_{I,\mathrm{c}}) \cap I\times \{0\} = \emptyset, \quad \sigma(H_{I,\mathrm{r}}) \cap I\times \R \subset I\times \{0\}.
$$
By Hypothesis~\ref{hyp:eigenvalues}, we may find $a_I>0$ small enough such that $S_{a_I}(I)$ does not contain any non-real eigenvalues of $H$. In particular, $z\mapsto\Res_{H_{I,\mathrm{c}}}(z)$ 
is holomorphic on $\overline{S_{a_I}}$ and this provides \eqref{eq:resolrestC} since there are only a finite number of complex eigenvalues. Finally to obtain \eqref{eq:resolrestR}, we write for all $z\in S_{a_I}$,
\begin{align*}
    \|\Res_{H_{I,\mathrm{r}}}\|_{\mathcal{B}(\Id-\Ran(\Pi_I))}&=\sup_{u\in\Ran(\Id-\Pi_I)\backslash\{0\}} \frac{\|\Res_{H_{I,\mathrm{r}}}u\|_{\Ran(\Id-\Pi_I)}}{\|u\|_{\Ran(\Id-\Pi_I)}}=\sup_{u\in\Ran(\Id-\Pi_I)\backslash\{0\}} \frac{\|\Res_{H}u\|_{\Hi}}{\|u\|_{\Hi}}\\
    &\leq \|\Res_{H}\|_{\mathcal{B}(\Ran(\Id_\Hi-\Pi_I))}\leq c_I|\im(z)|^{-1}\left( \frac{\langle \re(z)\rangle }{|\im(z)|}\right)^{n_I},
\end{align*}
where in the last step, we use Hypothesis~\ref{hyp:estimate resolvente}.
\end{proof}

\subsubsection{Definition and properties}
Let us now fix an open interval $I \subset \R$, and define $f(H)$ for $f$ supported in $I$ under the Hypotheses ~\ref{hyp:eigenvalues} and ~\ref{hyp:estimate resolvente} on $I$.

\begin{definition}\label{defi:almostanalyticextensionadmissible}
    Let $f\in\mathscr{A}(I)$ and $N\in\N$. We say that $\tilde{f}:\C\to\C$ is a almost-analytic extension of order $N$ which is \emph{admissible} for $H$ (or a $H$-admissible almost-analytic extension) if there exists a cut-off function $\tau$ %{\clr F : ca ne serait pas $\chi$ ?} 
    such that the support of $\chi$, considered in Definition~\ref{defi:almostanalyticextension}, does not contain any non-real eigenvalues of $H$ whose real part lies in $I$. 
\end{definition}

\begin{exem}\label{ex:almost-analytic-extension}
    As an example, suppose that Hypothesis \ref{hyp:eigenvalues} hold on $\R$. Then there exists $\mathfrak{a}>0$ and $\mathfrak{b}>0$ such that
    \begin{equation*}
        U_{\mathfrak{a},\mathfrak{b}}:=\lbrace z\in \C \, | \,|\re(z)|>\mathfrak{a}\text{ and } |\im(z)|<\mathfrak{b}\rbrace 
    \end{equation*}
    does not contain any non-real eigenvalues of $H$ %{\clb N : "whose real part lies in $I$"}. {\clr F : il faut sans doute supposer $I= \mathbb{R}$ N : Normalement, avec l'ajout que j'ai fais dans la définition précédente, la fonction $\psi$ fonctionne même pour $I$ quelconque.} Let us define the function
    \begin{align*}
        \psi:\R^2&\to \R\\
        (x,y)&\mapsto \frac{y}{\mathfrak{b}}\frac{\langle \mathfrak{a}\rangle}{\langle x\rangle},
    \end{align*}
    and let $\tau\in\mathscr{C}^\infty(\R,\C)$ a cutoff function which satisfies $\tau(s)=1$ for $|s|<1/2$ and $\tau(s)=0$ for $|s|>1$. Then, the support of the function $\chi:=\tau\circ \psi$, does non contain any non-real eigenvalues of $H$. %{\clb N : "whose real part lies in $I$"}. 
\end{exem}

\begin{defiprop}\label{definprop:HJ}
    Suppose that Hypotheses \ref{hyp:eigenvalues} and \ref{hyp:estimate resolvente} hold on $I$. Then for all $f\in\mathcal{A}(I)$, the operator $f(H)$ defined by 
    \begin{equation*}
        f(H):=\frac{1}{\pi}\int_\C \partial_{\bar{z}}\tilde{f}(z)\Res_{H}(z)\mathrm{d}x\mathrm{d}y,
    \end{equation*}
    where $\tilde{f}$ is an almost-analytic extension of $f$ as in Definition~\ref{defi:almostanalyticextension}, is well-defined. Moreover, we have the estimate
    \begin{equation}\label{eq:estimatefunctional}
        \|f(H)\|_{\mathcal{B}(\Hi)}\leq c\|f\|_{n_I+1}.
    \end{equation}
    Finally the definition of $f(H)$ does not depend on $N$, or on the cut-off function $\tau$, provided to $N>n_I$ in Definition~\ref{defi:almostanalyticextension}. 
    Moreover, it is independent of the choice of the interval $I$
as long as $I$ contains the support of $f$ and
Hypotheses~\ref{hyp:eigenvalues} and \ref{hyp:estimate resolvente}
hold on $I$.
\end{defiprop}

To prove the last proposition, the idea is to work on
\[
\operatorname{Ran}(\Id-\Pi_I)\oplus \operatorname{Ran}(\Pi_I)
\]
instead of $\mathcal H$, and to analyse the action of $H$ on each of these spaces.
For any interval $I\subset\R$ such that Hypotheses~\ref{hyp:eigenvalues}
and~\ref{hyp:estimate resolvente} hold, the operator $H_{I,\mathrm r}$ defined on
$\operatorname{Ran}(\Id_{\mathcal H}-\Pi_I)$ satisfies the assumptions required by
Davies~\cite{Da95_08} to define $f(H_{I,\mathrm r})$ via the Helffer--Sj\"ostrand
formula~\eqref{eq:HSformula}, for any $f\in\mathcal A(I)$.
It therefore remains to prove that $f(H_{I,\mathrm c})$ is well defined on
$\operatorname{Ran}(\Pi_I)$.

\medskip
This is the purpose of the following lemma.

\begin{lemma}\label{lm:calculfunctional-complex}
    Suppose that Hypotheses~ \ref{hyp:eigenvalues} and \ref{hyp:estimate resolvente} hold on $I$. Let $F=\Ran(\Pi_I)$. Then for all $f\in\mathcal{A}(I)$ and $\tilde{f}$ an $H$-admissible almost analytic extension, the integral
    \begin{equation}\label{eq:integrand}
        \int_\C\partial_{\bar{z}}\tilde{f}(z)\Res_{H_{I,\mathrm{c}}}(z)\mathrm{d}x\mathrm{d}y,
    \end{equation}
 is well defined on $F$ and vanishes. Moreover, this does not depend on the considered $H$-admissible almost analytic extension. 
\end{lemma}

\begin{proof}[Proof of Lemma~\ref{lm:calculfunctional-complex}]
Let $f\in\mathcal{A}(I)$. Let $\tilde{f}$ be an $H$-admissible almost analytic extension given by \eqref{defaaf}.
%with $\chi$ as in Example~\ref{ex:almost-analytic-extension}. 
Then for all integer $N\ge 1$, a direct computation gives that \begin{equation}\label{eq:PPA}
        \partial_{\bar{z}}\tilde{f}(z)=\frac{1}{2}f^{(N+1)}(x)\frac{(iy)^N}{N!}\chi(x,y)+\frac{1}{2}\left[\sum_{k=0}^Nf^{(k)}(x)\frac{(iy)^k}{k!}\right]\partial_{\bar{z}}\chi(x,y).
    \end{equation}
    %\footnote{N: Le début de la preuve semble se passer comme dans le cas scalaire. Il y a peut-être juste une citation à faire.} 
    Next a direct computation gives that 
    \begin{equation*}
        \partial_{\bar{z}}\chi(x,y)=\frac{1}{2}\left(-\frac{y\langle\mathfrak{a}\rangle}{\mathfrak{b}}\frac{x}{\langle x\rangle^3}+i\frac{\langle \mathfrak{a}\rangle }{\mathfrak{b}\langle x\rangle }\right)\tau'\left(\chi(x,y)\right). 
    \end{equation*}
    In particular if we denote 
    \begin{equation*}
        V_1=\left\lbrace (x,y)\in\R^2\,\bigg|\,0<\frac{|y|}{\mathfrak{b}}<\frac{\langle x\rangle}{\langle \mathfrak{a}\rangle}\right\rbrace,\quad V_2=\left\lbrace (x,y)\in\R^2\,\bigg|\, \frac{\langle x\rangle}{2\langle \mathfrak{a}\rangle} \le \frac{|y|}{\mathfrak{b}} \le  \frac{\langle x\rangle}{\langle \mathfrak{a}\rangle}\right\rbrace,
    \end{equation*}
    then $\supp(\chi)\subset V_1$ and $\supp(\partial_{\bar{z}}\chi)\subset V_2$. 
    Then for all $z\in U_{\mathfrak{a},\mathfrak{b}}$, using \eqref{eq:PPA}, we have
    \begin{align}
        \left\|\partial_{\bar{z}}\tilde{f}(z)\Res_H(z)\right\|_{\mathcal{B}(\Ran(\Pi_I))}&\leq C\left(\left|f^{(N+1)}(x)\right||y|^{N}\mathds{1}_{V_1}(z)+\sum_{k=0}^{N}\left|f^{(k)}(x)\right||y|^{k}\mathds{1}_{V_2}(z)\right) \nonumber \\
        &\leq C\left(\left|f^{(N+1)}(x)\right||y|^{N}\mathds{1}_{V_1}(z)+\sum_{k=0}^{N}\left|f^{(k)}(x)\right|\langle x\rangle ^{k-2}\mathds{1}_{V_2}(z)\right) \label{estimationcf}.
    \end{align}
    Finally, integration with respect to $y$ gives 
    \begin{align*}
        \|f(H)\|_{\mathcal{B}(\Ran(\Pi_{I}))}&\leq C\int_\C\|\partial_{\bar{z}}\tilde{f}(z)\Res_H(z)\|_{\mathcal{B}(\Ran(\Pi_{I}))}\mathrm{d}x \mathrm{d}y\\
        &\le  C\left(\int_\R\left|f^{(N+1)}(x)\right|\langle x\rangle^{N}\mathrm{d}x+\sum_{k=0}^{N}\int_\R\left|f^{(k)}(x)\right|\langle x\rangle ^{k-1}\mathrm{d}x\right)\leq C\|f\|_{N+1}.
    \end{align*}
    Next to prove that \eqref{eq:integrand} vanishes, we proceed by density. Let $f\in\mathscr{C}_\mathrm{c}^\infty(I,\C)$ supported in an open set $\Omega$. Then $\tilde{f}\in\mathscr{C}^\infty(\R^2,\C)$ and the Stokes's formula gives
    \begin{equation*}
        \int_\C\Res_{H_{I,\mathrm{c}}}(z)\partial_{\bar{z}}\tilde{f}(z)\mathrm{d}x\mathrm{d}y=\frac{1}{2i}\int_{\partial\Omega}\Res_{H_{I,\mathrm{c}}}(z)\tilde{f}(z)\mathrm{d}z,
    \end{equation*}
    where $\partial\Omega$ is the boundary of $\Omega$ and the above contour integral is taken counterclockwise. So as $\tilde{f}$ is compactly supported,
    \begin{equation*}
        \int_{\partial\Omega}\Res_{H_{I,\mathrm{c}}}(z)\tilde{f}(z)\mathrm{d}z=0. 
    \end{equation*}
    Using the same argument, we can show as in \cite[Proof of Theorem 2]{Da95_08}, that the value of \eqref{eq:integrand} is independent of the $H$-admissible almost analytic extension considered. 
    Finally Lemma \ref{lm:calculfunctional-complex} follows by using that $D(I)$ is dense in $\mathcal{A}(I)$ %for each of the norm $\|.\|_n$ 
    and the inequality
    \begin{equation*}
      \forall f\in\mathcal{A}(I), \quad  \|f(H)\|_{\mathcal{B}(\Ran(\Pi_{I}))}\le C\|f\|_{N+1}.
    \end{equation*}
\end{proof}

\begin{proof}[Proof of Proposition~\ref{definprop:HJ}]
    Let $f\in\mathcal{A}(I)$. By Proposition~\ref{prop:decomp} and Lemma~\ref{lem:decomposition_commuting_projection}, there exists an isomorphism
    $$\Psi:\Ran(\Pi_I(H))\oplus\Ran(\Id-\Pi_I(H))\to \Hi$$
    such that for all $z\in\rho(H)$,
    \begin{equation*}
        \Res_H(z)=\Psi\left( 
            \begin{array}{c|c}
                \Res_{H_{I,\mathrm{r}}}(z) & 0 \\
                \hline
                0 & \Res_{H_{I,\mathrm{c}}}(z) \\
            \end{array}
        \right)\Psi^{-1}.
    \end{equation*}
    So using Davies's work \cite{Da95_08}, $f(H_{I,\mathrm{r}})$ is well defined as an operator on $\Ran(\Id-\Pi_I)$ and its definition does not depend on the almost-analytic extension considered in \eqref{eq:HSformula} as soon as its order is larger than $n_I$ in Hypothesis~\ref{hyp:estimate resolvente}. Moreover, by Lemma~\ref{lm:calculfunctional-complex} $f(H_{I,\mathrm{c}})$ is well-defined independently of the $H$-admissible almost-analytic extension of $f$. This also gives independence in the choice of the interval $I$ containing the support of $f$ (where Hypotheses~ \ref{hyp:eigenvalues} and \ref{hyp:estimate resolvente} are fulfilled) because for $I \subset J$, we have 
    $$H_{J,\mathrm{c}}= H_{I,\mathrm{c}} \oplus H_{J\setminus I,\mathrm{c}}; \qquad H_{I,\mathrm{r}} = H_{J,\mathrm{r}} \oplus H_{J\setminus I,\mathrm{c}}$$
    and then $f(H_{I,\mathrm{r}})= f(H_{J,\mathrm{r}}) \oplus 0$. 
\end{proof}

In particular in the following, we may write 

\begin{equation}\label{eq:representation functional calculus}
    f(H)=\Psi\left(\begin{array}{c|c}
            f(H_{I,\mathrm{r}}) & 0 \\
            \hline
            0 & 0 \\
        \end{array}
    \right)\Psi^{-1}.
\end{equation}

\noindent
This representation of $f(H)$ and \cite[Theorem 3 and 4]{Da95_08} implies immediately the following Proposition:

\begin{prop}\label{prop:propertiesfunctionalcalculus}
   Let $I\subset \R$ be an open interval and suppose that Hypotheses~ \ref{hyp:eigenvalues} and \ref{hyp:estimate resolvente} hold on $I$.  Then the map 
    \begin{equation*}
        \mathscr{A}(I) \ni f\mapsto f(H)\in\mathcal{B}(\Hi)
    \end{equation*}
    is a morphism of algebra for the pointwise multiplication. Moreover if $f\in\mathcal{C}_\mathrm{c}^\infty(I,\C)$ has disjoint support from $\sigma(H)$, then $f(H)=0$. 
\end{prop}

\subsubsection{Comparison with Frantz-Faupin calculus}
	
In \cite[Proposition 5.2]{FaFr23}, Frantz and Faupin introduce a functional calculus under similar assumptions.  Their construction is inspired by Stone's formula:
\begin{equation}\label{eq:functional representation}
	f(H) = \frac{1}{2\pi i} \lim_{\varepsilon \to 0^+} \int_{\supp(f)} f(\lambda) \left( \Res_H(\lambda + i\varepsilon) - \Res_H(\lambda - i\varepsilon) \right) \, \mathrm{d}\lambda.
\end{equation}
They prove that the formula \eqref{eq:functional representation} makes sense in
the weak topology for functions $f$ in $\mathscr{C}_{\mathrm{b}}(I,\mathbb{C})$,
where $I$ is an interval of the essential spectrum without spectral
singularities (see Section 6 for the definition)  and where
$\mathscr{C}_{\mathrm{b}}(I,\mathbb{C})$ denotes the space of bounded continuous
functions on $I$. They develop also a functional calculus which takes into account spectral singularities (see \cite[Proposition 5.3]{FaFr23}). These approaches require a limiting absorption principle for $H_0$, and the
proofs are based on the interchange of a limit and an integral. In the following
proposition, we show that one can give a meaning to the formula
\eqref{eq:functional representation} by means of the Helffer-Sjöstrand formula
constructed above at least for compactly supported functions.
	
\begin{prop}
    Let $I\subset \R$ be an open interval and suppose that Hypotheses \ref{hyp:eigenvalues} and \ref{hyp:estimate resolvente} hold on $I$. Let $f\in\mathscr{C}_\mathrm{c}^\infty(I,\C)$. Then 
	\begin{equation*}
	   f(H)=\frac{1}{2\pi i} \lim_{\varepsilon \to 0^+} \int_{\supp(f)} f(\lambda) \left( \Res_H(\lambda + i\varepsilon) - \Res_H(\lambda - i\varepsilon) \right) \, \mathrm{d}\lambda.
    \end{equation*}
\end{prop}

\begin{proof}
    We claim that 
    \begin{equation}\label{eq:cfFF0}
	   f(H)=\pi^{-1}\lim_{\varepsilon\rightarrow 0^+}\left(\int_{\im(z)>0} \partial_{\bar{z}} \tilde{f}(z)\Res_H(z+i\varepsilon)\mathrm{d}z+\int_{\im(z)<0} \partial_{\bar{z}} \tilde{f}(z)\Res_H(z-i\varepsilon)\mathrm{d}z\right).
	 \end{equation}
	 Indeed, if we denote $J:=\supp(f)\subset I$, there exists $a_I>0$ sufficiently small such that  $\supp(\tilde{f})\subset S_{a_I}(J)$ and $S_{a_I}(J)$ does not contain any non real eigenvalue of $H$. So, under Hypotheses \ref{hyp:eigenvalues} and \ref{hyp:estimate resolvente}, for $\varepsilon>0$ sufficiently small the map $z\to \partial_{\bar{z}}\tilde{f}(z) \Res_H(z\pm i\varepsilon)$ is continuous on  $\supp(\tilde{f})\cap \C^\pm$ and for all $z\in S_{a_I}(J)$, we have
	    \begin{equation*}
	      \lim_{\varepsilon\rightarrow 0^+} \Res_H(z\pm i\varepsilon)=\Res_H(z), \qquad 
	        \|\partial_{\bar{z}}\tilde{f}(z) \Res_H(z\pm i\varepsilon)\|_{\mathcal{B}(\Hi)}\leq c_0
	    \end{equation*}
	    for some constant $c_0>0$, uniformly with respect to $\varepsilon$  and $z$.
	Consequently, Lebesgue's theorem implies \eqref{eq:cfFF0}.
	Now since $z\mapsto \Res_H(z\pm i\varepsilon)$ are holomorphic on $S_{a_I}(J)$, we may apply Stokes's theorem to obtain 
	\begin{equation}\label{eq:cfFF2}
	    \int_{\pm\im(z)>0} \partial_{\bar{z}}\tilde{f}(z) \Res_H(z\pm i\varepsilon)\mathrm{d}x\mathrm{d}y=\pm \frac{1}{2i}\int_{\Gamma^\pm} \tilde{f}(z)\Res_H(z \pm i\varepsilon)\mathrm{d}z,
	\end{equation}
	where $\Gamma^\pm$ denotes the boundary of $S_{a_I}(J)$ in $\overline{\C^\pm}$. Finally the support properties of $\tilde{f}$ implies that \begin{equation}\label{eq:cfFF3}
	    \int_{\Gamma^\pm} \tilde{f}(z)\Res_H(z \pm i\varepsilon)\mathrm{d}z=\int_{J}f(\lambda)\Res_H(\lambda\pm i\varepsilon)\mathrm{d}\lambda.
	\end{equation}
	So the conclusion follows by substituting first \eqref{eq:cfFF2} and then \eqref{eq:cfFF3} in \eqref{eq:cfFF0}. 
\end{proof}
	
To conclude this introduction to the functional calculus, observe that if
$\omega \in \rho(H)$ and $r_\omega(x)=(x-\omega)^{-1}$, then $r_\omega (H)$ is precisely the
resolvent $(H-\omega)^{-1}$. More precisely,
%For the next Proposition, we suppose that $I=\R$. If $H$ has a finite number of no-real eigenvalues, then
%the projection onto the complex eigenvalues of $H$ is well defined. The next Proposition establishes that the the resolvent map coincides may be written with the help of  the Helffer-Sjöstrand formula up to a spectral projection.
under the Hypothesis \ref{hyp:eigenvalues} on $I=\R$, $H$ has a finite number of no-real eigenvalues and the projection onto the non-real eigenvalues of $H$ is well defined by:
\begin{equation*}  
    \Pi_{\mathrm{complex}}(H):=\sum_{\lambda\in\sigma_\mathrm{disc}(H)\backslash \R}\Pi_\lambda(H).
\end{equation*}
Then we have:

\begin{prop}\label{prop:comparison resolvant}
    Suppose that Hypotheses \ref{hyp:eigenvalues} and \ref{hyp:estimate resolvente} hold on $\R$. Let $\omega\in\rho(H)$ and $r_\omega:\R\to\C$ be defined for all $x\in\R$ by 
    $$r_\omega(x)=(x-\omega)^{-1}.$$
    Then
     \begin{equation}\label{eq:decomposition}
        r_\omega(H)=\Res_H(\omega)(\Id-\Pi_\mathrm{complex}(H)), \qquad 
        \Res_H(\omega)=r_\omega(H) + \sum_{\lambda\in\sigma_\mathrm{disc}(H)\backslash \R}(\lambda-\omega)^{-1}\Pi_\lambda(H).
    \end{equation}
\end{prop}

\begin{proof}
Applying Proposition \ref{prop:decomp} for $I=\R$, we introduce the real part of $H$: $H_r$ the restriction of $H$ to the range of $\Id-\Pi_\mathrm{complex}(H)$.
    Using \eqref{eq:representation functional calculus} and \cite[Theorem 5]{Da95_08}, for all $u\in\Hi$, we have :
    \begin{align*}
        r_\omega(H)u&=\Psi\left(
        \begin{array}{c|c}
            (H_{\mathrm{r}}-\omega)^{-1} & 0 \\
            \hline
            0 & 0 \\
        \end{array}
        \right)\Psi^{-1}u\\
         &=(H_{\mathrm{r}}-\omega)^{-1}(\Id-\Pi_{\mathrm{complex}}(H))u\\
         &=(H-\omega)^{-1}(\Id-\Pi_{\mathrm{complex}}(H))u.
    \end{align*}
    Then we deduce \eqref{eq:decomposition}.
    %For the resolvent of $H$, we have: $$\Res_H(\omega)=r_\omega(H) + \sum_{\lambda\in\sigma_\mathrm{disc}(H)\backslash \R}(\lambda-\omega)^{-1}\Pi_\lambda(H). $$
%    {\clr F : ce dernier resultat n apparait pas dans l enonce de la proposition. Le met-on ? N : Je pense que c'est une bonne idée. ou alors en remarque ? }
    
\end{proof}

\vspace{0.1cm}\noindent
		
%	\begin{remark} {\clb V: Voir si on reprend certaines remarques mises en commentaire ...}

%	\end{remark}
\begin{comment}
	 \begin{itemize}
	     \item On doit pouvoir obtenir la Proposition 3.9 pour des fonctions $f$ dans $\mathscr{A}$ en utilisant un argument de déformation de contour, mais ce n'est pas l'objet de ce papier donc on ne le fait pas. {\clr F : je suis d'accord.}
	     \item Contrairement au calcul fonctionnel de FF, c'est uniquement le comportement de la fonction à l'infini qui permet de gérer le mauvais comportement de la résolvante près de l'axe réel. 
	     \item Dans FF, il fallait multiplier la fonction par une fraction rationnelle qui régularise les singularités spectrales et la limite dans \eqref{} était au sens faible. Ici elle est au sens opérateur. 
	     \item L'avantage de FF est qu'on ne demande pas de régularité de régularité à la fonction à l'infini. Cela permet notamment d'obtenir un calcul fonctionnel pour une semi-groupe fortement continu à un paramètre, ce que HJ ne permet pas de faire car $\mathscr{S}^0(\R,\C)$ n'est pas dans $\mathscr{A}$. 
	     \item Pour cette construction on n'a jamais utilisé qu'on avait besoin d'un nombre fini de valeurs propres réelles. {\clr F : je suis perdu. Dans ce papier, on a un nombre fini de valeurs propres réelles ? ou bien c'est dans FF que vous prenez cette hypothèse ? N : Ici le calcul fonctionnel qu'on construit englobe aussi les valeurs propres réelles discrètes. dans FF, il est construit seulement pour le spectre essentiel. }
	 \end{itemize}
\end{comment}

	\subsection{Spectral changing of variables}\label{fog(H)}

	In this section, for any $-c \in \R \cap \rho(H)$ with $c>0$, and $m \in \N^*$, we consider a function $f \in D((0, +\infty))$ and define $g_m \in D((-c, +\infty))$ by the relation
	\[
	g_m(\lambda) := f((\lambda + c)^{-m}).
	\]
	Our goal is to identify sufficient conditions on $H$ under which the following composition formula holds:
	\[
	f((H + c)^{-m}) = g_m(H).
	\]
	%First we give a general result for an operator $L$ (later it will be applied to $L=H+c$). 
	Initially, instead of $H+c$, we consider a general operator $L$.
	%Without loss of generality, we assume $c = 0$ (this simply amounts to replacing $H + c$ by $H$ in the analysis). 
	For intervals $I \subset \R$, $R \subset (0, +\infty)$, $\Theta_0 \subset [-\pi/2, \pi/2]$, and $a > 0$, we introduce the following subset of the complex plane:
	\begin{equation}\label{DefSets}
		%S_a(I) := \left\{ z \in \C \,\middle|\, \re(z) \in I, \ 0 < |\im(z)| < a \right\}, \quad
		R \cdot e^{i\Theta_0} := \left\{ \rho e^{i\theta} \in \C \,\middle|\, \rho \in R, \ \theta \in \Theta_0 \right\}.
	\end{equation}
	
	\vspace{0.1cm}\noindent
	To establish relations between the resolvents of $L$ and those of $L^{-m}$, we state the following lemma concerning geometric sums of bounded operators.
	
	\begin{lemma}\label{lemT}
		Let $m \in \N^*$ and let $T\in\mathcal{B}(\Hi)$ be a bounded operator such that the geometric sum
		\begin{equation}\label{defsum}
			G_m(T) := \sum_{k=0}^{m-1} T^k
		\end{equation}
		is invertible. Then $(T - \Id)$ is invertible if and only if $(T^m - \Id)$ is invertible, and in this case the following identities hold:
		\begin{align}
		(T - \Id)^{-1} &= G_m(T) (T^m - \Id)^{-1} \label{eq:LemT Identity 1}\\
		             &=m (T^m - \Id)^{-1} + \sum_{k=1}^{m-1} G_k(T) G_m(T)^{-1}.\label{eq:LemT Identity 2}
	    \end{align}
	\end{lemma}
	
	\begin{proof}
		The first identity follows from the factorization $(T^m - \Id) = (T - \Id) G_m(T)$. For the second identity, we write for all $k\in\N^\star$, $T^k = \Id + (T^k - \Id) = I + G_k(T)(T - \Id)$ and use the fact that $(T - \Id)(T^m - \Id)^{-1} = G_m(T)^{-1}$.
	\end{proof}

We define ${\C_>}$ the set of complex number with positive real part : 
	\begin{equation*}
	    {\C_>}:=\lbrace z\in\C\,|\, \re(z)>0\rbrace. 
	\end{equation*}
	By applying the above lemma with $T = z L^{-1}$, we deduce:

	\begin{lemma}\label{lemRes}
	Let $m \in \N^*$ and let $\left(L, \mathcal{D}(L)\right)$ be a closed operator on a complex separable Hilbert space $\Hi$, such that
		\begin{equation}\label{HSect}
			\sigma(L) \subset (0, +\infty) \cdot e^{i[-\theta_0, \theta_0]} \quad \text{for some } \theta_0 \in \left(0, \frac{\pi}{2m} \right).
		\end{equation}
		Then,
		\begin{enumerate}
		    \item\label{item1} $z \in \rho(L)\backslash\{0\}$ if and only if $z^{-m} \in \rho(L^{-m})$,
		    \item the operator-valued map
		\begin{equation}\label{am}
			\Psi :z \longmapsto z^{-m+1} G_m(z L^{-1}) = \sum_{k=0}^{m-1} z^{k+1-m} L^{-k}
		\end{equation}
		is biholomorphic in the sector $\, (0, +\infty) \cdot e^{i[-\theta_0, \theta_0]}$,
		\item for every $z \in \rho(L) \cap (0, +\infty) \cdot e^{i[-\theta_0, \theta_0]}$, the following resolvent identities hold:
		\begin{align}
			(L - z)^{-1} &= -z^{-m} L^{-1} G_m(z L^{-1}) \left( L^{-m} - z^{-m} \right)^{-1}, \label{eqRes1} \\
			&= -m z^{-1 - m} \left( L^{-m} - z^{-m} \right)^{-1}
			- z^{-1} \sum_{k=1}^{m-1} G_k(z L^{-1}) G_m(z L^{-1})^{-1}. \label{eqRes2}
		\end{align}
		\end{enumerate}
	\end{lemma}

    \begin{proof}	    
		        Since \( 0 \in \rho(L) \), it follows from holomorphic functional calculus of closed operator (\cite{DuSc58}, Theorem VII.9.8.4) that for any \( k \in \N \), the operator \( L^{-k} \) is bounded and its spectrum is contained in the sector
				$$
				\sigma(L^{-k})=\lbrace z^{-k}\,|\, z\in\sigma(L)\rbrace\cup \lbrace 0\rbrace \subset [0, +\infty) \cdot e^{i[-k\theta_0,\, k\theta_0]}.
				$$
				This prove item \eqref{item1}. 
				On the other  hand, for any $k\in\llbracket 0,m-1\rrbracket$, for any \( z \in (0, +\infty) \cdot e^{i[-\theta_0,\, \theta_0]} \), we have
				$$
				z^{k+1 - m} \in (0, +\infty) \cdot e^{i[-(k+1-m)\theta_0,\, (k+1-m)\theta_0]}.
				$$
				Therefore, the spectrum of the operator \( z^{k+1 - m} L^{-k} \) is contained in the sector
				$$
				[0, +\infty) \cdot e^{i[-(m - 1)\theta_0,\, (m - 1)\theta_0]} \subset [0, +\infty) \cdot e^{i(-\tfrac{\pi}{2},\, \tfrac{\pi}{2})}.
				$$
				As shown in (\cite{KaRi83}, Proposition 3.2.10), if two bounded linear operators \( A \) and \( B \) commute, then the spectrum of their sum satisfies the inclusion \( \sigma(A + B) \subset \sigma(A) + \sigma(B) \). So, it follows that the spectrum of
				$$
				z^{-m+1} \big( G_m(zL^{-1}) - \Id \big) = \sum_{k=1}^{m-1} z^{k+1 - m} L^{-k}
				$$
				is contained in $\overline{{\C_>}}$.
				So using one against Proposition 3.2.10, we deduce that, for all $z \in (0, +\infty) \cdot e^{i[-\theta_0, \theta_0]}$, the spectrum of the operator
				\[
				z^{-m+1} G_m(zL^{-1})
				\]
				is contained in the open right half-plane ${\C_>}$. In particular, the operator $G_m(zL^{-1})$ is invertible for all $z\in{\C_>}$. Applying \eqref{eq:LemT Identity 1} with $T = zL^{-1}$, we obtain for every $z \in \rho(L) \cap (0, +\infty)\cdot e^{i[-\theta_0, \theta_0]}$:
				\begin{align*}
					(L - z)^{-1} 
					&= L^{-1} (\Id - zL^{-1})^{-1} \\
					&= - L^{-1} G_m(zL^{-1}) \left( (zL^{-1})^m - \Id \right)^{-1} \\
					&= - L^{-1} z^{-m} G_m(zL^{-1}) \left( L^{-m} - z^{-m} \right)^{-1},
				\end{align*}
				which yields identity~\eqref{eqRes1}.  Identity~\eqref{eqRes2} then follows from the second formula in Lemma~\ref{lemT} as follows:
				\begin{align*}
					(L - z)^{-1} 
					&= L^{-1} (\Id - zL^{-1})^{-1} \\
					&= - z^{-1}\Id - z^{-1} (zL^{-1} - \Id)^{-1} \\
					&= - z^{-1}\Id - m z^{-1 - m} \left( L^{-m} - z^{-m} \right)^{-1} 
					- z^{-1} \sum_{k=1}^{m-1} G_k(zL^{-1}) G_m(zL^{-1})^{-1}.
				\end{align*}	
				The biholomorphy of $\Psi$ can be proved by using Neumann series.
				\end{proof}

%\begin{lemma}\label{lem:sectorial}
   %     Suppose that Hypothesis~\ref{hyp:eigenvalues} hold. Then there exists $-c>0$ and $\theta_0\in (0,\pi/2)$ such that
   %     \begin{equation}\label{eq:sec}
  %          \sigma(H+c)\subset (0,+\infty).e^{i[-\theta_0,\theta_0]}.
  %      \end{equation}
 %   \end{lemma}
			
%	\begin{proof}
%	    As $\sigma_\mathrm{ess}(H)=\sigma_\mathrm{ess}(H_0)$ (see \eqref{eq:ess spec}) and $H$ has a finite number of eigenvalues by assumption, there exists $-c>0$ such that 
%	    \begin{equation*}
%	        \inf\lbrace \re(z)\,|\, z\in\sigma(H)\rbrace >c.
%	    \end{equation*}
%	    Thus we deduce that $\sigma(H+c)\subset {\C_>}$. On the other hand using one again the finiteness of the number of eigenvalues of $H$, we have
%	    \begin{equation*}
%	        \sup\lbrace \im(z)\,|\,z\in\sigma(H)\rbrace
%	    \end{equation*}
%	    is finite. So there exists $\theta\in (0,\frac{\pi}{2})$ such that \eqref{eq:sec} hold. 
%	\end{proof}		
%			\vspace{0.1cm}\noindent
%			Applying Lemma \ref{lemRes} to $H$ we deduce the following result:
			
%			\textcolor{red}{N: Correction jusqu'à là}

			\begin{prop}\label{relative}
			    Suppose that Hypotheses~\ref{hyp:eigenvalues} and \ref{hyp:estimate resolvente} hold for $H$ on $I=(s_0,s_1)$, $s_1>s_0$. Fix $m \in \N^*$. Then there exists $c>0$ such that for any $J=(r_0,r_1)$, with $0<(s_1+c)^{-m} < r_0 < r_1 < (s_0+c)^{-m}$ the operator $(H+c)^{-m}$ satisfies Hypotheses~\ref{hyp:eigenvalues} and \ref{hyp:estimate resolvente} on $J$ with $n_J=n_I$.
				%Let $\left(H, \mathcal{D}(H)\right)$ be a closed operator acting on a complex separable Hilbert space $\Hi$, and suppose that the spectrum of $H$ admits a free strip near the real axis, i.e., there exists $a > 0$ such that $S_a(\R) \subset \rho(H)$. Assume that:
				%\begin{enumerate}
				%	\item For every compact interval $I \subset \R$, there exist constants $c_0 > 0$ and $n_0 \in \N$ such that for all $z \in S_a(I)$,
				%	\begin{equation}\label{eq:hypGen1}
				%		\| \Res_H(z) \|_{\mathcal{B}(\Hi)} \leq c_0 |\im(z)|^{-n_0}.
				%	\end{equation}
				%	\item There exists $\theta_0 \in (0, \frac{\pi}{2})$ such that
				%	\begin{equation}\label{HSect}
				%		\sigma(H) \subset (0, +\infty) \cdot e^{i[-\theta_0, \theta_0]}.
				%	\end{equation}
				%\end{enumerate}
				%for any $m \in \N^*$ which satisfies $m \theta_0 < \frac{\pi}{2}$, there exist constants $a_m > 0$ and $C_m > 0$ such that $S_{a_m}((0,\infty)) \subset \rho\left((H+c)^{-m}\right)$ and for all $Z \in S_{a_m}((0,\infty))$,
				%\begin{equation}\label{estiRm}
				%	\left\| \left( H^{-m} - Z \right)^{-1} \right\|_{\mathcal{B}(\Hi)} \leq C_m |\im(Z)|^{-n_0}.
			%	\end{equation}
			%	{\clr F: est ce plutôt $(H+c)^{-m}$ ?}
				Moreover, for any $f \in D(J)$, we have
				\[
				f((H+c)^{-m}) = g_m(H),
				\]
				where $g_m \in D(I)$ is defined by $g_m(\lambda) = f((\lambda+c)^{-m})$.
			\end{prop}

			\begin{proof}
			Let us fix $\theta_0\in (0,\pi/2)$. As $H_0$ is bounded from below and $V$ is bounded, there exists $c>0$ such that $\sigma(H+c) \subset (0, +\infty) \cdot e^{i[-\theta_0, \theta_0]}$. Moreover, since $	\sigma_\mathrm{ess}(H) = \sigma_\mathrm{ess}(H_0)$, each non-real number in $\sigma(H+c)$ is a discrete eigenvalue. By using Lemma~\ref{lemRes} with $L=H+c$ and the biholomorphic function
				\begin{align}
					\varphi_m : -c + (0, +\infty) \cdot e^{i[-\theta_0, \theta_0]} &\longrightarrow (0, +\infty) \cdot e^{i[-m\theta_0, m\theta_0]} \\
					z =-c + \rho e^{i\theta} &\longmapsto (z+c)^{-m} = \rho^{-m} e^{-im\theta},
				\end{align}
				we have $\sigma((H+c)^{-m}) \subset (0, +\infty) \cdot e^{i[-m\theta_0,m \theta_0]} \subset \C_>$ and the non-real spectrum of $(H+c)^{-m}$ consists of discrete eigenvalues $\Lambda = \varphi_m(\lambda)$ with $\lambda$ a non-real eigenvalue of $H$. So if $H$ satisfies Hypothesis~\ref{hyp:eigenvalues}, no point in the interval $I$ is an accumulation point of discrete eigenvalues and Hypothesis~\ref{hyp:eigenvalues} holds true for $(H+c)^{-m}$ on $\varphi_m(I)= ((s_1+c)^{-m} , (s_0+c)^{-m})$.
				
				Moreover, under the Hypothesis~\ref{hyp:estimate resolvente} for $H$ on $I$, there exists $a_I>0$ such that $S_{a_I}(I)\subset \rho(H)$ and then, $\varphi_m(S_{a_I}(I))\subset \rho((H+c)^{-m})$.
				Let $J = (r_0, r_1)$ with $(s_1+c)^{-m} < r_0 < r_1 < (s_0+c)^{-m}$. By taking $a_J>0$ sufficiently small such that $S_{a_J}(J) \subset \varphi_m(S_{a_I}(I))$, we have $S_{a_J}(J))\subset \rho((H+c)^{-m})$ and for $Z \in S_{a_J}(J)$ there exists $z \in S_{a_I}(I)$ such that $Z= \varphi_m(z)$.
				By combining \eqref{eqRes1} for $L=H+c$ (and for $z+c$ instead of $z$) with Hypothesis~\ref{hyp:estimate resolvente} for $H$ on $I$, we obtain:
				\begin{align*}
				    \| \Big( (H+c)^{-m} - Z\Big)^{-1} \|_{\mathcal{B}(\Hi)} & \leq \sup_{z \in \overline{S_{a_I}(I)}} \|((z+c)^{-m} G_m((z+c)(H+c)^{-1}))^{-1} \| \, \|(H+c)(H-z)^{-1} \|\\
				    & \leq C_J |\im(z)|^{-n_I-1},
				\end{align*}
				for some $C_J>0$.
				 Writing $z = -c + \rho e^{i\theta}$, with $\theta \in (-\frac{\pi}{2m} , \frac{\pi}{2m})$ (it is possible for $a_I$ sufficiently small), we have
				$$
				Z = (z+c)^{-m} = \rho^{-m} e^{- i m \theta} , \quad  \text{ and } \quad \im (Z) =  -\rho^{-m} \sin(m\theta).
				$$
				Thanks to the inequality $\frac{2}{\pi}x \leq \sin(x) \leq x$ for $x \in [0, \pi/2]$, we deduce
				\begin{equation}\label{eqZz}
					\rho^{-(m+1)}  \frac{2}{\pi} m |\im (z)| \leq 
					\frac{2}{\pi} \rho^{-m}  m |\theta| \leq 
					|\im (Z)| \leq \rho^{-m}  m |\theta| \leq \rho^{-(m+1)} \frac{\pi}{2} m |\im (z)|,
				\end{equation}
			and the Hypothesis~\ref{hyp:estimate resolvente} holds for $(H+c)^{-m}$ on $J$.
			
				Now let $f \in D(J)$ and let $\tilde{f} \in \mathcal{C}^\infty_c(\C)$ be an almost analytic extension of $f$, supported in $J \times (-a_J, a_J)$. Using the change of variables $Z = \varphi_m(z)$ (with $Z = X + iY$, $z = x + i y$) in the functional calculus formula
				$$
				f((H+c)^{-m}) := \pi^{-1} \int_{\C} \partial_{\overline{Z}} \tilde{f}(Z) \left((H+c)^{-m} - Z\right)^{-1} \, \mathrm{d}X \mathrm{d}Y,
				$$
				we obtain
				$$
				f((H+c)^{-m}) := \pi^{-1} \int_{\C} \partial_{\overline{Z}} \tilde{f}(\varphi_m(z)) \left((H+c)^{-m} - (z+c)^{-m}\right)^{-1} \left| \partial_z \varphi_m(z) \right|^2 \, \mathrm{d}x \mathrm{d}y.
				$$
				
				\vspace{0.1cm}\noindent
				From \eqref{eqZz}, we know that the function $\tilde{g}_m := \tilde{f} \circ \varphi_m$ is an almost analytic extension of $g_m = f \circ \varphi_m$, %\footnote{A détailler ?}
				and satisfies
				$$
				\partial_{\bar{z}} \tilde{g}_m(z) = \partial_z \varphi_m(\bar{z}) \cdot \partial_{\overline{Z}} \tilde{f}(\varphi_m(z)).
				$$
				Then, using the identity \eqref{eqRes2} from Lemma~\ref{lemRes}, and the fact that the map 
				$$
				z \mapsto \sum_{k=1}^{m-1} G_k((z+c)(H+c)^{-1}) G_m((z+c)(H+c)^{-1})^{-1}
				$$
				is holomorphic on the support of $\tilde{g}_m$, we obtain
				$$
				f((H+c)^{-m}) = \pi^{-1} \int_{\R^2} \partial_{\bar{z}} \tilde{g}_m(z)
				\left( - \frac{(z+c)^{m+1}}{m} (H - z)^{-1} \right) \partial_z \varphi_m(z) \, \mathrm{d}x \mathrm{d}y = g_m(H),
				$$
				where we use the identity $- \frac{(z+c)^{m+1}}{m} \cdot \partial_z \varphi_m(z) = 1$.
			\end{proof}
			
	        \section{SSF for Trace class perturbations}\label{subsec:SSF}
			
			%In this section we return to the framework described in subsection \ref{sub:setting}. 
			
			\subsection{Existence}
			
			The following theorem provides the key analytic framework needed to define the spectral shift function (SSF) in a general non-self-adjoint setting. It ensures the trace-class property of the difference $f(H) - f(H_0)$ under suitable assumptions, allowing one to associate a distribution to this difference.

			\begin{theorem}\label{thm:existence SSF}
				Suppose that Hypotheses \ref{hyp:eigenvalues}, \ref{hyp:estimate resolvente} and \ref{hyp:Trace} hold on an open interval $I$ with $m=-1$. 
				Then, for all $f\in D(I)$ the operator difference $f(H) - f(H_0)$ belongs to $\mathcal{L}_1(\Hi)$. Moreover, the map
        	\begin{equation}\label{eq:SSFD'}
					f \longmapsto \Tr\big(f(H) - f(H_0)\big)
				\end{equation}
				defines a distribution on $I$ which vanishes on $I \cap \rho(H) \cap \rho(H_0)$.
			\end{theorem}
			
			\begin{proof}
				Let $\tilde{f}$ be an almost analytic extension of $f$ of order $n_I+2$, 
				supported in $\overline{S_a(K)}$, with $K$ a compact subinterval of $I$ containing the support of $f$ and  $a>0$ chosen small enough so that $\sigma(H)\cap S_a(K)=\emptyset$. Then, %using the Helffer–Sj\"ostrand formula, 
				by definition, we have
				\begin{align*}
					f(H) &- f(H_0)
					= \pi^{-1} \int_{\C} \partial_{\bar{z}} \tilde{f}(z) \left( \Res_H(z) - \Res_0(z) \right) \, \mathrm{d}x\mathrm{d}y   \\
					&= -\pi^{-1} \int_{\C} \partial_{\bar{z}} \tilde{f}(z) \Res_H(z) V \Res_0(z) \,
					\mathrm{d}x\mathrm{d}y \\
					&= -\pi^{-1} \left(\int_{\im(z)>0} \partial_{\bar{z}} \tilde{f}(z) \Res_H(z) V \Res_0(z) \, \mathrm{d}x\mathrm{d}y + \int_{\im(z)<0} \partial_{\bar{z}} \tilde{f}(z) \Res_H(z) V \Res_0(z) \, \mathrm{d}x\mathrm{d}y\right). 
				\end{align*}
				%By construction of $\tilde{f}$, 
				%\begin{equation}\label{eq:proof trace class1}
				%	|\partial_{\bar{z}} \tilde{f}(z)| \cdot \|\Res_0(z)\|_{\mathcal{B}(\Hi)} \cdot %\|\Res_H(z)\|_{\mathcal{B}(\Hi)} \leq c_0\|f\|_{n_I+2}.
				%\end{equation}
				\noindent
				For all $z \in \rho(H) \cap \rho (H_0)$, $\partial_{\bar{z}} \tilde{f}(z) \Res_H(z) V \Res_0(z)$ is of trace class and since $\tilde{f}$ is compactly supported, this shows that the full integral belongs to $\mathcal{L}_1(\Hi)$. Finally,	in the same way as \eqref{estimationcf}, we get 
			
				\begin{equation*}
					\left| \Tr\left(f(H) - f(H_0)\right) \right| \leq \beta_{K} \|V\|_{\mathcal{L}_1}\|f\|_{n_{K}+2},
				\end{equation*}
				which prove that \eqref{eq:SSFD'} is a distribution on $I$. Finally Proposition~\ref{prop:propertiesfunctionalcalculus} implies that the distribution vanishes on $\R\cap\rho(H)\cap\rho(H_0)$.  
			\end{proof}			
			
			%%%%%%%%%%%%%%%%%%%%%%  Définition de la SSF %%%%%%%%%%%%%%%%%

			Under the assumptions of Theorem~\ref{thm:existence SSF}, the trace formula naturally leads to the definition of the Spectral Shift Function on $I$, which captures the variation of the spectrum under perturbation. This function is initially defined only up to an additive constant, which can be fixed by prescribing its value on a spectral gap when such a gap exists in $I \cap \rho(H) \cap \rho(H_0)$.

			\begin{definition}\label{def:xi}
			   Suppose that Hypotheses \ref{hyp:eigenvalues}, \ref{hyp:estimate resolvente} and \ref{hyp:Trace} hold on an open interval $I$ with $m=-1$. The \emph{spectral shift function} $\xi(\cdot; H, H_0)$ on $I$ is defined, up to an additive constant, as the distribution satisfying
				\begin{equation*}
					\langle \xi', f \rangle := \Tr\big(f(H) - f(H_0)\big), \quad \text{for all } f \in D(I).
				\end{equation*}
				When $I \cap \rho(H) \cap \rho(H_0)$ is nonempty (for example, if $I = \R$ and both $H_0$ and $H$ are bounded or semi-bounded), the additive constant can be uniquely determined by prescribing the value of $\xi$ on a spectral gap, i.e., an interval contained in $I \cap \rho(H) \cap \rho(H_0)$.
			\end{definition}
			
			%%%%%%%%%%%%%%%%%%%%%%%%%%%%%%%%%%%%%%%%%%%%%%%%%%%%%%%%%%%%%%
			
			\subsection{A first representation formula for the SSF}
			
			Once the spectral shift function (SSF) has been defined, we aim to establish more precise properties, starting with a first representation formula. For $z \in \rho(H) \cap \rho(H_0)$, we define
			\begin{equation*}
				\sigma(z) := \Tr\big(\Res_H(z) - \Res_0(z)\big).
			\end{equation*}
			This quantity is well-defined whenever $H - H_0 \in \mathcal{L}_1(\Hi)$, as ensured by the resolvent identity. Our goal is to show that the derivative of the spectral shift function, as defined in Definition~\ref{def:xi}, admits a concrete representation in terms of the discontinuity of $\sigma(z)$ across the real axis.
			
			\begin{prop}
		\label{lm:green formula}
				Let $I \subset \R$ be an open interval. Suppose that Hypotheses \ref{hyp:eigenvalues}, \ref{hyp:estimate resolvente} and \ref{hyp:Trace} hold on $I$ with $m=-1$. Then, in the sense of distributions on $I$, we have
				\begin{equation}\label{ssf=sigma}
					\xi^\prime(\cdot; H, H_0) = \frac{1}{2\pi i} \lim_{\varepsilon \rightarrow 0^+} \left( \sigma(\cdot + i\varepsilon) - \sigma(\cdot - i\varepsilon) \right).
				\end{equation}
			\end{prop}

			\begin{remark}
				Note that the right-hand side of~\eqref{ssf=sigma} vanishes on the set $I \cap \rho(H) \cap \rho(H_0)$, since $\sigma(\cdot)$ is well-defined and continuous there. More precisely, for any $\lambda \in I \cap \rho(H) \cap \rho(H_0)$, we have
				$$
				\lim_{\varepsilon \rightarrow 0^+} \sigma(\lambda \pm i\varepsilon) = \sigma(\lambda).
				$$
			\end{remark}

			\begin{proof}
				Let $f \in D(I)$ be a test function supported in a compact interval $K \subset I$, and let $\tilde{f}$ be an almost analytic extension %of order $m+1$ of $f$, 
				supported in $\overline{S_a(K)}$, where $a > 0$ is choosen small enough so that $\sigma(H)\cap S_a(K)=\emptyset$. Then, %using the Helffer–Sj\"ostrand formula, 
				we write:
				\begin{align}
					\langle \xi', f \rangle 
					&= \pi^{-1} \int_{\C} \partial_{\bar{z}} \tilde{f}(z) \Tr\left( \Res_H(z) - \Res_0(z) \right) \, \mathrm{d}x\mathrm{d}y  \nonumber \\
					&= \pi^{-1} \left( \lim_{\varepsilon \to 0^+} \int_{\im(z)>0} \partial_{\bar{z}} \tilde{f}(z) \, \sigma(z + i\varepsilon) \, \mathrm{d}x\mathrm{d}y \,
					+ \lim_{\varepsilon \to 0^+} \int_{\im(z)<0} \partial_{\bar{z}} \tilde{f}(z) \, \sigma(z - i\varepsilon) \, \mathrm{d}x\mathrm{d}y \,  \right),
					\label{eq:proof Green representation1}
				\end{align}
				where $\sigma(z) := \Tr\left( \Res_H(z) - \Res_0(z) \right)$. 
				
				\vspace{0.2cm}\noindent
			Indeed, by construction, the support of $\tilde{f}$ does not contain any complex eigenvalue of $H$. Hence, for $\varepsilon > 0$ small enough, the map $z \mapsto \partial_{\bar{z}} \tilde{f}(z) \Res_H(z \pm i\varepsilon)$ is continuous on $\supp(\tilde{f})$. Moreover, since $H - H_0 \in \mathcal{L}_1(\Hi)$, resolvent identity implies that $\Res_H(z) - \Res_0(z)$ is trace-class for $z \in \rho(H) \cap \rho(H_0)$, and we have
				$$
				\lim_{\varepsilon \to 0^+} \sigma(z \pm i\varepsilon) = \sigma(z), \quad \text{for all } z \in S_a(K) \text{ with } \pm \im(z) > 0.
				$$
				In addition, for all sufficiently small $\varepsilon > 0$ and all $z \in S_a(K)$ with $\pm \im(z) > 0$, we estimate
				\begin{align*}
					\left| \partial_{\bar{z}} \tilde{f}(z) \, \sigma(z \pm i\varepsilon) \right|
					&= \left| \partial_{\bar{z}} \tilde{f}(z) \, \Tr\left( \Res_H(z \pm i\varepsilon) V \Res_0(z \pm i\varepsilon) \right) \right| \\
					&\leq \left| \partial_{\bar{z}} \tilde{f}(z) \right| \cdot \|V\|_{\mathcal{L}_1(\Hi)} \cdot \|\Res_H(z \pm i\varepsilon)\| \cdot \|\Res_0(z \pm i\varepsilon)\| .
				%	&\leq c_K \|f\|_{n_K+2} \|V\|_{\mathcal{L}_1}.
				\end{align*}
				%for some constant $c_0 > 0$, uniformly on the support of $\tilde{f}$.
				
				\vspace{0.1cm}\noindent
				Now, since $z \mapsto \sigma(z \pm i\varepsilon)$ is holomorphic on $S_a(K) \cap \C^\pm$, we may apply Stokes' theorem to obtain:
				\begin{equation}
					\int_{\pm \im(z) > 0} \partial_{\bar{z}} \tilde{f}(z) \, \sigma(z \pm i\varepsilon) \, \mathrm{d}x \, \mathrm{d}y 
					= \pm \frac{1}{2i} \int_{\Gamma_a^\pm} \tilde{f}(\lambda) \, \sigma(\lambda \pm i\varepsilon) \, \mathrm{d}\lambda,
					\label{eq:proof Green representation2}
				\end{equation}
				where $\Gamma_a^\pm$ denotes the boundary of $S_a(K) \cap \C^\pm$. This boundary integral can be decomposed into four parts, three of which vanish due to the support properties of $\tilde{f}$. Therefore,
				\begin{equation}\label{eq:proof Green representation3}
					\pm \frac{1}{2i} \int_{\Gamma_a^\pm} \tilde{f}(\lambda) \, \sigma(\lambda \pm i\varepsilon) \, \mathrm{d}\lambda
					= \pm \frac{1}{2i} \int_{\supp(f)} f(\lambda) \, \sigma(\lambda \pm i\varepsilon) \, \mathrm{d}\lambda.
				\end{equation}
				Substituting~\eqref{eq:proof Green representation3} into~\eqref{eq:proof Green representation1} yields the desired identity~\eqref{ssf=sigma}.
			\end{proof}
			
			\begin{remark}
			 Proposition \ref{lm:green formula}  allows to give the following extension to non-selfadjoint operators of the representation formula \eqref{det}. Under the assumptions of Theorem \ref{thm:existence SSF}, 
			 %for $V:=H - H_0 \in \mathcal{L}_1(\Hi)$, $z \in \C \setminus \R$
			 let us introduce the perturbation determinant 
			 \begin{equation}\label{defDV}
			     D_V(z):= \Det \left( (H_0 + V-z) (H_0 - z)^{-1} \right) = \Det \left( I + V (H_0 - z)^{-1} \right).
			 \end{equation}
			 It is a non vanishing analytic function on each $S_a(K)$, $K \subset I$ compact (because it is in the resolvent set of $H$ and $H_0$) and it satisfies (see e.g. \cite[Chapter 8]{Ya92}):
			 \[D_V(z)^{-1} D_V^\prime(z)= \Tr\big(\Res_H(z) - \Res_0(z)\big) = 	\sigma(z). \]
			 It follows that, up to a constant, the function $\ln D_V(z)$ is well defined on each $S_a(K)$ and satisfies:
			 \[\frac{d}{d\lambda}\ln D_V(\lambda \pm i \varepsilon) = \sigma(\lambda \pm i \varepsilon), \quad \lambda \in J, \quad \varepsilon >0.\]
			 Then, up to a constant, Proposition \ref{lm:green formula} suggest the following extension of \eqref{det} (when the limit exists):
			 \begin{equation}\label{ssfLnD}
			    \xi(\lambda; H, H_0) = \frac{1}{2\pi i} \lim_{\varepsilon \to 0^+} \left( \ln D_V(\lambda + i \varepsilon) - \ln D_V(\lambda - i \varepsilon) \right).
			 \end{equation}
			 Of course, when $H_0+V$ is selfadjoint, it corresponds to \eqref{det} because in this case, $\ln D_V(\lambda - i \varepsilon)$ is the complex conjugate of $\ln D_V(\lambda + i \varepsilon)$.
			\end{remark}
			
		 \begin{remark} 
			%\footnote{\clb V: INDIQUER et DEMONTRER CE RESULTAT PLUS TOT de facon generale?}
			 From the above results it follows that %when it exists, 
			 the SSF for the adjoint $H^*$ coincides with the complex conjugate of the SSF for $H$:
			\begin{equation}\label{ssf*}
			     \xi(\lambda; H^*, H_0) = \overline{\xi(\lambda; H, H_0)}.
			 \end{equation} 
			\end{remark}

			\section{The SSF for non-selfadjoint relatively trace class perturbations}\label{Sec:relTr}

			In many applications, the difference $H - H_0$ does not belong to the trace class. However, the difference between their resolvents, or between powers of their resolvents, may indeed be trace class. The aim of this section is to extend the definition of the spectral shift function (SSF) to such settings.

			\vspace{0.2cm}\noindent
			The central idea is to consider the pair of operators
			$$
			(X, X_0) := \left( (H + c)^{-m},\ (H_0 + c)^{-m} \right),
			$$
			where $c \in \R$ and $m \in \N^*$ are chosen so that $X - X_0 \in \mathcal{L}_1(\Hi)$. In this case, Definition~\ref{def:xi} can be applied directly to $(X, X_0)$, thus defining a spectral shift function for this pair. Then, to retrieve the SSF associated with the original pair $(H, H_0)$, we perform a change of variables: for $\lambda > -c$, the spectral parameter $\mu$ for $X$ is taken as $\mu = (\lambda + c)^{-m}$. Accordingly, we define
			\begin{equation}
				\xi(\lambda; H, H_0) := \xi\big((\lambda + c)^{-m};\ (H + c)^{-m},\ (H_0 + c)^{-m}\big).
			\end{equation}

			\subsection{Definition of the SSF for relatively trace class perturbations}

			We begin with a theorem that allows the definition of the spectral shift function (SSF) in the case where the perturbation is relatively trace class, i.e., when a suitable power of the resolvent difference belongs to $\mathcal{L}_1(\Hi)$.

			%%%%%%%%%%%%%%%%%%%%%%%%%%% Commentaires %%%%%%%%%%%%%%%%%%%%%%%%%%%

			\begin{prop}\label{prop:ssfRelTr}
				Let $\left(H_0,\mathcal{D}(H_0)\right)$ be a self-adjoint, semi-bounded operator, and let $\left(H,\mathcal{D}(H_0)\right)$ be a closed operator acting on a complex separable Hilbert space $\Hi$.
				
				Assume there exist $c \in \R$ and $m \in \N^*$ such that:
				\begin{enumerate}
					\item $V:=H-H_0$ is a bounded operator
					\item\label{it2:existenceSSF} The operator $H+c$ is invertible with
					\begin{align}
					%	(-\infty, -c] &\subset \rho(H), \\
						\sigma\big((H + c)^{-m}\big) &\subset (0, +\infty) \cdot e^{i(-\frac{\pi}{2}, \frac{\pi}{2})} 
					%	(H + c)^{-m} - (H_0 + c)^{-m} &\in \mathcal{L}_1(\Hi).
					\end{align}
					\item
				Hypotheses \ref{hyp:eigenvalues}, \ref{hyp:estimate resolvente} and \ref{hyp:Trace} hold on $I=(s_0,s_1)$, $s_0<s_1$ for some $m \in \N^*$	
					%\label{it1:existenceSSF} There exist constants $a > 0$, $n_0 \in \N$, and $c_0 > 0$ such that 
				%	\[	U_a := \{ z \in \C : 0 < |\im(z)| < a \} \subset \rho(H) \cap \rho(H_0),	\]
				%	and, on $U_a$, the resolvent of $H$ satisfies the estimate
				%	\begin{equation}\label{eq:estRes}
				%		\|\Res_H(z)\|_{\mathcal{B}(\Hi)} \leq c_0 |\im(z)|^{-n_0}.
				%	\end{equation}
				\end{enumerate}
				
				Then the spectral shift function $\xi_{c,m} := \xi(\cdot\, ; (H + c)^{-m}, (H_0 + c)^{-m})$ is well-defined on 
				$((s_1+c)^{-m} , (s_0+c)^{-m})$ up to a constant. 
				%$(0, +\infty)$ and is uniquely determined by the normalization condition
			%	\[	\xi_{c,m}(\mu) = 0 \quad \text{for all } \mu \geq A,	\]
				%where $A > \| (H + c)^{-m} \|$.
				
				Moreover, the distribution $\xi(\cdot\, ; H, H_0)$ defined by
				\begin{equation}\label{defssfRT}
					\xi(\lambda; H, H_0) := \xi_{c,m}((\lambda + c)^{-m})  
					%\begin{cases}
					%	\xi_{c,m}((\lambda + c)^{-m}) & \text{for } \lambda > -c, \\
					%	0 & \text{for } \lambda \leq -c,
					%\end{cases}
				\end{equation}
			for $\lambda \in I=(s_0,s_1)$, is independent of the choice of $(c,m)$ and satisfies the trace formula \eqref{ssf}. Uniqueness is guaranteed by choosing $\xi(\cdot; H, H_0)= 0$ on a spectral gap (if it exists).
			\end{prop}
			
			\begin{proof}
				Since the spectrum of $(H+c)^{-m}$ lies in $(0, +\infty) \cdot e^{i(-\frac{\pi}{2}, \frac{\pi}{2})}$, there exists $\theta_0 \in\, ]0, \frac{\pi}{2m}[$  such that
				$$
				\sigma(H+c) \subset ]0, +\infty[ e^{i[-\theta_0, \theta_0]}.
				$$
				Then, by  Proposition~\ref{relative}, for any interval $J=(r_0,r_1)$ relatively compact in  $((s_1+c)^{-m} , (s_0+c)^{-m}) \subset ]0, +\infty[$, the Hypotheses \ref{hyp:eigenvalues}, \ref{hyp:estimate resolvente} are satisfied by $(H+c)^{-m}$ on $J$. Applying Theorem~\ref{thm:existence SSF}, the spectral shift function $\xi_{m,c} := \xi(\cdot; (H+c)^{-m}, (H_0+c)^{-m})$ is therefore well defined on $J$, up to an additive constant. For every $f \in D(I)$, we have:
				$$
				\Tr \big( f((H+c)^{-m}) - f((H_0+c)^{-m}) \big) = - \langle \xi_{m,c}', f \rangle.
				$$
				%Moreover, we may normalize $\xi_{m,c}$ by requiring that $\xi_{m,c} = 0$ on $[A, +\infty[ \subset \rho((H+c)^{-m})$, for some $A > \| (H+c)^{-m} \|$.
				
				\vspace{0.1cm}\noindent
				Furthermore, for any $g \in D(((s_1+c)^{-m} , (s_0+c)^{-m}))$, Proposition~\ref{relative} (applied to the function $f \in D(I)$ defined by $f(\mu) = g(\mu^{-\frac{1}{m}} - c)$) yields:
				$$
				f((H+c)^{-m}) = g(H),
				$$
				and consequently:
				$$
				\Tr \big( g(H) - g(H_0) \big) = \langle \xi, g^\prime \rangle.
				$$
				Since $g(H)$ and $g(H_0)$ are independent of the choice of $(m, c)$, the distribution $\xi$ is uniquely determined on $I$ by the normalization $\xi = 0$ on a spectral gap (if it exists). 
				%the interval $]-\infty, \inf(\sigma(H) \cap \mathbb{R})[$, which contains $]-\infty, -c[$ by assumption.
			\end{proof}

			Thanks to this Proposition, in the context of relatively trace-class perturbations, the spectral shift function can still be defined, as stated in the following definition.

			\begin{definition}\label{def:ssf relative trace class}
				Under the assumptions of Proposition~\ref{prop:ssfRelTr}, we define the Spectral Shift Function (SSF) associated with the pair $(H, H_0)$ by the relation~\eqref{defssfRT}.
			\end{definition}
			
			\begin{remark}
			 Let $s_{\text min}:= \inf \sigma(H) \cap \sigma(H_0) \cap \R$. If, in Proposition \ref{prop:ssfRelTr}, $s_0 < s_{\text min}$, then the SSF can be extended to $0$ throughout the spectral gap $(-\infty, s_{\text min})$.
			\end{remark}
			\subsection{Representation formula in the case of relatively trace class perturbations}

			In this section, we establish an analogue of Lemma~\ref{lm:green formula} in the framework where the perturbation is only relatively trace class, rather than trace class. This generalization is essential for extending the trace formula and the definition of the spectral shift function to broader classes of operators.

			\begin{prop}\label{prop:repSSFRT}
				Suppose that the pair $(H_0,H)$ satisfies the assumptions of Proposition~\ref{prop:ssfRelTr} for some $I= (s_0, s_1)$, $c \in \R$ and $m \geq 1$, and that
				\begin{equation}\label{tr2}
					(H_0 + c)^{-m} \big( (H + c)^{-1} - (H_0 + c)^{-1} \big) \in \mathcal{L}_1(\Hi).
				\end{equation}
				Then, in the sense of distributions, on $I$, we have:
				\begin{equation}\label{ssf=sigma_m}
					\xi^\prime(\cdot; H,H_0) = \frac{1}{2\pi i} \lim_{\varepsilon \to 0^+} \left( \sigma_m(\cdot + i\varepsilon) - \sigma_m(\cdot - i\varepsilon) \right),
				\end{equation}
				where
				\begin{equation*}
					\sigma_m(z) := (z + c)^{m-1} \Tr\left( (H + c)^{-m+1} \Res_H(z) - (H_0 + c)^{-m+1} \Res_0(z) \right), \qquad \im(z) \neq 0.
				\end{equation*}
			\end{prop}

			\vspace{0.1cm}\noindent
			
			\begin{proof}
				Let us denote $X = (H + c)^{-m}$ and $X_0 = (H_0 + c)^{-m}$. Using the identity
				\begin{align*}
				(H + c)^{-m+1} \Res_H(z) - (H_0 + c)^{-m+1} \Res_0(z) = (X - X_0)(H + c)\Res_H(z) + X_0 \big( (H + c)\Res_H(z) - (H_0 + c)\Res_0(z) \big),
				\end{align*}
				together with the resolvent identity and the strong mapping spectral theorem (see e.g., \cite[Lemma 2, XIII.4]{ReSi80_01}), one checks that $\sigma_m(z)$ is well-defined for all $z \in \rho(H) \cap \rho(H_0)$ (see also the relation \eqref{Sigsig} below for further justification).
				
				\vspace{0.1cm}\noindent
				Moreover, for $\lambda \in \R \cap \rho(H) \cap \rho(H_0)$, we have
				$$
				\lim_{\varepsilon \to 0^+} \sigma_m(\lambda + i\varepsilon) = \sigma_m(\lambda) = \lim_{\varepsilon \to 0^+} \sigma_m(\lambda - i\varepsilon),
				$$
				so that both side of \eqref{ssf=sigma_m} vanishes on $\R \cap \rho(H) \cap \rho(H_0) \supset (-\infty, \inf (\sigma(H) \cup \sigma(H_0))) \supset (-\infty, -c]$. By Definition~\ref{def:ssf relative trace class}, on $I$, we have 
				$$
				\xi'(\lambda; H, H_0) = -m(\lambda + c)^{-m-1} \, \xi'_{c,m}((\lambda + c)^{-m}).
				$$
				
				\vspace{0.1cm}\noindent
				We now apply Proposition~\ref{lm:green formula} to the pair $((H + c)^{-m}, (H_0 + c)^{-m})$, which yields, in the sense of distributions:
				$$
				\xi'_{c,m}(\mu) = \frac{1}{2\pi i} \lim_{\varepsilon \to 0^+} \left( \Sigma(\mu + i\varepsilon) - \Sigma(\mu - i\varepsilon) \right),
				$$
				where
				$$
				\Sigma(Z) = \Tr\big( (X - Z)^{-1} - (X_0 - Z)^{-1} \big), \qquad \im(Z) \neq 0.
				$$
				For $Z = \mu \pm i\varepsilon$ with $\mu > 0$, $\varepsilon > 0$, there exists $z = \lambda \mp i\delta(\varepsilon)$ with $\lambda > -c$ and $\delta(\varepsilon) \to 0^+$ such that $Z = (z + c)^{-m}$. Using Lemma~\ref{lemRes} for with $L=H+c$ and with $z+c$ instead of $z$, we have:
				\begin{align*}
				  (H + c)^{-m+1} \Res_H(z) & = X(H+c) \Res_H(z)\\
				  &= X + (z+c) X \Res_H(z)\\
				 & =  X - mZ X(X - Z)^{-1} - X B_H(z)\\
				 & =X - mZ -mZ^2(X - Z)^{-1} - X B_H(z)
				\end{align*}
				where
				$$
				B_H(z) := \sum_{k=1}^{m-1} G_k((z + c)(H + c)^{-1}) \, G_m((z + c)(H + c)^{-1})^{-1},
				$$
				is holomorphic near $I\times \{0\}$.
				
				%{SI OK, adapter la suite} and $B_H(z)$ is holomorphic near the real axis. 
				The same identity holds for $H_0$. Therefore, we obtain
				\begin{align} \nonumber
					(z + c)^{m-1} & \Big( (H + c)^{-m+1} \Res_H(z) - (H_0 + c)^{-m+1} \Res_0(z) \Big) \\
					&= (z + c)^{m-1} (X - X_0) - m(z + c)^{-m-1} \big( (X - Z)^{-1} - (X_0 - Z)^{-1} \big) \label{Sigsig} \\
					&\quad - (z + c)^{m-1} \big( X B_H(z) - X_0 B_{H_0}(z) \big). \nonumber
				\end{align}
				
				\vspace{0.1cm}\noindent
				The first and third terms in the right-hand side of \eqref{Sigsig} are holomorphic near the real axis and hence do not contribute to the jump in the limit. The only contribution to the difference of boundary values comes from the second term. This yields
				$$
				-m(\lambda + c)^{-m-1} \lim_{\varepsilon \to 0^+} \big( \Sigma((\lambda + c)^{-m} + i\varepsilon) - \Sigma((\lambda + c)^{-m} - i\varepsilon) \big) = \lim_{\varepsilon \to 0^+} \left( \sigma_m(\lambda + i\varepsilon) - \sigma_m(\lambda - i\varepsilon) \right),
				$$
				which proves \eqref{ssf=sigma_m}.
				
				\vspace{0.1cm}\noindent
				Finally, we briefly justify the trace class property of the third term in \eqref{Sigsig}. By the structure of $B_H(z)$ and the assumption \eqref{tr2}, one can verify that $z \mapsto X_0 \big( B_H(z) - B_{H_0}(z) \big)$ is holomorphic with values in $\mathcal{L}_1(\Hi)$.
				
			\end{proof}

			%			\subsection{Existence of the SSF under assumption on spectral singularities}

			%				Building on the previous results, the following theorem shows that, under suitable assumptions—including a relatively trace-class resolvent difference—the spectral shift function introduced in Proposition~\ref{prop:ssfRelTr} is well defined for the pair $(H, H_0)$.

			%\begin{theorem}
			%	Suppose that Assumptions \ref{hyp:eigenvalues} and \ref{hyp:spectral singularity} hold, and that there exists $c > 0$ such that
			%	\[
			%	(H + c)^{-1} - (H_0 + c)^{-1} \in \mathcal{L}_1(\Hi).
			%	\]
			%	Then the spectral shift function $\xi_1$ associated with the pair $(H, H_0)$, as defined in Definition \ref{def:ssf relative trace class}, is well defined.
			%\end{theorem}
			
			%\vspace{0.1cm}\noindent
			%{\clb F: doit-on dire quelque chose de plus ici ?}

			%%%%%%%%%%%%%%%%%%%%%%%%%%%%%%%%%  NOUVELLE SECTION %%%%%%%%%%%%%%%%%%%%%%%%%%%%%%%%%%%%%%
			
\section{Limiting absorption principle for non-self-adjoint perturbation}\label{Sec:spsing}

 In this section, we show how the results of the previous sections can be applied in the general setting introduced by J. Faupin and the second author in \cite{FaFr23}. First, we  recall some well-known definitions related to the limiting
absorption principle for non-self-adjoint perturbations of self-adjoint
operators, and then we provide sufficient conditions for
Hypothesis~\ref{hyp:estimate resolvente} to hold.
We consider $H_0$ and $H=H_0+V$ as in Section \ref{sub:setting}, and we suppose that $V$ is of the form
	\begin{equation*}
		V = C W C, 
	\end{equation*}
	where $C \in \mathcal{B}(\mathcal{H})$ is selfadjoint, 
	%positive definite\footnote{Je pense qu'on peut enlever positive definite. F :  Se méfier. En tout $C$ doit etre a.a} 
	relatively compact with respect to $H_0$, and $W \in \mathcal{B}(\mathcal{H})$. %Recall that a positive definite operator is positive and injective. 
	%The adjoint of $V$ is then given by
%	\begin{equation*}
%		V^\ast = C W^\ast C.
%	\end{equation*}

\subsection{Spectral singularities}
	
	A central concept in the study of non-self-adjoint operators introduced in \cite{FaFr23} is that of a \emph{spectral singularity}, which refers to points of the essential spectrum that fail to be regular in the following sense.
	
	\begin{definition}[Regular spectral point and spectral singularity]\label{def:point_spectral_regulier_classique_pour_H}
		Let $\lambda \in \Lambda := \sigma_\mathrm{ess}(H) = \sigma_\mathrm{ess}(H_0)$.
		\begin{enumerate}[label=(\roman*)]
			\item We say that $\lambda$ is an \emph{outgoing} (respectively \emph{incoming}) \emph{regular spectral point} of $H$ if $\lambda$ is not an accumulation point of eigenvalues lying in the set $\lambda \pm i(0, \infty)$, and if the limit
			\begin{equation}\label{eq:def_reg_spec_pt}
				C \Res_H(\lambda \pm i0^+) C W := \lim_{\varepsilon \to 0^+} C \Res_H(\lambda \pm i\varepsilon) C W
			\end{equation}
			exists in the norm topology of $\mathcal{B}(\Hi)$. If this limit does not exist, then $\lambda$ is said to be an \emph{outgoing} (respectively \emph{incoming}) spectral singularity of $H$.
			\item We say that $\lambda$ is a \emph{regular spectral point} of $H$ if it is both an outgoing and an incoming regular spectral point. Otherwise, $\lambda$ is called a \emph{spectral singularity} of $H$.
		\item We say that infinity is an outgoing/incoming regular spectral point of $H$ if there exists $m>0$ such that for all $\lambda>m$, $\lambda$ is an outgoing/incoming regular spectral point and if the map 
			\begin{equation*}
			    [m,\infty)\ni\lambda\mapsto C\Res_H(\lambda\pm i0^+)CW 
			\end{equation*}
			is bounded for the topology of the norm of operator. If one of this condition does not hold, we say that infinity is an \emph{outgoing} (respectively \emph{incoming}) spectral singularity of $H$. 
		\end{enumerate}
	\end{definition}

	The notion of spectral singularity is closely related to the concept of spectral projection for non-self-adjoint operators introduced in \cite{Sc60_01}. Assuming a limiting absorption principle for $H_0$, a characterization of spectral singularities has been provided in \cite{FaFr23}. This notion also plays a fundamental role in the study of the dynamics of solutions to the Schrödinger equation governed by a non-self-adjoint Hamiltonian. For instance, Faupin and Fröhlich \cite{FaFr18_01} show that the dissipative wave operators are complete if and only if $H$ has no spectral singularities, while Faupin and Nicoleau \cite{FaNi18_01} demonstrate that the dissipative scattering matrix fails to be invertible at spectral singularities. Finally we refer to \cite{Fr24_05} for a construction of wave operators (named "regularized wave operators") taking into account spectral singularities in the non-dissipative case.

	\vspace{0.1cm}\noindent
	In the sequel, we denote by $\Lambda_\mathrm{reg}$ the set of regular spectral points of $H$. It is of particular interest to understand how rapidly the weighted resolvent of $H$ diverges as the spectral parameter approaches a spectral singularity from the upper or lower half-plane. This leads to the notion of the \emph{order} of a spectral singularity.

	\begin{definition}[Order of a spectral singularity]\label{def:order spectral singularity}
		Let $\lambda \in \sigma_\mathrm{ess}(H)$ be an outgoing or incoming spectral singularity of $H$. We say that $\lambda$ is a spectral singularity of \emph{finite order} if there exist an integer $n \in \N$ and $\varepsilon > 0$ such that
		\begin{equation}\label{eq:def order of spectral singularity}
			\sup_{z \in D(\lambda, \varepsilon) \cap \C^\pm} \left|\lambda - z\right|^n \| C \Res_H(z) C W \|_{\mathcal{B}(\Hi)} < \infty.
		\end{equation}
		Otherwise, $\lambda$ is said to be an outgoing or incoming spectral singularity of \emph{infinite order}. If $\lambda$ is an outgoing or incoming spectral singularity of finite order, we define its \emph{order} as the smallest integer $n$ for which \eqref{eq:def order of spectral singularity} holds.\\
	We say that infinity is a outgoing or incoming spectral singularity of finite order if there exists an integer $n$, $\varepsilon_0>0$, $m>0$ and $z_0\in\rho(H)\backslash \R$ such that 
		\begin{equation}\label{eq:def order infinity}
		    \sup_{\substack{\re(z)>m\\ |\pm\im(z)|<\varepsilon_0}}|z-z_0|^{-n}\|C\Res_H(z)CW\|_{\mathcal{B}(\Hi)}<\infty
		\end{equation}
		Otherwise, infinity is said to be an outgoing or incoming spectral singularity of \emph{infinite order}. If infinity is an outgoing or incoming spectral singularity of finite order, we define its \emph{order} as the smallest integer $n$ for wich \eqref{eq:def order infinity} holds.
		
	\end{definition}
	
	In other words, a spectral singularity is of finite order if the weighted resolvent of $H$ can be regularized by a polynomial factor in a neighborhood of the singularity within the upper or lower half-plane and a singularity at infinity can be regularized by a rational function in a neighborhood of infinity in the upper or lower half-plane. As we will recall in Section \ref{sec:Schro}, for the Schr\"odinger operator spectral singularities are related to real resonances (see Section \ref{subsec:ss=Res}).

				\subsection{Resolvent estimates}\label{subsub:Resolvent estimate}
	Here sufficient conditions on spectral singularities and eigenvalues are given so that $H$ satisfies the resolvent estimate of Hypothesis~\ref{hyp:estimate resolvente}.

	\begin{prop}\label{prop:estime resolvent singularities}
		Assume that $H$ satisfies Hypothesis~\ref{hyp:eigenvalues} on an open interval $I$ and that its closure $\overline{I}$ contains a finite number of spectral singularities of finite order. 
		Then Hypothesis~\ref{hyp:estimate resolvente} holds for $H$ on $I$. 
	\end{prop}
	
    %\begin{remark}
   % \begin{itemize}
    %\item The estimate \eqref{eq:estime resolvente} is related to Davies's assumption in \cite[(H3)]{Da95_08}. 
    %\item If we denote $\nu_1,\ldots,\nu_n$ the order of the $n$ spectral singularities of $H$ and $\nu_\infty$ the order of the spectral singularity at infinity, then {\clb we can take}
   % \begin{equation*}
   %     n_0=\max(\nu_k)_{k\in\lbrace 1,\ldots,n\rbrace\cup\{\infty\}}+1
   % \end{equation*}
   % {\clr F : il faut dire qu'il y a qu'un nombre fini de singularités spectrales ? Est-ce correct ?}
  %  \item We remark that if $H$ has no spectral singularity, {\clb then we can take $n_0=2$}.
  %  \end{itemize}
  %  \end{remark}
    
	\begin{proof}
	If $I$ is bound from above, it is sufficient to prove the result for $I$ bounded because below the essential spectrum Hypothesis~\ref{hyp:estimate resolvente} is always true (see Remark \ref{rqHyp}). By combining Hypothesis~\ref{hyp:eigenvalues} with \eqref{eq:def order of spectral singularity}, near each spectral singularity there exists $c_1>0$ and $n_1\in \N$ such that 
\begin{equation}\label{eq:global estimate out/in}
			\| C \Res_H(z) C W \|_{\mathcal{B}(\Hi)} \leq c_1 |\im(z)|^{-n_1},
		\end{equation}	
	and near regular points we can take $n_1=0$. Then having a finite number of spectral singularities, by compactness of $\overline{I}$, \eqref{eq:global estimate out/in} holds on $S_a(I)$ for some $a>0$ and with the maximum of all the finite order instead of $n_1$. 
	If we denote $\nu_1,\ldots,\nu_n$ the order of the $n$ spectral singularities of $H$ in $I$, we deduce Hypothesis~\ref{hyp:estimate resolvente} on $I$ with 
	%and $\nu_\infty$ the order of the spectral singularity at infinity, then {\clb we can take}
    \begin{equation*}
    n_I=\max(\nu_k)_{k\in\lbrace 1,\ldots,n\rbrace}+1
  % n_I=\max(\nu_k)_{k\in\lbrace 1,\ldots,n\rbrace\cup\{\infty\}}+1
   \end{equation*}
   by using the resolvent identity 
   \begin{equation}\label{eq:resolvent identity}
			\Res_H(z) = \Res_0(z) - \Res_0(z)V\Res_0(z) + \Res_0(z) C W C \Res_H(z) C W C \Res_0(z),
		\end{equation}
   and the resolvent estimate for the self-adjoint operator $H_0$: $\| \Res_0(z) \|_{\mathcal{B}(\Hi)} \leq |\im(z)|^{-1}$.
   
   If $I$ has no upper bound, it is sufficient to prove the result on an interval $[s_1, + \infty)$ having only a spectral singularity at infinity (on $I \cap (- \infty, s_1]$ we apply the previous argument). 
   Hypothesis~\ref{hyp:eigenvalues}  and with \eqref{eq:def order infinity} give 
   \begin{equation*}
		    \|C\Res_H(z)CW\|\leq c|\re(z)|^{n},
		\end{equation*}
   on $S_a([s_1, + \infty))$ for some $n \in \N$, $c>0$ and $a>0$. We conclude by using again \eqref{eq:resolvent identity}.
   %\begin{equation*}
	%		\| \Res_0(z) \|_{\mathcal{B}(\Hi)} \leq |\Im(z)|^{-1}.
	%	\end{equation*}
	%	Using the resolvent identity, we write

	%	\begin{equation*}
	%	    \|C\Res_H(z)CW\|\leq c|\re(z)|^{n_3}.
	%	\end{equation*}
	%	As above with \eqref{eq:resolvent identity}, there exists $c_3>0$ such that for all $z\in V_{a_2,b_3}$, 
	%	\begin{equation}\label{eq:estimresol3}
	%	    \|C\Res_H(z)CW\|\leq c_3|\Im(z)^{-2}|.\langle\re(z)\rangle^{n_3}.
	%	\end{equation}
	%	Finally taking the smallest of $b_1$, $b_2$ and $b_3$, combining \eqref{eq:estimresol1}, \eqref{eq:estimresol2} and \eqref{eq:estimresol3} and adjusting the constants we obtain \eqref{eq:estime resolvente}.
	\end{proof}

\section{SSF for Schrödinger Operators with Complex-Valued Potential}\label{sec:Schro}
			
	In this section	we consider the self-adjoint operator 
\[
H_0 = -\Delta \quad \text{on } L^2(\R^3),
\]
with domain 	$\mathcal{D}(H_0)= H^2(\R^3):= \{ u \in L^2(\R^3); \; \Delta u \in L^2(\R^3)\} $ %(Sobolev space)
and $H$ 
is the Schrödinger operator $H:= H_0 + V$ where $V$ is a multiplication operator by a function $V \in L^\infty(\R^3; \C)$. The relative compactness is guaranteed by the assumption:
\begin{equation}\label{ShortR}
    |V(x)| \;\le\; M\,\langle x\rangle^{-\delta}, 
\qquad \forall x \in \R^3, \qquad \delta >0, 
\end{equation}
for some constant \(M>0\) and the essential spectrum is given by $\Lambda:=\sigma_\mathrm{ess}(H) = \sigma_\mathrm{ess}(H_0)=[0, + \infty)$.
 We focus on the three-dimensional setting but this section could easily be extended to any odd dimension.
 We show how to apply the previous results to this non-self-adjoint Schrödinger operator; in particular, we recall the link between resonances and spectral singularities, and then we give regularity properties of the spectral shift function outside of spectral singularities. Finally, we show that the behavior at high energies is very close to the self-adjoint case. The existence of the spectral shift function will be given under 
  the short-range condition $\delta > 3$

We have $V = C W C$ by considering $C$ the multiplication operator by $C(x):=\langle x \rangle^{-\delta/2}$ and $W := \langle \cdot \rangle^{\delta} V $. If \(V\) is a compactly supported potential, we can also define \(C\), (resp. $W$) as the multiplication operator by \(\rho\), (resp. $V$) with \(\rho\) a smooth compactly supported function  such that $\rho \equiv 1$ on $\supp(V)$.
%\[
%\rho(x) \equiv 1 \quad \text{on } \supp(V),
%\qquad \text{and hence} \qquad \rho\,v = v.
%\]

%This framework covers a broad class of physically relevant potentials and ensures sufficient decay to carry out the spectral analysis below. The restriction to dimension three (or more generally to odd dimensions) is crucial, as it guarantees improved resolvent behavior near the continuous spectrum, in particular through the limiting absorption principle. 

%Finally note that in the following all the results available for bounded short range potentials hold for bounded compactly supported potentials. 

	\subsection{Spectral singularities and resonances}\label{subsec:ss=Res}
	
	For Schrödinger operators, the notion of spectral singularity is closely related to that of resonances. We present here two related notions, depending on the decay of the potential $V$.
	
	\subsubsection{Short-range condition} We define the weighted space
	\begin{align*}
		L^2_{\delta} := \left\{ f : \R^3 \to \C \mid x \mapsto \langle x \rangle^{\delta} f(x) \in L^2(\R^3) \right\},
	\end{align*}
	In \cite{FaFr23}, Faupin and Frantz established the following Proposition:
	\begin{prop}\cite[Proposition 3.10]{FaFr23}
		Suppose that $V$ is a complex-valued potential satisfying the short-range condition \eqref{ShortR}  with $\delta> 1$.
		Then for all $\lambda>0$, the following conditions are equivalent:
		\begin{enumerate}[label=(\roman*)]
			\item\label{it:spec_sing_schr} $\lambda$ is an outgoing/incoming spectral singularity of $H$ in the sense of Definition \ref{def:point_spectral_regulier_classique_pour_H},
			\item\label{it:res_state_schr} There exists $\Psi\in L^2_{-\frac{\delta}{2}}$, with $\Psi\neq0$, such that
			\begin{equation*}
				(-\Delta+V(x)-\lambda)\Psi=0.
			\end{equation*} 
		\end{enumerate}
	\end{prop}
	
	\vspace{0.1cm}\noindent
	If the function $\Psi$ from \ref{it:res_state_schr} lies in $L^2(\mathbb{R}^3)$, then $\lambda$ is an eigenvalue of $H$. Otherwise, $\lambda$ is referred to as a real resonance, associated with a resonant state $\Psi \notin L^2(\mathbb{R}^3)$. Such a state satisfies the outgoing/incoming Sommerfeld radiation condition:
	\begin{equation*}
		 \Psi(x) = |x|^{1/2} e^{\pm i \lambda^{1/2} |x|} \left( a\left( \frac{x}{|x|} \right) + o(1) \right), \quad |x| \to \infty,
	\end{equation*}
	where $a \in L^2(S^2)$ and $a \neq 0$. Moreover, if the operator $H$ is dissipative, (in the sense that $\mathrm{Im}( \langle u , H u \rangle ) \le 0$ for all $u$ in the domain of $H$),   $H$ cannot have outgoing spectral singularities in $(0,\infty)$ (see \cite[Corollary 3.2]{Wa12}). But, in general, spectral singularities may still occur. As shown in \cite[Remark 5.4]{Wa12}, for any $\lambda > 0$, one can construct a smooth, compactly supported potential $V$ such that $\lambda$ is an incoming spectral singularity of $H$ in the dissipative setting.

	\subsubsection{Compactly supported potential}\label{sec:spectralsing_compsupp}

	If $V \in L^\infty_\mathrm{c}(\R^d, \C)$, with $d \geq 3$ odd, then the map
	\begin{equation*}
		\left\{ z \in \C \mid \im(z) > 0 \right\} \ni z \mapsto (H - z^2)^{-1} : L^2(\R^d, \C) \to L^2(\R^d, \C)
	\end{equation*}
	is meromorphic and admits a meromorphic extension to the whole complex plane as a map
	\begin{equation}\label{eq:mero ext}
		\C \ni z \mapsto F(z) : L^2_\mathrm{c}(\R^d) \to L^2_\mathrm{loc}(\R^d),
	\end{equation}
	where
	\begin{eqnarray*}
		L^2_\mathrm{c}(\R^d) &:=& \{ u \in L^2(\R^d) \mid \supp(u) \text{ is compact} \},  \\
		L^2_\mathrm{loc}(\R^d) &:=& \{ u : \R^d \to \C \mid u \in L^2(K) \text{ for all compact } K \subset \R^d \}.
	\end{eqnarray*}
	The poles of the meromorphic extension in \eqref{eq:mero ext} are called \emph{resonances} of $H$.

	\vspace{0.1cm}\noindent
	One can verify that any real resonance $\pm \lambda_0$ of $H$, with $\lambda_0 \geq 0$, corresponds to an outgoing/incoming spectral singularity $\lambda_0^2$ in the sense of Definition~\ref{def:point_spectral_regulier_classique_pour_H}. Moreover, $H$ has only finitely many spectral singularities, and the order of a spectral singularity coincides with the multiplicity of the corresponding resonance pole; see \cite[Theorem 3.8]{DyZw19_01}.
	
	\vspace{0.1cm}\noindent
	We refer the reader to \cite{DyZw19_01} and the references therein for an overview of the resonance theory for Schrödinger operators, and to \cite[Section 3.3.1]{FaFr23} for a more detailed comparison between the notions of resonances and spectral singularities considered in this paper.
	
	Finally note that in both cases, $\infty$ is an outgoing and an incoming regular spectral point. This is a consequence of point (3) of Proposition \ref{prop:proprieteT_0}. (See e.g the proof of Proposition \ref{prop:high energy asymptotic} about the high energy asymptotic of the derivative of the SSF for more details).

\subsection{Preliminary results}

Here we recall some useful results on the free resolvent, \(\Res_0(z)\), that  will be used to derive qualitative results on the spectral shift function. 
%\vspace{0.2cm}\noindent
We define $\sqrt{z}$ on $\C \setminus \R_+$ by requiring that $\im(\sqrt{z}) > 0$. 
Hence, for all $\lambda \in (0,\infty)$, we have
\[
\lim_{\varepsilon \to 0^+} \sqrt{\lambda \pm i\varepsilon} = \pm \sqrt{\lambda}.
\]

For \(z \in \C \setminus \R^+\), set 
\[
T_0(z) := C\,\Res_0(z)\,C\,W,
\]
and recall that, for such $z$, the integral kernel of \(T_0(z)\) is the function 
\[
K_0(z)(x,y)
:= \frac{1}{4\pi}\, C(x)\, \frac{e^{i\sqrt{z}\,|x-y|}}{|x-y|}\, C(y)\, W(y),
\qquad (x,y) \in \R^3 \times \R^3.
\]

\vspace{0.2cm}\noindent
The well-known next proposition collects some useful properties of the operator \(T_0(z)\), (see \cite{Ya10_01} for instance).

\begin{prop}\label{prop:proprieteT_0}
Assume that \(V\) is a short-range potential (i.e. satisfies \eqref{ShortR}) with \(\delta > 3\).
Then, for all \(z \in \C \setminus \R^+\), the integral kernel \(K_0(z)\) belongs to \(L^2(\R^3 \times \R^3)\). 
In particular, \(T_0(z)\) is a Hilbert--Schmidt operator. 
Moreover:
\begin{enumerate}
    \item For all \(\lambda \in \Lambda\), the limits 
    \[
    T_0(\lambda \pm i0^+) := \lim_{\varepsilon \to 0^+} T_0(\lambda \pm i\varepsilon)
    \]
    exist in the Hilbert--Schmidt topology.
    
    \item The maps 
    \(\Lambda \ni \lambda \mapsto T_0(\lambda \pm i0^+)\)
    are continuous in the Hilbert--Schmidt topology.
     \item For any $\delta>0$, there exist constants $c>0$ and $C>0$ such that, 
for all $z \in \C \setminus \R^+$ with $|z|>c$, one has 
\begin{equation}\label{eq:resolvent_poids_haute_energie}
    \|T_0(z)\|_{\mathcal{B}(\Hi)} 
    \;\le\; C\,|z|^{-\frac{1}{2}+\delta}.
\end{equation} %(Chapitre 7 Yafaev)
\end{enumerate}
\end{prop}

\vspace{0.2cm}\noindent
The $k$-derivative with respect to $z$ of $T_0$ is given by
\begin{equation*}
    T_{0}^{(k)}(z) := (-1)^kC\Res_0(z)^{k+1}CW.
\end{equation*}
For $k=1$, its integral kernel is given for all $(x,y)\in\R^3\times\R^3$ by 
\begin{equation*}
    K_0^{(1)}(z)(x,y)=\frac{i}{8\pi\sqrt{z}}C(x)e^{ i\sqrt{z}|x-y|}C(y)W(y),\quad \pm \im(z)>0.
\end{equation*}

\vspace{0.2cm}\noindent
The following Proposition collects some useful properties of the operators \(T_0^{(k)}(z)\),
(see \cite{JeKa79_09}, Lemma 2.2):
\begin{prop}\label{prop:proprietedT_0}
Assume that \(V\) satisfies \eqref{ShortR}) with \(\delta > 3\). Then:
\begin{enumerate}
    \item For all \(k \in \N \setminus \{0\}\) and all \(z \in \C \setminus \R^+\), the operator \(T_0^{(k)}(z)\) is of trace class.
    
    \item If \(\delta > 2k + 1\), then for all \(\lambda \in (0, +\infty)\), the limits 
    \[
    T_0^{(k)}(\lambda \pm i0^+) := \lim_{\varepsilon \to 0^+} T_0^{(k)}(\lambda \pm i\varepsilon)
    \]
    exist in the Hilbert--Schmidt topology.  
    \item Moreover, the map 
    \[
    (0, +\infty) \ni \lambda \longmapsto T_0(\lambda \pm i0^+)
    \]
    belongs to \(\mathscr{C}^k((0,+\infty), \mathcal{L}_2(\Hi))\), and one has
    \[
    T_0^{(k)}(\lambda \pm i0^+) = \partial_\lambda^{(k)} T_0(\lambda \pm i0^+).
    \]
\end{enumerate}
\end{prop}

\vspace{0.2cm}\noindent
We now derive explicit trace formulas for \(T_0^{(1)}(z)\) and \(T_0^{(1)}(z)T_0(z)\), 
which will play a key role in the computation of the spectral shift function.

\begin{lemma}
Assume that \(V\) is a short-range potential with exponent \(\delta > 3\), and let \(\lambda \in (0,+\infty)\). Then:
\begin{enumerate}
    \item One has
    \begin{equation}\label{eq:diff Trace2}
        \lim_{\varepsilon \to 0^+} 
        \Tr\!\left(T_0^{(1)}(\lambda + i\varepsilon) 
        - T_0^{(1)}(\lambda - i\varepsilon)\right)
        \;=\;
        \frac{i}{4\pi \sqrt{\lambda}}
        \int_{\R^3} V(x)\, \mathrm{d}x.
    \end{equation}

    \item Moreover,
    \begin{equation}\label{eq: trace T_0^1T_0}
        \Tr\!\bigl(T_{0}^{(1)}(\lambda \pm i0^{+})\, T_{0}(\lambda \pm i0^{+})\bigr)
        \;=\;
        \pm\, \frac{i}{32\pi^{2}\sqrt{\lambda}}
        \int_{\R^{6}} 
        V(x)\,V(y)\,
        \frac{e^{ \pm 2i\sqrt{\lambda}\,\lvert x-y\rvert}}{\lvert x-y\rvert}\,
        \mathrm{d}x\,\mathrm{d}y.
    \end{equation}
\end{enumerate}
\end{lemma}

\begin{proof}
\noindent\textbf{(1)}  
For all \(\varepsilon > 0\), the operator \(T_0^{(1)}(\lambda \pm i\varepsilon)\) is of trace class, and its trace is given by
\[
\Tr\bigl(T_0^{(1)}(\lambda \pm i\varepsilon)\bigr)
\;=\;
\int_{\R^3} K_0^{(1)}(\lambda \pm i\varepsilon)(x,x)\,\mathrm{d}x
\;=\;
\frac{i}{8\pi \sqrt{\lambda \pm i\varepsilon}}
\int_{\R^3} V(x)\,\mathrm{d}x.
\]
Subtracting the two expressions and taking the limit as \(\varepsilon \to 0^+\) yields~\eqref{eq:diff Trace2}.

\vspace{0.2cm}
\noindent\textbf{(2)}  
Let \(z \in \C \setminus \R^+\).  
The integral kernel of the product \(T_0^{(1)}(z)\,T_0(z)\), denoted \(K_0^{(0,1)}(z)\), is given by
\[
K_0^{(0,1)}(z)(x,y)
=
\frac{i}{32\pi^{2}\sqrt{z}}
\int_{\R^3}
C(x)\, e^{i\sqrt{z}\,|x-t|}\, V(t)\,
\frac{e^{i\sqrt{z}\,|t-y|}}{|t-y|}\, C(y)\, W(y)\,\mathrm{d}t.
\]
Since \(T_0^{(1)}(z)T_0(z)\) is of trace class, we have
\[
\Tr\bigl(T_0^{(1)}(z)\,T_0(z)\bigr)
=
\frac{i}{32\pi^{2}\sqrt{z}}
\int_{\R^{6}}
V(x)\,V(y)\,
\frac{e^{ 2i\sqrt{z}\,|x-y|}}{|x-y|}\,
\mathrm{d}x\,\mathrm{d}y.
\]
Finally, taking \(z = \lambda \pm i\varepsilon\) and letting \(\varepsilon \to 0^+\) yields~\eqref{eq: trace T_0^1T_0}.
\end{proof}

\vspace{0.2cm}
\noindent
Before stating the next lemma, let us recall that, by the resolvent identity,  
for every $z \in \rho(H) \cap \rho(H_0)$ the operator $\Id + T_0(z)$ is invertible and we have
\begin{equation}\label{invertible}
   \bigl[\Id + T_0(z) \bigr]^{-1}  \;=\; \Id - C \Res_H(z) C W.
\end{equation}

\vspace{0.2cm}
\noindent
We conclude this section with the following useful identity :

\begin{lemma}\label{lemTrR0}
Suppose that $V$ is a short-range potential with decay exponent $\delta > 3$. 
Then, for every $z \in \rho(H) \cap \rho(H_0)$, the difference 
\(\Res_H(z) - \Res_0(z)\) is a trace-class operator. Moreover, the following identities hold:
\begin{align}
    \Tr\!\bigl(\Res_H(z) - \Res_0(z)\bigr)
    &= \Tr\!\left(T_0^{(1)}(z)\,[I + T_0(z)]^{-1}\right), 
    \label{eq:Id_Trace1} \\[0.4em]
    &= \Tr\!\left(T_0^{(1)}(z)\right) 
       - \Tr\!\left(T_0^{(1)}(z)\,[I + T_0(z)]^{-1} T_0(z)\right),
    \label{eq:Id_Trace2} \\[0.4em]
    &= \Tr\!\left(T_0^{(1)}(z)\right) 
       - \Tr\!\left(T_0^{(1)}(z)\,T_0(z)\right) 
       + \Tr\!\left(T_0^{(1)}(z)\,C\,\Res_H(z)\,C\,W\,T_0(z)\right).
    \label{eq:Id_Trace3}
\end{align}
\end{lemma}

\begin{proof}
Let $z \in \rho(H) \cap \rho(H_0)$. 
Since $\delta > 3$, the resolvent identity implies that 
\(\Res_H(z) - \Res_0(z)\) is a trace-class operator 
(see, for instance,~\cite[Theorem~XI.21]{ReSi80_01}).
Applying again the resolvent identity successively (twice and three times, respectively), we obtain
\begin{align*}
    \Res_H(z) - \Res_0(z)
    &= -\,\Res_0(z)\,V\,\Res_0(z) 
       \;+\; \Res_0(z)\,V\,\Res_H(z)\,V\,\Res_0(z) \\[0.3em]
    &= -\,\Res_0(z)\,V\,\Res_0(z) 
       \;+\; \Res_0(z)\,V\,\Res_0(z)\,V\,\Res_0(z) 
       \;-\; \Res_0(z)\,V\,\Res_H(z)\,V\,\Res_0(z)\,V\,\Res_0(z).
\end{align*}
Next, by a clever use of the cyclicity of the trace, we obtain
\begin{align*}
    \Tr\!\bigl(\Res_H(z) - \Res_0(z)\bigr) 
    &= \Tr\!\left(T_0^{(1)}(z)\,[\Id - C\,\Res_H(z)\,C\,W]\right) \\[0.3em]
    &= \Tr\!\left(T_0^{(1)}(z)\,[\Id + T_0(z)]^{-1}\right),
\end{align*}
by~\eqref{invertible}, which proves~\eqref{eq:Id_Trace1}. The identities~\eqref{eq:Id_Trace2} and \eqref{eq:Id_Trace3} follow immediately.
\end{proof}

\subsection{Existence of the SSF}
The following theorem provides sufficient conditions for the existence of the spectral shift function 
\(\xi(\cdot; -\Delta + V(x), -\Delta)\).

\begin{prop}
Let $V$ be a short-range potential with decay exponent $\delta > 3$, 
and assume that the operator $H$ admits only finitely many eigenvalues and finitely many spectral singularities, each of finite order. 
Then there exists \( c > 0 \) such that the spectral shift function
\[
\xi(\cdot;\, -\Delta,\, -\Delta + V) 
\;:=\; \xi\!\left(\cdot;\, (H + c)^{-1},\, (H_0 + c)^{-1}\right)
\]
is well defined and does not depend on the particular choice of \(c>0\).
\end{prop}

\begin{proof}
As \( H \) has a finite number of eigenvalues and a finite number of spectral singularities of finite order, it follows from Proposition~\ref{prop:estime resolvent singularities} that Hypotheses \ref{hyp:eigenvalues} are \ref{hyp:estimate resolvente} is satisfied on $\R$. 
%{\clb F : pour moi ce n'est pas clair. Comment contrôler la partie réelle de la Proposition 4.1 ? N : On peut montrer qu'il n'y a pas de singularité spectral à l'infini.}
%\vspace{0.1cm}\noindent
Furthermore, since \( V \) is bounded and \( H \) has only finitely many eigenvalues, there exists \( c > 0 \) such that \( -c \in \rho(H) \), and \( a > 0 \) such that
\begin{equation*}
	\sigma(H) \subset \left\{ z \in \C \,\middle|\, \re(z) > -c \text{ and } |\im(z)| < a \right\}.
\end{equation*}
\vspace{0.1cm}\noindent
Moreover, as $\delta>3$, thanks to Lemma \ref{lemTrR0}, \( (H_0 + c)^{-1} - (H + c)^{-1} \) is a trace class operator.
Thus, condition~\eqref{it2:existenceSSF} in Proposition~\ref{prop:ssfRelTr} is also satisfied. We conclude that the spectral shift function
\[
\xi(\cdot, -\Delta, -\Delta + V) := \xi(\cdot, (H + c)^{-1}, (H_0 + c)^{-1})
\]
exists and does not depend on the particular choice of $c$.
\end{proof}
			
As a consequence of the trace-class property of $(H - c)^{-1} - (H_0 - c)^{-1}$ 
and Proposition~\ref{prop:repSSFRT} with $m=1$, the derivative $\xi'(H, H_0, \cdot)$ admits the following representation:

\begin{prop}\label{derivative}
Let $V$ be a short-range potential with decay exponent $\delta > 3$, 
and assume that the operator $H$ has only finitely many eigenvalues and finitely many spectral singularities, each of finite order. 
Then, in the sense of distributions, one has
\begin{equation*}
    \xi'(H, H_0; \lambda)
    \;=\;
    \frac{1}{2\pi i}
    \lim_{\varepsilon \to 0^+} 
    \Tr\!\Big(
        \Res_H(\lambda + i\varepsilon) - \Res_0(\lambda + i\varepsilon) \Big)
        \;-\; \Tr\!\Big(
        \Res_H(\lambda - i\varepsilon) - \Res_0(\lambda - i\varepsilon)
    \Big).
\end{equation*}
\end{prop}

\begin{remark}
The assumption that the operator has only finitely many eigenvalues may appear restrictive. 
Nevertheless, there exist sufficient conditions on the potential \(V\) ensuring that 
\( -\Delta + V \) admits only finitely many eigenvalues. 
For instance, in odd dimensions, Frank, Laptev, and Safronov~\cite{FrLaSa16_01} showed that if \(V\) decays exponentially at infinity, 
then the operator \( -\Delta + V \) has finitely many eigenvalues.
\end{remark}

\subsection{Regularity of the Spectral Shift Function}

In this section, we show that the spectral shift function extends to a regular function 
in a neighborhood of any regular spectral point. Since the spectral singularities are isolated, this will allow us to define the SSF locally away from these singularities.
%Recall that $\Lambda$ denotes the essential spectrum of both $H$ and $H_0$. 
We have the following proposition:

\begin{prop}\label{Prop:regularity}
Assume that $V$ is a short-range potential with decay exponent $\delta > 2k + 1$, 
where $k \in \N^*$. 
Assume also that $H$ has only finitely many eigenvalues and finitely many spectral singularities, 
each of finite order. Let $\lambda \in (0, +\infty)$ be a regular spectral point of $H$. 
Then the spectral shift function $\xi(\cdot; H, H_0)$ is of class $\mathscr{C}^{k+1}$ 
in a neighborhood of~$\lambda$.
\end{prop}

\begin{proof}
It follows from point~(2) of Proposition~\ref{prop:proprieteT_0} together with~\cite[Proposition~4.7]{FaFr23} 
that there exists a neighborhood $\omega_\lambda$ of $\lambda$ free of spectral singularities.  
Using~\eqref{eq:Id_Trace2}, we obtain, in the sense of distributions and for all $\mu \in \omega_\lambda$,
\begin{align}
    \xi'(\mu; H, H_0) 
    &= \frac{1}{2\pi i} \lim_{\varepsilon \to 0^+} \Big(
        \Tr\big(T_0^{(1)}(\mu + i\varepsilon)\big)
        - \Tr\big(T_0^{(1)}(\mu - i\varepsilon)\big)
        \label{eq:regularity1}\\
    &\quad - \Tr\!\big(T_0^{(1)}(\mu + i\varepsilon)\,[I + T_0(\mu + i\varepsilon)]^{-1}\,T_0(\mu + i\varepsilon)\big)
        \nonumber\\
    &\quad + \Tr\!\big(T_0^{(1)}(\mu - i\varepsilon)\,[I + T_0(\mu - i\varepsilon)]^{-1}\,T_0(\mu - i\varepsilon)\big)
        \Big).
    \label{eq:regularity2}
\end{align}
Since every $\mu \in \omega_\lambda$ is a regular spectral point of~$H$, 
\cite[Proposition~4.7]{FaFr23} ensures that the operators $I + T_0(\mu \pm i0^+)$ are invertible, and that
\[
    \lim_{\varepsilon \to 0^+} [I + T_0(\mu \pm i\varepsilon)]^{-1} 
    = [I + T_0(\mu \pm i0^+)]^{-1}
\]
exist in $\mathcal{B}(\mathcal{H})$. Moreover, the mappings
\[
    \operatorname{Inv}:\; \mu \in \omega_\lambda \longmapsto [I + T_0(\mu \pm i0^+)]^{-1}
\]
are continuous.  The $\mathscr{C}^k$-regularity of $\operatorname{Inv}$ on $\omega_\lambda$ follows from the differentiability of the inversion map and from point~(2) of Proposition~\ref{prop:proprietedT_0}.  
Combining this fact with points~(1) and~(2) of Proposition~\ref{prop:proprieteT_0} and point~(2) of Proposition~\ref{prop:proprietedT_0}, we deduce that
\[
\Tr\!\left(T_0^{(1)}(\mu \pm i0^+)\,[I + T_0(\mu \pm i0^+)]^{-1}\,T_0(\mu \pm i0^+)\right)
:= \lim_{\varepsilon \to 0^+} 
\Tr\!\left(T_0^{(1)}(\mu \pm i\varepsilon)\,[I + T_0(\mu \pm i\varepsilon)]^{-1}\,T_0(\mu \pm i\varepsilon)\right)
\]
exists, and that the mappings
\[
    \mu \in \omega_\lambda \longmapsto 
    \Tr\!\left(T_0^{(1)}(\mu \pm i0^+)\,[I + T_0(\mu \pm i0^+)]^{-1}\,T_0(\mu \pm i0^+)\right)
\]
are of class~$\mathscr{C}^k$.  Finally, by~\eqref{eq:diff Trace2}, the mapping
\[
    \mu \in \omega_\lambda \longmapsto 
    \Tr\!\left(T_0(\mu + i0^+) - T_0(\mu - i0^+)\right)
\]
is well defined and of class~$\mathscr{C}^k$.  
This completes the proof.
\end{proof}

\subsection{Asymptotics near a Spectral Singularity}

In this section, we derive the asymptotic behavior of the spectral shift function (SSF) 
in a neighborhood of a spectral singularity, in the case of a compactly supported and bounded potential.

\vspace{0.2cm}\noindent
Assume that $V \in L^\infty_{\mathrm{c}}(\R^3, \C)$, the space of essentially bounded functions with compact support. 
Then there exists a compactly supported function $\rho$ such that 
\begin{equation*}
    \rho \equiv 1 \quad \text{on } \supp(V),
    \qquad \rho\,V = V.
\end{equation*}

\vspace{0.2cm}\noindent
As mentioned in Subsection~\ref{sec:spectralsing_compsupp}, 
the outgoing and incoming spectral singularities $\lambda_0>0$ of $H$ 
correspond to the real poles $\pm \sqrt{\lambda_0}$ 
of the meromorphic continuation to $\C$ of the map
\begin{equation*}
     \C_+ \ni z \longmapsto (H - z^2)^{-1} : L^2_{\mathrm{c}}(\R^3) \to L^2_{\mathrm{loc}}(\R^3).
\end{equation*}
Equivalently (see~\cite[Theorem~3.8]{DyZw19_01}), they correspond to real poles of the meromorphic continuation  to $\C$ of the map
\begin{equation*}
  \C_+ \ni z \longmapsto  \rho\,(H - z^2)^{-1}\rho : L^2(\R^3) \to L^2(\R^3).
\end{equation*}
We denote by $F(z)$ this meromorphic continuation. 

\vspace{0.2cm}\noindent
To fix ideas, suppose that $\lambda_0>0$ is an outgoing spectral singularity of $H$ of order~$\nu_0$. 
It then follows that $\sqrt{\lambda_0}$ is a pole of $F(z)$. 
According to~\cite{DyZw19_01}, there exist finite-rank operators $A_{-j}$, $1 \le j \le \nu_0$, 
and an operator-valued function 
$z \mapsto A_0(z)$, holomorphic in a complex neighbourhood of~$\sqrt{\lambda_0}$, such that 
\begin{equation}\label{eq:DL singularity}
    F(z)
    \;=\; \sum_{j=1}^{\nu_0} \frac{A_{-j}}{(z^2 - \lambda_0)^{j}} \;+\; A_0(z),
\end{equation}
in a complex neighbourhood of~$\sqrt{\lambda_0}$.

\vspace{0,2cm}\noindent
We recall the following standard expansion, which is a direct consequence of the theory of distributions associated with analytic functions 
(see, e.g., \cite[Chap.~III, §3.5]{GelfandShilov64}).

\begin{lemma}[Multiplication of a principal part by a smooth function]
\label{lem:product_principalpart}
Let $j\in\mathbb{N}^*$, $\lambda_0>0$, and let $g$ be holomorphic in $V_{\lambda_0}\cap \C_\pm$ and admit a smooth extension up to the real axis, denoted by $g_\pm$, where $V_{\lambda_0}$ denotes a complex neighbourhood of $\lambda_0$. 
Then, in the sense of distributions in a neighbourhood of $\lambda_0$ on the real line, one has
\begin{equation}\label{eq:prod_dist_limit_pm}
\lim_{\varepsilon\to 0^+}
\frac{g(\lambda \pm i\varepsilon)}{(\lambda \pm i\varepsilon - \lambda_0)^j}
=
\sum_{k=0}^{j-1}
\frac{g_\pm^{(k)}(\lambda_0)}{k!}\,
(\lambda - \lambda_0 \pm i0)^{-(j-k)}
\;+\;
h^\pm(\lambda),
\end{equation}
where $h^\pm$ is smooth near $\lambda_0$.
\end{lemma}

\begin{proof}
We only treat the case with the \(+\) sign, the other one being identical.  
We expand \(g\) in a Taylor series around \(\lambda_0\)  up to order \(j-1\):
\[
g(\lambda+i\varepsilon)
=
\sum_{k=0}^{j-1}
\frac{g_\pm^{(k)}(\lambda_0)}{k!}\,
(\lambda+i\varepsilon-\lambda_0)^k
\;+\;
R_j(\lambda+i\varepsilon),
\]
where the remainder satisfies 
$R_j(\lambda+i\varepsilon)
=\mathcal{O}\big((\lambda+i\varepsilon-\lambda_0)^j\big)$ in a neighborhood of $V_{\lambda_0}$.
Dividing by $(\lambda+i\varepsilon-\lambda_0)^j$ yields
\[
\frac{g(\lambda+i\varepsilon)}{(\lambda+i\varepsilon-\lambda_0)^j}
=
\sum_{k=0}^{j-1}
\frac{g_\pm^{(k)}(\lambda_0)}{k!}\,
(\lambda+i\varepsilon-\lambda_0)^{-(j-k)}
\;+\;
r_j(\lambda,\varepsilon),
\]
where $r_j(\lambda,\varepsilon)$ is a smooth function and converges, as $\varepsilon\to0^+$, 
to a smooth function $h^+(\lambda)$ near $\lambda_0$.
Since each family 
$(\lambda+i\varepsilon-\lambda_0)^{-m}$ converges to 
$(\lambda-\lambda_0+i0)^{-m}$ in $\mathcal{D}'(\mathbb{R})$,
we may pass to the limit term by term, which gives~\eqref{eq:prod_dist_limit_pm}.
\end{proof}

\begin{remark}[Real and imaginary parts]
\label{rem:real_imag_parts}
Each distribution $(\lambda-\lambda_0+i0)^{-m}$ appearing in~\eqref{eq:prod_dist_limit_pm}
admits the following decomposition into its real and imaginary parts:
\begin{equation}\label{eq:decomp_real_imag}
(\lambda-\lambda_0+i0)^{-m}
=
\operatorname{p.v.}\frac{1}{(\lambda-\lambda_0)^m}
\;-\;
i\pi\,\frac{(-1)^{m-1}}{(m-1)!}\,\delta^{(m-1)}(\lambda-\lambda_0),
\qquad m\ge 1.
\end{equation}
%In particular,
%\begin{align*}
%(\lambda-\lambda_0+i0)^{-1}
%&=\operatorname{p.v.}\frac{1}{\lambda-\lambda_0}
%- i\pi\,\delta(\lambda-\lambda_0),\\[4pt]
%(\lambda-\lambda_0+i0)^{-2}
%&=\operatorname{p.v.}\frac{1}{(\lambda-\lambda_0)^2}
%+ i\pi\,\delta'(\lambda-\lambda_0).
%\end{align*}
Hence, the real part of $(\lambda-\lambda_0+i0)^{-m}$ corresponds to the principal value 
$\operatorname{p.v.}\frac{1}{(\lambda-\lambda_0)^m}$,
while its imaginary part is supported at $\lambda=\lambda_0$ and involves derivatives of the Dirac delta distribution.
\end{remark}

%%%%%%%%%%%%%%%%%%%%%%%  Raccourci %%%%%%%%%%%%%%%%%%%%%%%

\vspace{0.2cm}\noindent
Now, according to Proposition~\ref{derivative}, and since \(v\) is compactly supported, the following representation holds in the sense of distributions:
\begin{equation*}
	\xi'(H, H_0; \lambda)
	\;=\;
	\frac{1}{2\pi i}
	\lim_{\varepsilon \to 0^+} 
	\Tr\!\Big(
	\Res_H(\lambda + i\varepsilon) - \Res_0(\lambda + i\varepsilon)\Big) 
	\;-\;
		\Tr\!\Big(
	\Res_H(\lambda - i\varepsilon) + \Res_0(\lambda - i\varepsilon)
	\Big).
\end{equation*}
From~\eqref{eq:Id_Trace3} and recalling that \(C = \rho\), \(W = V\), it follows that
\begin{align*}
\Tr\,\bigl(\Res_H(\lambda \pm i\varepsilon) - \Res_0(\lambda \pm i\varepsilon)\bigr)
&=\Tr\,\bigl(T_0^{(1)}(\lambda \pm i\varepsilon)\bigr)
\;-\;
\Tr\,\bigl(T_0^{(1)}(\lambda \pm i\varepsilon)\,T_0(\lambda \pm i\varepsilon)\bigr) \\
&\quad+\;
\Tr\,\bigl(T_0^{(1)}(\lambda \pm i\varepsilon)\,\rho\,\Res_H(\lambda \pm i\varepsilon)\,\rho\,V\,T_0(\lambda \pm i\varepsilon)\bigr) \\
&:=\sum_{n=1}^3 a_n^{\pm}(\varepsilon).
\end{align*}
We first consider the case with the plus sign and in particular the third term \(a_3^+ (\epsilon)\).  
We define
\[
g_j(\mu)
\;:=\;
\Tr\,\bigl(T_0^{(1)}(\mu)\,A_{-j}\,V\,T_0(\mu)\bigr).
\]
Since $V$ is compactly supported, $g_j$ satisfies the assumptions of Lemma~\ref{lem:product_principalpart}. 
Let $\lambda$ lie in a real neighbourhood of~$\lambda_0$, and let $\varepsilon>0$ small enough. 
Substituting $z = \sqrt{\lambda + i\varepsilon} \in V_{\lambda_0}\cap \C_+$ into~\eqref{eq:DL singularity}, 
where $V_{\lambda_0}$ denotes a complex neighbourhood of~$\lambda_0$, 
we obtain 
\[
F(z) = \rho\,R_H(\lambda + i\varepsilon)\,\rho,
\]
and a straightforward computation yields, in the sense of distributions 
\begin{align*}
\lim_{\varepsilon \to 0^+} a_3^+(\varepsilon)
& \;=\;
\sum_{j=1}^{\nu_0} 
\sum_{k=0}^{j-1} 
\frac{g_{j,+}^{(k)}(\lambda_0)}{k!}\,
\frac{1}{(\lambda - \lambda_0 + i0)^{j-k}}
\;+\;
H^+(\lambda)\\
& \;=\;
\sum_{l=1}^{\nu_0} 
\frac{\alpha_l(\lambda_0)}{(\lambda - \lambda_0 + i0)^{l}}
\;+\;
H^+(\lambda)
\end{align*}
where $g_{j,+}$ denotes the smooth extension of $g_j$ up to the real axis, $H^+$ denotes a smooth function near $\lambda_0$ and 
$\displaystyle{\alpha_l = \sum_{k=0}^{\nu_0-l} 
\frac{1}{k!}}g_{k+l,+}^{(k)}$.

\vspace{0.2cm}\noindent
The situation in the case with the minus sign is different, since 
$\lambda - i\varepsilon$ can be written as $z^2$ with 
$z$ lying in a neighbourhood of $-\sqrt{\lambda_0}$ intersected with~$\C_+$. 
As $-\sqrt{\lambda_0}$ is not a pole, the truncated resolvent extends holomorphically 
to this region, and we obtain
\[
\lim_{\varepsilon \to 0^+} a_3^-(\varepsilon) = H^-(\lambda),
\]
where $H^-$ denotes a smooth function near $\lambda_0$.

%%%%%%%%%%%%%%%%%%%%%%%%%%%%%%%%%%%%%
\vspace{0.2cm}\noindent
We now compare the first two terms with opposite boundary values on the real axis. 
For each $\varepsilon>0$, we define
\[
a_1^{\pm}(\varepsilon)
:= 
\Tr\,\bigl(T_0^{(1)}(\lambda \pm i\varepsilon)\bigr),
\qquad
a_2^{\pm}(\varepsilon)
:= 
-\,\Tr\,\bigl(T_0^{(1)}(\lambda \pm i\varepsilon)\,T_0(\lambda \pm i\varepsilon)\bigr).
\]
%Passing to the limit as $\varepsilon \to 0^+$ yields 
%\(a_n^{\pm} := \lim_{\varepsilon\to 0^+} a_n^{\pm}(\varepsilon)\), for \(n=1,2\).

\vspace{0.2cm}\noindent
By~\eqref{eq:diff Trace2}, we immediately obtain, in the sense of distributions,
\begin{align*}
\lim_{\varepsilon \to 0^+}
\Bigl(
a_1^+(\varepsilon) - a_1^-(\varepsilon)
\Bigr)
%=\lim_{\varepsilon \to 0^+}
%\Tr\,\Bigl(T_0^{(1)}(\lambda+i\varepsilon)-T_0^{(1)}(\lambda-i\varepsilon)\Bigr) \\
&=\frac{i}{4\pi\sqrt{\lambda}}
\int_{\R^3} V(x)\,dx.
\end{align*}
In the same way, using~\eqref{eq: trace T_0^1T_0}, we get in the sense of distributions,
\begin{align*}
%a_2^+ - a_2^-
\lim_{\varepsilon \to 0^+}
\Bigl(a_2^+(\varepsilon)-a_2^-(\varepsilon)\Bigr) 
%&=\,-\lim_{\varepsilon \to 0^+}
%\Bigl(
%\Tr\,\bigl(T_{0}^{(1)}(\lambda + i\varepsilon)\, T_{0}(\lambda + i\varepsilon)\bigr)
%-
%\Tr\,\bigl(T_{0}^{(1)}(\lambda - i\varepsilon)\, T_{0}(\lambda - i\varepsilon)\bigr)
%\Bigr) \\
&=\,-\,\frac{i}{32\pi^{2}\sqrt{\lambda}}
\int_{\R^{6}} 
V(x)\,V(y)\,
\frac{e^{ 2i\sqrt{\lambda}\,\lvert x-y\rvert} + e^{- 2i\sqrt{\lambda}\,\lvert x-y\rvert}}{\lvert x-y\rvert}\,
dx\,dy \\
&=\,-\,\frac{i}{16\pi^{2}\sqrt{\lambda}}
\int_{\R^{6}} 
V(x)\,V(y)\,
\frac{\cos\!\bigl(2\sqrt{\lambda}\,\lvert x-y\rvert\bigr)}{\lvert x-y\rvert}\,
dx\,dy.
\end{align*}

\vspace{0.2cm}\noindent
Using the previous notations, we now derive an explicit expression for the singular part 
of the derivative of the spectral shift function near an outgoing spectral singularity.

\begin{theorem}[Explicit distributional formula for \(\xi'(H,H_0;\lambda)\)]
\label{thm:explicit_xiprime}
Assume that \(v\) is compactly supported and that \(\lambda_0>0\) is an outgoing spectral singularity of \(H\) of finite order~\(\nu_0\).  
Then, in the sense of distributions on a real neighbourhood of~\(\lambda_0\),
\[
\xi'(H,H_0;\lambda)
= \sum_{l=1}^{\nu_0} 
\frac{\alpha_l(\lambda_0)}{(\lambda - \lambda_0 + i0)^{l}}
%= \frac{1}{2\pi i} \ 
%\sum_{j=1}^{\nu_0} 
%\sum_{k=0}^{j-1} 
%\frac{g_{j,+}^{(k)}(\lambda_0)}{k!}\,
%\frac{1}{(\lambda - \lambda_0 + i0)^{j-k}}
\;+\;
H(\lambda),
\]
where \(H(\lambda)\) is smooth near~\(\lambda_0\), and where
$\displaystyle{\alpha_l := \sum_{k=0}^{\nu_0-l} 
\frac{1}{k!}}\, g_{k+l,+}^{(k)}$
with $g_{j,+}$  the smooth extension up to the real axis of
\[
g_j(\mu)
:= \Tr\!\ \bigl(T_0^{(1)}(\mu)\,A_{-j}\,V\,T_0(\mu)\bigr),
\qquad 
\mu \in V_{\lambda_0}\cap \C_+,
\]
with \(V_{\lambda_0}\) a complex neighbourhood of~\(\lambda_0\).
\end{theorem}

\noindent
This formula shows that the singular behavior of \(\xi'(H,H_0;\lambda)\) near \(\lambda_0\)
is entirely determined by the finite-rank residues \(A_{-j}\) of the weighted resolvent. 
In particular, the singular part of \(\xi'(H,H_0;\lambda)\) has the same distributional 
structure as the boundary values of the resolvent \(\Res_H(\lambda+i0^+)\).
This phenomenon is quite natural, as it originates from the loss of uniform control of the resolvent 
in the upper half-plane near the outgoing spectral singularity.

\begin{remark}
In the case where $\lambda_0>0$ is an \emph{incoming} spectral singularity of $H$, 
an analogous formula holds with the boundary value $(\lambda - \lambda_0 - i0)^{-(j-k)}$ 
and the coefficients $g_{j,-}$.  
If both incoming and outgoing spectral singularities occur at $\lambda_0$, 
the distributional derivative $\xi'(H,H_0;\lambda)$ receives contributions from both sides.
\end{remark}

\subsection{High Energy Asymptotic}

In this section, we are interested in the asymptotic behavior of the derivative of the spectral shift function as the spectral parameter tends to \( +\infty \). We compute the leading term in the asymptotic expansion and show that it coincides with the self-adjoint case. The method we use is the same as in the self-adjoint setting (see \cite{CdV81}).\\
	
The next proposition establishes the high energy asymptotic of the spectral shift function associated to $H$ and $H_0$ under a short range assumption on $V$. 
		
\begin{prop}\label{prop:high energy asymptotic}
    Suppose that $V$ is a bounded short range potential with $\delta>3$. Suppose that $H$ has a finite number of eigenvalues and a finite number of spectral singularities of finite order. Then  
    \begin{equation}\label{eq:high asymp 1}
		\xi'(\lambda,H,H_0)=\frac{1}{8\pi^2 \sqrt{\lambda}}\int_{\R^3} V(x) \ \mathrm{d}x+o\left(\frac{1}{\sqrt{\lambda}}\right)\quad \lambda\to\infty.
    \end{equation}
    Moreover if $V$ belongs to $\mathscr{C}^1(\R^3,\C)$ and satisfies $| \nabla V(x) | \leq C_\alpha \langle x\rangle^{-\delta}$ for some $\delta >3$, then 
	\begin{equation}\label{eq:high asymp 2}
		\xi'(\lambda,H,H_0)=\frac{1}{8\pi^2 \sqrt{\lambda}}\int_{\R^3}V(x)\ \mathrm{d}x+\mathcal{O}\left(\frac{1}{\lambda}\right)\quad \lambda\to\infty.
	\end{equation}
\end{prop}

\begin{proof}
    It follows from point (3) of Proposition \ref{prop:proprieteT_0} that there exist constants \(\lambda_{0}>1\) and \(\varepsilon_{0}>0\) such that for all \(\lambda\ge\lambda_{0}\) and all \(0\le\varepsilon\le\varepsilon_{0}\), one has 
    \[
        \bigl\|T_{0}(\lambda\pm i\varepsilon)\bigr\|_{\mathcal{B}(L^{2})}<1
    \]
    in the uniform (operator-norm) topology. Hence \(\Id+T_{0}(\lambda\pm i\varepsilon)\) is invertible, and its inverse admits the Neumann series expansion
    \begin{equation}\label{eq:Neumann}
        [\Id + T_{0}(\lambda\pm i\varepsilon)]^{-1}
        =\sum_{k=0}^{\infty}(-1)^{k}\,T_{0}(\lambda\pm i\varepsilon)^{k},
    \end{equation}
    with the series converging uniformly in \(\varepsilon\) with respect to the operator norm up to $\varepsilon\to 0^+$. Combined with \cite[Proposition 4.7]{FaFr23} this implies that if $\lambda>\lambda_0$, then $\lambda$ is a regular spectral point. With Proposition \ref{Prop:regularity}, $\xi'(H,H_0,.)$ is continuous on $(\lambda_0,+\infty)$ and this allows us to compute the asymptotic of $\xi'(H,H_0,.)$ when $\lambda\to +\infty$.\\

    Next with \eqref{eq:Id_Trace1} and \eqref{eq:Neumann} we have that for all $\lambda\in (\lambda_0,+\infty)$, 
    \begin{align}
        \xi'(\lambda,H,H_0)=(2\pi i)^{-1}&\bigg(\lim_{\varepsilon\rightarrow 0^+}\Tr\left(T_0^{(1)}(\lambda+i\varepsilon)-T_0^{(1)}(\lambda-i\varepsilon)\right)\label{eq:dev asympto regularity1}\\
        &-\lim_{\varepsilon\rightarrow 0^+}\Tr\left(T_0^{(1)}(\lambda+i\varepsilon)T_0(\lambda+i\varepsilon)-T_0^{(1)}(\lambda-i\varepsilon)T_0(\lambda-i\varepsilon)\right)\label{eq dev: asympto regularity2}\\
        &+\sum_{k=2}^\infty (-1)^{k}\lim_{\varepsilon\rightarrow 0^+}\Tr\left(T_0^{(1)}(\lambda+i\varepsilon)T_0(\lambda+i\varepsilon)^{k}-T_0^{(1)}(\lambda-i\varepsilon)T_0(\lambda-i\varepsilon)^k \right)\bigg)\label{eq: dev asympto regularity3}
    \end{align}
    
   To compute the limit of \eqref{eq:dev asympto regularity1} it suffices to use \eqref{eq:diff Trace2} and we have
    \begin{equation*}
        (2\pi i)^{-1}\lim_{\varepsilon\rightarrow 0^+}\Tr\left(T_0^{(1)}(\mu+i\varepsilon)-T_0^{(1)}(\mu-i\varepsilon)\right)=\frac{1}{8\pi^2 \sqrt{\lambda}}\int_{\R^3}V(x)\mathrm{d}x. 
    \end{equation*}
    
    To analyze the asymptotic behavior of the terms corresponding to \eqref{eq dev: asympto regularity2}, one computes directly by using \eqref{eq: trace T_0^1T_0} 
     \begin{align}
        &\lim_{\varepsilon\to0^{+}}\left[\Tr\left(T_{0}^{(1)}(\lambda + i\varepsilon)\,T_{0}(\lambda + i\varepsilon)-T_{0}^{(1)}(\lambda - i\varepsilon)\,T_{0}(\lambda - i\varepsilon)\right)\right]\nonumber\\
        &\quad=\Tr\left(T_{0}^{(1)}(\lambda + i0^{+})\,T_{0}(\lambda + i0^{+})-T_{0}^{(1)}(\lambda - i0^{+})\,T_{0}(\lambda - i0^{+})\right)\nonumber\\
        &\quad=\frac{i}{32\pi^{2}\sqrt{\lambda}}\int_{\mathbb{R}^{6}}V(x)\,V(y)\,\frac{\cos\bigl(2\sqrt{\lambda}\,\lvert x-y\rvert\bigr)}{\lvert x-y\rvert}\,\mathrm{d}x\,\mathrm{d}y.\label{eq:RiemannLebesgue}
    \end{align}
    By the Riemann–Lebesgue lemma, this term decays like \(o(\lambda^{-1/2})\) as \(\lambda\to\infty\). Moreover, if \(V\in \mathscr{C}^1(\R^3,\C)\) and its gradient satisfies
    \[
        | \nabla V(x) | =O\bigl(|x|^{-\delta }\bigr)\quad\text{for some }\delta>3,
    \]
    then the remainder is  \(\mathcal{O}(1/\lambda)\).  
    
    It remains to estimate \eqref{eq: dev asympto regularity3}. 
    Applying (3) of Proposition \ref{prop:proprieteT_0} we see that  $\bigl\|T_{0}(\lambda\pm i \varepsilon)^{\,k}\bigr\|_{\mathcal{B}(\mathcal{H})}=\mathcal{O}\bigl(\lambda^{-\tfrac{k}{2}}\bigr)$. Then
    for \(k\ge2\), one has :
    \begin{align}\label{eq:estimeremainderterm}
        \bigl\|T_{0}^{(1)}(\lambda\pm i\varepsilon) T_{0}(\lambda\pm i\epsilon)^{k}\,\bigr\|_{\mathcal{L}_{1}}
        &\;\le\;\bigl\|T_{0}^{(1)}(\lambda\pm i\varepsilon)\bigr\|_{\mathcal{L}_{2}}\;\bigl\|T_{0}(\lambda\pm i\epsilon)^{k-1}\bigr\|_{\mathcal{B}(\mathcal{H})}\;\bigl\|T_{0}(\lambda\pm i\varepsilon)\bigr\|_{\mathcal{L}_{2}}\\
        &=\mathcal{O}\bigl(\lambda^{-\frac{k}{2}}\bigr)\nonumber,
    \end{align}
    uniformly with respect to $\varepsilon>0$. Thus 
    $\eqref{eq: dev asympto regularity3} =\mathcal{O}(\lambda^{-1})$ uniformly with respect to $\varepsilon$. 
    \end{proof}

\begin{remark}
 In \cite{Ro1991}, D.~Robert investigated the high‐energy asymptotics of the scattering phase for the free Laplacian \(H_0\) perturbed by a smooth, real‐valued, decaying potential \(V\) on \(\mathbb{R}^n\). Under the decay condition
\[
|\partial_x^\alpha V(x)| \le C_\alpha \langle x\rangle^{-\rho-|\alpha|},\quad \rho>n,
\]
the spectral shift function is defined in terms of the scattering matrix \(S(\lambda)\) by
\begin{equation}\label{lienssf}
\det S(\lambda) \;=\;\exp\bigl(-2\pi i\,\xi(\lambda;H,H_0)\bigr)\,.
\end{equation}
In the dissipative case (i.e., when \(\im V\le0\)), Faupin and Nicoleau \cite{FaNi18_01} have shown that the scattering matrices \(S(\lambda)\) are well‐defined and admit an explicit representation formula. It is important to note that, although the relation \eqref{lienssf} has not been proven in the non‐selfadjoint setting, the result of our theorem is consistent with the fact that, in the dissipative case,  \(S(\lambda)\) is a contraction.
\end{remark}

%\medskip         
            
\section{Explicit simple examples}\label{Sec:Toys}

%	\subsection{Explicit simple examples}
			
			In order to illustrate and comment our definition of the spectral shift function (SSF), let us give some explicit calculations for very simple examples. On these toy models we have phenomena known in the self-adjoint context (integer jump at real eigenvalues, fractional jump at  spectral singularities/resonances, smoothness outside these singularities, ...), but also some new phenomena. A first novelty is the non-integrability of the SSF in the presence of non-real eigenvalues (for perturbations of trace class). A second is the fact that the SSF is no longer real. On our toy models, it remains real if the perturbation does not interact with the continuous spectrum and the sign of its imaginary part is related to that of the perturbation.

\subsection{In finite dimension: diagonalizable operators }\label{ssfEx1} When the Hilbert space is  $\mathcal{H} = \C$, consider $H_0$ the multiplication operator by $\lambda_0\in \R$ and $H$ the multiplication by $\lambda_0 + v$, $v\in \R$. For $f \in D(\R)$, $f(H_0)$ is the multiplication by $f(\lambda_0)$. For $H$, it depends if $v$ is real or not. When $v \in \R$, $f(H)$ is the multiplication by $f(\lambda_0+v)$, while when $v \in \C \setminus \R$, $f(H)=0$. Thus when $v \in \R$ the SSF corresponds to the standard SSF. 
From the Fundamental Theorem of Calculus, we have
\[ \Tr (f(H)- f(H_0)) = f(\lambda_0+v) - f(\lambda_0) = \int_{\lambda_0}^{\lambda_0+v} f^\prime(\lambda)  d \lambda \]
and, up to an additive constant, almost everywhere, we have
\[
	\xi(\lambda; H, H_0) = \left\{ \begin{aligned}
\one_{[\lambda_0, \lambda_0+v]} (\lambda) & \textrm{ if }  & v\geq 0\\
- \one_{[\lambda_0+v, \lambda_0]} (\lambda) & \textrm{ if }  & v\leq 0\\
\end{aligned} \right. .
\]
The SSF has two jumps at the eigenvalues of $H_0$ and of $H$.
When $v \in \C \setminus \R$,
\[ \Tr (f(H)- f(H_0)) =  - f(\lambda_0) = \int_{\lambda_0}^{+ \infty} f^\prime(\lambda)  d \lambda \] 
and, up to an additive constant, we have
\[
	\xi(\lambda; H, H_0) = \one_{[\lambda_0, + \infty)} (\lambda),
\]
which has a unique jump at the real eigenvalue $\lambda_0$ which becomes a non-real eigenvalue under the perturbation. Unlike the selfadjoint case, $\xi(\cdot; H, H_0)$ is not compactly supported (nor integrable). This is related to the fact that the perturbed eigenvalue $\lambda_0+v$ is no longer real.
%but does not depend on $v$ (i.e. the SSF "does not see" the complex eigenvalue $\lambda_0+v$).

Obviously, we have the same phenomena in $\mathcal{H} = \C^4$ for diagonal matrices (or diagonalizable matrices). For example, if we consider the diagonal matrices
\[ H_0:= \left(  \begin{array}{cccc}
 \lambda_1 & 0 & 0 &0  \\
 0& \lambda_2 & 0 & 0  \\
  0& 0 &  \lambda_3 & 0  \\
   0& 0 & 0 & \lambda_4   \\
 \end{array}   \right)
 \quad 
  H:= H_0 + V, \quad V:=  \left(  \begin{array}{cccc}
 v_1 & 0 & 0 &0  \\
 0& v_2 & 0 & 0  \\
  0& 0 &  v_3 & 0  \\
   0& 0 & 0 & 0   \\
 \end{array}   \right)
 \]
with $\lambda_1 < \lambda_2 < \lambda_3 < \lambda_4$ and $v_1>0$, $v_2<0$, $v_3 \in \C \setminus \R$ then 
\[ \Tr (f(H)- f(H_0)) = f(\lambda_1+v_1) - f(\lambda_1)+ f(\lambda_2+v_2) - f(\lambda_2) - f(\lambda_3),\]
and, up to an additive constant, we have
\[
	\xi(\lambda; H, H_0) = \one_{[\lambda_1, \lambda_1+ v_1]} (\lambda) - \one_{[\lambda_2 + v_2, \lambda_2]} (\lambda) + \one_{[\lambda_3, + \infty)} (\lambda).
\]
This function has jumps at the real eigenvalues of $H_0$ and $H$, excepted at the unperturbed eigenvalue $\lambda_4$. The fact that the SSF is equal to 1 up to infinity is related to the appearance of the non-real eigenvalue $\lambda_3+v_3$. %It is independent of the complex eigenvalue $\lambda_3+v_3$. 

\subsection{In finite dimension: an undiagonalizable  case }\label{ssfEx2}

Let us consider, in $\mathcal{H} = \C^2$, a case when the non-selfadjoint perturbation is not diagonalizable. For example let
\[ H_0:= \left(  \begin{array}{cc}
 \lambda & 0  \\
 0& \lambda  \\
 \end{array}   \right)
 \quad 
  H:= H_0 + V, \quad V:=  \left(  \begin{array}{cc}
 0 & v  \\
 0 & 0  \\
 \end{array}   \right) , \quad \lambda \in \R, \quad v \in \R^*.
 \]
We have $f(H_0)= f(\lambda) I_2$ and $f(H)= f(\lambda) I_2 + f^\prime(\lambda) \left(  \begin{array}{cc}
 0 & v  \\
 0 & 0  \\
 \end{array}   \right)$ because 
 $$(H-z)^{-1} = (\lambda-z)^{-1}  I_2 - (\lambda-z)^{-2} \left(  \begin{array}{cc}
 0 & v  \\
 0 & 0  \\
 \end{array}   \right).$$
Then $ \Tr (f(H)- f(H_0)) = 0$ and in this case where the real spectrum is unchanged under the perturbation, the SSF is a constant that can be chosen to be zero. 

In the previous examples of operators having discrete spectra, the SSF is simply a step function which counts the eigenvalues of $H_0$ which are perturbed and the perturbed real eigenvalues of $H_\gamma$.  
The notion of SSF is mainly introduced to deal with operators having a non-empty continuous spectrum but the explicit calculus of the SSF can quickly become complicated. Let us discuss a simple example of a rank $1$ perturbation of a selfadjoint operator having continuous spectrum.

\subsection{Rank one perturbation: weak interaction with the continuous spectrum}\label{ssfEx3}

In the space $\mathcal{H} =L^2(\R)$, we consider $H_0$ the operator of multiplication by the function $h_0(x):= x \one_{[0,1]}(x)$ and for $\gamma \in  \C$, the perturbation $V_\gamma = \gamma \Pi_0$, where $\Pi_0$ is the orthogonal projection onto a normalized function $u_0 \in L^2(\R)$: $\Pi_0:= \langle \cdot , u_0 \rangle u_0$, $\| u_0 \|=1$. 
The spectrum of $H_0$ is $[0,1]$ and $0$ is an eigenvalue of infinite multiplicity (each function supported outside $[0,1]$ is an eigenfunction). 
Since $V_\gamma$ is a compact perturbation of $H_0$, then $H_\gamma=H_0+V_\gamma$ is closed and its essential spectrum is $[0,1]$. % (see e.g. \cite[Theorem 11.2.6]{Da07}).
Let us assume that $u_0$ is supported in $\R \setminus [0,1]$, then $H_0 u_0 = 0$ and $H_\gamma u_0:= (H_0 + V_\gamma) u_0 = \gamma u_0$, that is $u_0$ is an eigenfunction corresponding to the eigenvalue $\gamma$ for $H_\gamma$. 
In this case, we have
%$V_\gamma H_0 = 0 = H_0V_\gamma$ and the resolvent of $H_\gamma:= H_0 + V_\gamma$ satisfies:
\[(H_\gamma - z)^{-1} = (H_0  - z)^{-1} + \big( \frac{1}{z} +  \frac{1}{\gamma - z}\big) \Pi_0,\]
where $(H_0  - z)^{-1}$ is the operator of multiplication by $(x-z)^{-1} \one_{[0,1]} - z^{-1} \one_{\R\setminus [0,1]}$. 
Then for any $f\in D(\R)$,
\[ \Tr (f(H_\gamma)- f(H_0)) = \left\{ \begin{aligned}
f(\gamma) - f(0)& \quad \textrm{ if }  & \gamma \in \R\\
- f(0) & \quad \textrm{ if }  & \quad \gamma \in \C \setminus \R \\
\end{aligned} \right. ,
\]
and, up to a constant, the SSF is given, almost everywhere, by:
\[
	\xi(\lambda; H_\gamma, H_0) = \left\{ \begin{aligned}
\one_{[0, \gamma]} (\lambda) &\quad \textrm{ if }  & \gamma > 0\\
- \one_{[\gamma, 0]} (\lambda) & \quad \textrm{ if }  & \gamma < 0\\
 \one_{[0, + \infty[} (\lambda) & \quad \textrm{ if }  & \gamma \in \C \setminus \R .\\
\end{aligned} \right. 
\]
In this case, even though there is continuous spectrum, the spectra of $H_\gamma$ and $H_0$ differ by only one eigenvalue, and the SSF is still a step function with jumps at the real eigenvalues influenced by the perturbation. 

Let us mention that the formula \eqref{ssfLnD} gives the same expression by choosing $\arg  D_V( \lambda) = 0$ when $\lambda \rightarrow - \infty$. Let us check it for $\gamma = i \beta$, $\beta >0$ (the case $\gamma \in \R$ is treated in \cite{Bru25}). 
We have, 
\[ D_{V_\gamma}(z):=  \Det \left( I + V_\gamma (H_0 - z)^{-1} \right) = 1 + \gamma \langle (H_0 - z)^{-1} u_0 , u_0 \rangle = 1 - \frac{\gamma}{z}  \]
 whose the pole $z=0$ and the zero $z=\gamma$ are the eigenvalues of $H$. When $\lambda \rightarrow - \infty$, $D_{V_\gamma}(\lambda)$ tends to $1$ and we choose its argument equal to $0$. Then the logarithm of $D_{V_\gamma}(\lambda \pm i \varepsilon)$ is well defined and smooth with respect to $\lambda \in \R$ for $\varepsilon >0$ sufficiently small. We have 
 \[  \ln D_{V_\gamma}(\lambda + i \varepsilon) - \ln D_{V_\gamma}(\lambda - i \varepsilon) = a_\varepsilon (\lambda) + i \, b_\varepsilon (\lambda),
 \]
 with 
 \[a_\varepsilon (\lambda)= \ln \left| \frac{D_{V_\gamma}(\lambda + i \varepsilon)}{D_{V_\gamma}(\lambda - i \varepsilon)} \right|, \qquad 
 b_\varepsilon (\lambda) = \arg D_{V_\gamma}(\lambda + i \varepsilon) - 
 \arg D_{V_\gamma}(\lambda - i \varepsilon).
 \]
 Clearly, $\lim_{\varepsilon \to 0^+} a_\varepsilon (\lambda) = 0$ because 
 \[ \left|  D_{V_\gamma}(\lambda \pm i \varepsilon)
 \right|^2 = \left|  1 - \frac{i \beta}{\lambda \pm i \varepsilon} \right|^2 = \frac{\lambda^2 + (\beta \mp \varepsilon)^2}{\lambda^2 + \varepsilon^2} 
 %\rightarrow \frac{\lambda^2 + \beta^2}{\lambda^2}
 \quad \text{ and } 
 \quad \lim_{\varepsilon \to 0^+}  
\frac{\lambda^2 + (\beta + \varepsilon)^2}{\lambda^2 + (\beta - \varepsilon)^2}=1 .\]
For the imaginary part $b_\varepsilon (\lambda)$, we have to follow the argument of $1 - \frac{i \beta}{\lambda \pm i \varepsilon}$ when $\lambda$ goes from $-\infty$ to  $+\infty$ on the real axis. The case "$-$" is simpler because $$1 - \frac{i \beta}{\lambda - i \varepsilon}=1 + \frac{ \beta \varepsilon}{\lambda^2+  \varepsilon^2} - \frac{i \beta \lambda}{\lambda^2+  \varepsilon^2} $$
always has positive real  part and then its argument remains in $]-\frac{\pi}{2}, \frac{\pi}{2}[$. In particular when this quantity is real the argument vanishes (there is no turning point). In the case  "$+$" we have $$1 - \frac{i \beta}{\lambda + i \varepsilon}=1 - \frac{ \beta \varepsilon}{\lambda^2+  \varepsilon^2} - \frac{i \beta \lambda}{\lambda^2+  \varepsilon^2}$$
whose real part vanishes twice (for $\lambda = \pm \sqrt{\beta \varepsilon - \varepsilon^2}$) and the imaginary part once (for $\lambda =0$). Then its argument goes from $0$ (when $\lambda$ tends to $-\infty$) to $2\pi$ (when $\lambda$ tends to $+\infty$).
%(FAIRE UN DESSIN?). 
When $\lambda =0$ and $\varepsilon$ is small, these quantities have opposite sign and the phase shift is $\pi$. For $\lambda \neq 0$, $1 - \frac{i \beta}{\lambda + i \varepsilon}$ and $1 - \frac{i \beta}{\lambda - i \varepsilon}$ have the same limit $1 - \frac{i \beta}{\lambda}$ but the turning point $0$ produces a phase shift of $2\pi$, and we have:
\[\lim_{\varepsilon \to 0^+} b_\varepsilon (\lambda) = \left\{ \begin{aligned}
0 &\quad \textrm{ if }  & \lambda < 0\\
2 \pi & \quad \textrm{ if }  & \lambda >0 \\
\end{aligned} \right. .\] 
We obtain 
\[ \frac{1}{2i\pi} \lim_{\varepsilon \to 0^+} \ln D_{V_\gamma}(\lambda + i \varepsilon) - \ln D_{V_\gamma}(\lambda - i \varepsilon) = \frac12 \big( \one_{[0, + \infty[} (\lambda) + \one_{]0, + \infty[} (\lambda) \big),\]
which is consistent with the above expression of $\xi(\lambda; H_\gamma, H_0)$ (it coincides with $\one_{[0, + \infty[} (\lambda)$ excepted at the point $\lambda=0$).

\subsection{Rank one perturbation: stronger interaction with the continuous spectrum}\label{ssfEx4}
In the previous example, the calculation of the SSF was simplified by the fact that the spectra of $H_\gamma$ and $H_0$ differ by only two eigenvalues and the SSF is still a step function. 
A more interesting example for which the perturbation interacts with the continuous spectrum is the case where supp$u_0 \cap [0,1] \neq \emptyset$.
Let us consider the previous operators $H_0$ and $H_\gamma= H_0+ V_\gamma$ with $u_0= \one_{[0,1]}$. In this case, the computation of $f(H_\gamma)$ is more tricky, so we prefer to compute the SSF $\xi(\lambda; H_\gamma, H_0)$ by using the formula \eqref{ssfLnD}, for $\gamma = i \beta$, $\beta >0$. Thanks to \eqref{ssf*} for $\beta <0$, it suffices to take the complex conjugate. 

As in the selfadjoint case (see \cite{Bru25} for $\gamma \in \R$), for any $u \in L^2(\R)$,
\[V_\gamma \Res_0(z) u = \gamma \langle (H_0 - z)^{-1} u \, , \,  u_0 \rangle u_0 = \gamma \left(\int_{[0,1]} \frac{u(x)}{x-z} dx \right) \, u_0,\]
and for $\gamma = i \beta$, $\beta >0$, $z = x + i y$, 
\[ D_{V_\gamma}(z) = 1 + \gamma \langle \Res_0(z) u_0 , u_0 \rangle = 
1 - \beta \left( \arctan \frac{1-x}{y} +  \arctan \frac{x}{y} \right)
+i \frac{\beta}{2}
\ln \left( \frac{(1-x)^2+y^2}{x^2+y^2} \right).\]
The eigenvalues of $H_\gamma$ are $0$ (as for $H_0$, the associated eigenfunctions are supported outside $[0,1]$) and the zeroes of $D_{V_\gamma}$ in $\C \setminus [0,1]$, that is $z= x+ iy$ satisfying
\[ 1 = \beta \left( \arctan \frac{1-x}{y} +  \arctan \frac{x}{y} \right); \qquad 
\ln \left( \frac{(1-x)^2+y^2}{x^2+y^2} \right)=0, \qquad y \neq 0 \text{ or } x \notin [0,1].\]
It admits a unique solution if and only if $\beta >\frac1{\pi}$: $z_\beta:= \frac12 (1 + i \coth\frac1{2 \beta})$. Thus for $\beta \leq \frac1{\pi}$ the spectrum of $H_\gamma$ is $[0,1]$ and for $\beta >\frac1{\pi}$, it is $[0,1]\cup \{z_\beta \}$ where $z_\beta$ is an eigenvalue of multiplicity one associated with the eigenfunction $w_\beta(x):=\frac1{z_\beta -x} \one_{[0,1]}(x) $.  For $\beta = \frac1{\pi}$, we see that $z_\beta = \frac12 $ is an outgoing spectral singularity of order $1$. Indeed, $V_{i \beta} = CWC$ with $C=\Pi_0$, $W=i \beta$, and thanks to \eqref{invertible}, 
\[ C \Res_{H_{i \beta}} (z) C W = \Id - (\Id + T_0(z))^{-1} = \Big(1 - \frac{1}{D_{V_{i\beta}}(z)}\Big)\Pi_0, \]
where we used that $T_0(z):= C \Res_{H_{0}} (z) C W = i \beta \langle \Res_0(z) u_0 , u_0 \rangle \Pi_0$.
Then from the above expression of $D_{V_{i\beta}}(z)$ for $z= \frac12 + \delta + i \varepsilon$, $\varepsilon >0$, we deduce that $|z-\frac12|^n \| C \Res_{H_{i \beta}} (z) C W\|$ is uniformly bounded with respect to $\delta + i \varepsilon$ small enough, for $n\geq 1$ (but not for $n=0$).
%because for $w_\varepsilon (x):= \frac1{\frac12+i \varepsilon -x} \one_{[0,1]}(x) $, \[\Big(H_{\frac{i}{\pi}}-(\frac12 + i \varepsilon)\Big) w_\varepsilon = (\frac{2}{\pi} \arctan \frac1{2\varepsilon} -1 ) \one_{[0,1]} \longrightarrow 0 \quad \text{ as } \varepsilon \rightarrow 0^+.\]}

Now in order to apply \eqref{ssfLnD}, let us study the modulus and the argument of $D_{V_\gamma}(z)$
for $z = \lambda \pm i \varepsilon$, $\varepsilon >0$, given by:
\[ D_{V_\gamma}(\lambda \pm i \varepsilon) = 
1 \mp \beta \left( \arctan \frac{1-\lambda}{\varepsilon} +  \arctan \frac{\lambda}{\varepsilon} \right) + i  \frac{\beta}{2} \ln \left( \frac{(1-\lambda)^2+\varepsilon^2}{\lambda^2+\varepsilon^2} \right) .\]

\subsubsection{Imaginary part of the SSF}

The imaginary part of $\xi(\lambda; H_\gamma, H_0)$ is:
\[ \im \xi(\lambda; H_\gamma, H_0) = - \frac{1}{2\pi} \lim_{\varepsilon \to 0^+} 
\ln \left| \frac{D_{V_\gamma}(\lambda + i \varepsilon)}{D_{V_\gamma}(\lambda - i \varepsilon)} \right|.
 \]
When $\lambda \in \R \setminus [0,1]$, $(1-\lambda)$ and $\lambda$ have opposite signs then 
\[ \lim_{\varepsilon \to 0^+} D_{V_\gamma}(\lambda \pm i \varepsilon) = 1 + i  \beta \ln \left( \frac{\lambda -1}{\lambda} \right),\]
is independent of the sign in front of $\varepsilon$ and 
\begin{equation}\label{Inulle}
    \im \xi(\lambda; H_\gamma, H_0) = - \frac{1}{2\pi} \ln 1=0 , \qquad \textrm{for } \lambda \in \R \setminus [0,1].
\end{equation} 
When $\lambda \in ]0,1[$, we have
\[\lim_{\varepsilon \to 0^+} \left( \arctan \frac{1-\lambda}{\varepsilon} +  \arctan \frac{\lambda}{\varepsilon} \right) = \pi.\]
Then, 
\begin{equation}\label{INnulle}
    \im \xi(\lambda; H_\gamma, H_0) = - \frac{1}{4\pi} \ln \left( \frac{(1-\beta \pi)^2+\beta^2f(\lambda)^2}{(1+\beta \pi)^2+\beta^2f(\lambda)^2} \right) = G_\beta(f(\lambda))
    ,
    %, \qquad \textrm{for } \lambda \in  [0,1],
\end{equation} 
where $G_\beta$ is the even function defined on $\R$ by:
\[
G_\beta(X) :=  \frac{1}{4\pi} \ln \left( 1 + \frac{4\beta \pi}{(1-\beta \pi)^2+\beta^2X^2} \right),
\]
and
$f$ is the decreasing function from $]0,1[$ onto $\R$,  given by:
\begin{equation}
    f(\lambda):= \ln (\lambda^{-1} -1), \quad \lambda \in ]0,1[.
\end{equation}
It is a symmetric function with respect to $1/2$ (i.e. $f(1-\lambda) = - f(\lambda)$).

It follows that the imaginary part of the SSF is given by:
\begin{equation}\label{Imssf}
   \im \xi(\lambda; H_\gamma, H_0)  = \left\{ \begin{aligned}
0 &\quad \textrm{ if }  & \lambda \notin [0,1]\\
G_\beta(f(\lambda)) & \quad \textrm{ if }  & \lambda \in ]0,1[ \\
\end{aligned} \right.
\end{equation}
In particular it has the sign of the imaginary part of $V_\gamma$ and it is continuous excepted for $\beta = \frac1{\pi}$ (the spectral singularity $z_\beta =\frac12$ is a singularity of $\im \xi(\cdot; H_\gamma, H_0)$).
\\
%DESSIN du GRAPHE???\\
\subsubsection{Real part of the SSF}

The real part of the SSF is 
\[\re \xi(\lambda; H_\gamma, H_0) =  \frac{1}{2\pi} \lim_{\varepsilon \to 0^+}
\left( \arg D_{V_\gamma}(\lambda + i \varepsilon) - 
 \arg D_{V_\gamma}(\lambda - i \varepsilon) \right),\]
 where the argument is chosen with the condition 
 \[\lim_{\lambda \to - \infty} \arg D_{V_\gamma}(\lambda) = \arg 1 =0.\]
\begin{itemize}
    \item  When $\beta \in ]0, \frac1{\pi}[$, the real part of $D_{V_\gamma}(\lambda \pm i \varepsilon)$ is always positive, then $\arg D_{V_\gamma}(\lambda \pm i \varepsilon) \in ]- \frac{\pi}2, \frac{\pi}2 [$ and there is no turning point. It follows that 
 \[ \lim_{\varepsilon \to 0^+} D_{V_\gamma}(\lambda + i \varepsilon) = \lim_{\varepsilon \to 0^+}
  D_{V_\gamma}(\lambda - i \varepsilon) \quad \Longrightarrow \quad 
  \lim_{\varepsilon \to 0^+}
\left( \arg D_{V_\gamma}(\lambda + i \varepsilon) - 
 \arg D_{V_\gamma}(\lambda - i \varepsilon) \right) =0.\]
It is the case for $\lambda \notin [0,1]$ because 
\[ \lim_{\varepsilon \to 0^+} D_{V_\gamma}(\lambda \pm i \varepsilon) = 1 + i \frac{\beta}2 \ln \left( \frac{(1-\lambda)^2}{\lambda^2} \right) = 1 + i \beta \ln (1-\lambda^{-1}).\]
For $\lambda \in ]0,1[$, by using again that the real part of $D_{V_\gamma}(\lambda \pm i \varepsilon)$ is always positive and that for $a>0$, $\arg (a+ib) = \arctan \frac{b}{a}$, for $f(\lambda)= \frac12 \ln \frac{(1-\lambda)^2}{\lambda^2} = \ln(\lambda^{-1}-1) $, we have:
\[ \lim_{\varepsilon \to 0^+}
\left( \arg D_{V_\gamma}(\lambda + i \varepsilon) - 
 \arg D_{V_\gamma}(\lambda - i \varepsilon) \right) = \arctan \frac{\beta f(\lambda)}{1-\beta \pi} - \arctan \frac{\beta f(\lambda)}{1+\beta \pi} = 2 \pi F_\beta(f(\lambda))\]
 where $F_\beta$ is the odd function defined on $\R$ by
 \[ F_\beta(X) := \frac{1}{2 \pi} \left( \arctan \frac{\beta X}{1-\beta \pi} - \arctan \frac{\beta X}{1+\beta \pi}\right), \qquad \beta \neq \frac1{\pi}.\]
 We deduce that for $\beta \in ]0, \frac1{\pi}[$, the SSF is the following continuous function supported on $[0,1]$:
 \begin{equation}\label{ssfB<}
   \xi(\lambda; H_\gamma, H_0)  = \left\{ \begin{aligned}
0 &\quad \textrm{ if }  & \lambda \notin [0,1]\\
(F_\beta + iG_\beta)(f(\lambda)) & \quad \textrm{ if }  & \lambda \in ]0,1[ \\
\end{aligned} \right. .
\end{equation}

\item Now, let us consider the case $\beta > \frac1{\pi}$ for which the spectrum of $H_\gamma$ is $[0,1]\cup \{z_\beta \}$ with $z_\beta:= \frac12 (1 + i \coth\frac1{2 \beta})$. 
In this case, the real part of $D_{V_\gamma}(\lambda - i \varepsilon)$ remains positive, while $D_{V_\gamma}(\lambda + i \varepsilon)$ turns around $0$. 
It will yield a phase shift of $2 \pi$ for $\lambda \in ]0,1[$. First, as above, for $\lambda <0$, both $D_{V_\gamma}(\lambda \pm i \varepsilon)$ tend to $  1 + i \beta \ln (1-\lambda^{-1})$ as $\varepsilon \searrow 0$
and 
 \[ 
  \lim_{\varepsilon \to 0^+}
\left( \arg D_{V_\gamma}(\lambda + i \varepsilon) - 
 \arg D_{V_\gamma}(\lambda - i \varepsilon) \right) =0.\]
For $\lambda \in ]0,1[$, we have 
\[ \lim_{\varepsilon \to 0^+} D_{V_\gamma}(\lambda \pm i \varepsilon) = 1 \mp \beta \pi + i \beta f(\lambda); \qquad f(\lambda)= \ln(\lambda^{-1}-1).\]
Then as before 
\[ \lim_{\varepsilon \to 0^+} \arg D_{V_\gamma}(\lambda - i \varepsilon) = \arctan \frac{\beta f(\lambda)}{1+\beta \pi}.\]
Contrariwise since $(1 - \beta \pi) <0$, for $\arg D_{V_\gamma}(\lambda + i \varepsilon)$, we have to take into account a phase shift of $\pi$ and 
\[ \lim_{\varepsilon \to 0^+} \arg D_{V_\gamma}(\lambda + i \varepsilon) = \pi +  \arctan \frac{\beta f(\lambda)}{1-\beta \pi} %= \pi -  \arctan \frac{\beta f(\lambda)}{\beta \pi-1}
.\]
For $\lambda >1$, $D_{V_\gamma}(\lambda + i \varepsilon)$ and $D_{V_\gamma}(\lambda - i \varepsilon)$ have the same limit (as for $\beta <\frac1{\pi}$) but the phase shift becomes $2 \pi$. So finally, for $\beta >\frac1{\pi}$ we obtain:
\begin{equation}\label{ssfB>}
   \xi(\lambda; H_\gamma, H_0)  = \left\{ \begin{aligned}
0 &\quad \textrm{ if }  & \lambda <0 \\
(\frac12 + F_\beta + iG_\beta)(f(\lambda)) & \quad \textrm{ if }  & \lambda \in ]0,1[ \\
1 &\quad \textrm{ if }  & \lambda >1 \\
\end{aligned} \right. .
\end{equation}

 \item 
 To conclude this example, let us discuss the case $\beta = \frac1{\pi}$ for which $z_\beta =\frac12$ is a spectral singularity. As in the previous cases, for $\lambda <0$, we obtain $\xi(\lambda; H_\gamma, H_0)=0$. 
For $\lambda \in ]0,1[$, the real part of $D_{V_\gamma}(\lambda + i \varepsilon)$ tends to $0^+$ when $\varepsilon \searrow 0$ and the imaginary part change of sign at $\lambda = \frac12$. Then, we have
%\[\lim_{\varepsilon \to 0^+} \arg D_{V_\gamma}(\lambda + i \varepsilon)  = \left\{ \begin{aligned}
%\frac{\pi}{2} &\quad \textrm{ if }  & \lambda \in ]0, \frac12[\\
%0 & \quad \textrm{ if }  & \lambda = \frac12 \\
%-\frac{\pi}{2} &\quad \textrm{ if }  & \lambda \in ]\frac12,1[\\
%\end{aligned} \right.\]
\[\lim_{\varepsilon \to 0^+} \arg D_{V_\gamma}(\lambda + i \varepsilon)  =
\frac{\pi}{2} \one_{]0, \frac12[} (\lambda) - \frac{\pi}{2} \one_{]\frac12, 1[}(\lambda),\]
while 
\[ \lim_{\varepsilon \to 0^+} \arg D_{V_\gamma}(\lambda - i \varepsilon) =  \arctan \frac{\beta f(\lambda)}{1+\beta \pi} = \arctan \frac{f(\lambda)}{2\pi}.\]
For $\lambda >1$, the real part of $D_{V_\gamma}(\lambda + i \varepsilon)$ and of $D_{V_\gamma}(\lambda - i \varepsilon)$ are both positive and since they have the same limit  there is no phase shift.  Consequently, when $\beta =\frac1{\pi}$, the SSF is the following function, discontinuous at the spectral singularity $z_\beta = \frac12$:
\begin{equation}\label{ssfB=}
   \xi(\lambda; H_\gamma, H_0)  = \left\{ \begin{aligned}
0 &\quad \textrm{ if }  & \lambda \notin [0,1]\\
 F_{\pi^{-1}}(f(\lambda)) & \quad \textrm{ if }  & \lambda \in ]0, 1[\setminus \{\frac12\} \\
\end{aligned} \right. ,
\end{equation}
with
\begin{equation}\label{Fpi-}
F_{\pi^{-1}}(X) : = \lim_{\beta \rightarrow (1/\pi)^{-}} F_\beta (X) = \left\{ \begin{aligned}
- \frac{1}{2 \pi} \left(\frac{\pi}2 + \arctan \frac{X}{2\pi}   \right)  &\quad \textrm{ if }  & X<0\\
  \frac{1}{2 \pi} \left(\frac{\pi}2 - \arctan \frac{X}{2\pi}  \right) & \quad \textrm{ if }  & X >0 .\\
\end{aligned} \right.
\end{equation}

\end{itemize}

 In order to visualize the difference between the three formulas \eqref{ssfB<}, \eqref{ssfB>} and \eqref{ssfB=}, let us look at the graphical representation of the real and imaginary parts of the spectral shift function $ \xi(\cdot; H_\gamma, H_0)$ in the particular cases $\beta=0.2 < \frac1{\pi}$, $\beta= \frac1{\pi}$ and $\beta=0.4 > \frac1{\pi}$.

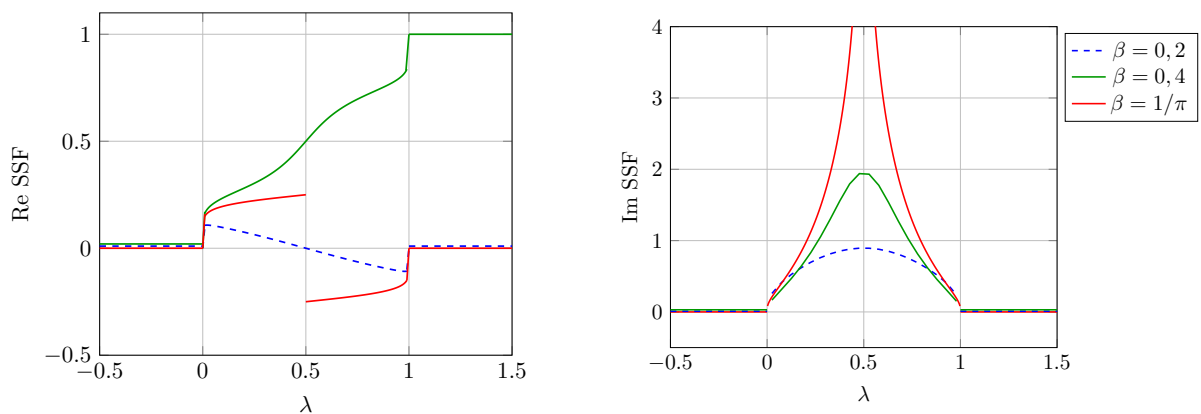
\begin{multicols}{2}

%==================== Partie Re de SSF

\begin{center}
\begin{tikzpicture}[scale=0.8]

\begin{axis}[
    xlabel = {$\lambda$},
    ylabel = {Re SSF},
    xmin = -0.5, xmax = 1.5,
    ymin = -0.5, ymax = 1.1,
    samples = 200,
    legend style={at={(1.02,0.98)},anchor=north west},
    grid = both
]

%b=0,2

\addplot[blue, thick, dashed, domain=0.01:0.99]{
 (rad(atan(0.2*ln((1-x)/x) / (1 - 0.2*pi )))
 -
 rad(atan(0.2*ln((1-x)/x) / (1 + 0.2*pi )))
   )/2/pi
};

%\addlegendentry{$\beta=0,2$}
%b=0,4
\addplot[green!60!black, thick,  domain=0.01:0.99]{
%\addplot{[green!60!black, thick, dashed, domain=0.001:0.999]{
0.5 + ( 
 rad(atan(0.4*ln(1/x -1) / (1 - 0.4*pi )))
 -
 rad(atan(0.4*ln(1/x -1) / (1 + 0.4*pi )))
   )/2/pi
};

%\addlegendentry{$\beta=0,4$}

%b=1/pi
\addplot[red, thick, domain=0.01:0.5]{
 0.25 - 
 (rad(atan(ln(1/x -1)/2/pi)
   )/2/pi
};
%\addlegendentry{$\beta=1/\pi$}

%b=0,2
\addplot [blue, thick, dashed, 
    domain=-0.5:0] {0.01};

%\addplot[blue, thick, dashed, domain=0.01:0.99]{
% (rad(atan(0.2*ln((1-x)/x) / (1 - 0.2*pi )))
% -
% rad(atan(0.2*ln((1-x)/x) / (1 + 0.2*pi )))
%   )/2/pi
%};

\addplot [blue, thick, dashed, 
    domain=1:1.5] {0.01};
    
\draw[blue, thick, dashed](0,0)
--(0.01,0.1);

\draw[blue, thick, dashed](0.99,-0.1)
--(1,0);

%\addlegendentry{$\beta=0,2$}
%b=0,4
%\addplot[green!60!black, thick, dashed, domain=0.01:0.99]{
%%\addplot{[green!60!black, thick, dashed, domain=0.001:0.999]{
%0.5 + ( 
% rad(atan(0.4*ln(1/x -1) / (1 - 0.4*pi )))
% -
% rad(atan(0.4*ln(1/x -1) / (1 + 0.4*pi )))
%   )/2/pi
%};

\addplot [green!60!black, thick, 
    domain=-0.5:0] {0.02};
    
\addplot [green!60!black, thick, 
    domain=1:1.5] {1};

\draw[green!60!black,  thick](0,0)
--(0.01,0.17);

\draw[green!60!black,  thick](0.9855,0.83)
--(1,1);

%\addlegendentry{$\beta=0,4$}

%b=1/pi
%\addplot[red, thick, dashed, domain=0.005:0.495]{
% 0.25 - 
% (rad(atan(ln(1/x -1)/2/pi)
%   )/2/pi
%};
%\addlegendentry{$\beta=1/\pi$}
\addplot[red, thick, domain=0.5:0.990]{
 -0.25 - 
 (rad(atan(ln(1/x -1)/2/pi)
   )/(2*pi)
};

\addplot [red, thick, 
    domain=-0.5:0] {0};
    
 \addplot [red, thick, 
    domain=1:1.5] {0};
    
\draw[red, thick](0,0)
--(0.01,0.16);

\draw[red, thick](0.99,-0.15)
--(1,0);

\end{axis}
\end{tikzpicture}
\end{center}

%==================== Partie Im de SSF 

\begin{center}
\begin{tikzpicture}[scale=0.75]
\begin{axis}[
    xlabel = {$\lambda$},
    ylabel = {Im SSF},
    xmin = -0.5, xmax = 1.5,
    ymin = -0.5, ymax = 4,
    samples = 200,
    legend style={at={(1.02,0.98)},anchor=north west},
    grid = both
]

%sur (0,1)
%\frac{1}{4\pi} \ln \left( 1 + \frac{4\beta \pi}{(1-\beta \pi)^2+\beta^2X^2} \right),\]
%    f(\lambda):= \ln (\lambda^{-1} -1)%

%%b=0,1
%\addplot[blue, thick, domain=0:1] {
% ln ( 1 + (4 * 0.1 * pi)/((1-0.1 * pi)^2+ 0.1*0.1 * ln(1/x -1)* ln(1/x -1) )/(4*pi)
%};
%\addlegendentry{$\beta=0,1$}

%b=0,2
\addplot[blue, thick, dashed]{
 ln ( 1 + (4 * 0.2 * pi)/((1-0.2 * pi)^2+ 0.2*0.2 * ln(1/x -1)* ln(1/x -1) )/(4*pi)
};
\addlegendentry{$\beta=0,2$}

%b=0,4
\addplot[green!60!black, thick]{
 ln ( 1 + (4 * 0.4 * pi)/((1-0.4 * pi)^2+ 0.4*0.4 * ln(1/x -1)* ln(1/x -1) )/(4*pi)
};
\addlegendentry{$\beta=0,4$}

%b=1/pi
\addplot[red, thick, domain=0:0.47]{
 ln ( 1 + 4 / (  ln(1/x -1)* ln(1/x -1)/(pi*pi))/(4*pi)
};
\addlegendentry{$\beta=1/\pi$}

\addplot[red, thick, domain=0.53:1]{
 ln ( 1 + 4 / (  ln(1/x -1)* ln(1/x -1)/(pi*pi))/(4*pi)
};
%\addlegendentry{$\beta=1/\pi$}

\addplot [red, thick, 
    domain=-0.5:0] {0};
    
 \addplot [red, thick, 
    domain=1:1.5] {0};

\addplot [blue, thick, dashed, 
    domain=-0.5:0] {0.01};

\addplot [blue, thick, dashed, 
    domain=1:1.5] {0.01};
    
    \addplot [green!60!black, thick, 
    domain=-0.5:0] {0.03};
    
\addplot [green!60!black, thick, 
    domain=1:1.5] {0.03};
    
\end{axis}
\end{tikzpicture}
\end{center}

\end{multicols}

When $0 < \beta < \frac1{\pi}$ the perturbation $V_{i\beta}$ does not create a new eigenvalue or spectral singularity, and $ \xi(\cdot; H_{i\beta}, H_0)$ is continuous on $\R$. For $\beta > \frac1{\pi}$, the jump of $1$ when $\lambda$ crosses the continuous spectrum $[0,1]$ can be explained by the appearance of the non-real eigenvalue $z_{\beta}$. Finally, for $\beta = \frac1{\pi}$, the real part of $ \xi(\cdot; H_{i\beta}, H_0)$ has a jump of high $\frac12$ at the spectral singularity $\lambda=\frac12$. Moreover at this point, the imaginary part of the SSF blows up.

%\end{itemize}

\section*{Acknowledgments}
This work was partially conducted within the France 2030 framework programme, Centre Henri Lebesgue ANR-11-LABX-0020-01. V.B is partially supported by the ANR-24-CE40-2939-01 grant. 
N.F. is supported by the R\'egion Pays de la Loire for the Connect Talent Project HiFrAn 2022 07750 led by Clotilde Fermanian Kammerer and the ANR-25-CE40-7296 (La Gabare). V.B and F.N. also thank the French GDR Dynqua for his support.

			\vspace{2cm}
		
		\noindent (V. Bruneau) Institut de Math\'ematiques de Bordeaux, UMR CNRS 5251, Universit\'e de Bordeaux 351 cours de la Lib\'eration, 33405 Talence cedex, France\\
    \textit{Email adress:} vbruneau@math.u-bordeaux.fr	\\

		\noindent	(N. Frantz) Univ Angers, CNRS, LAREMA, F-49000 Angers, France\\
\textit{Email address:} nicolas.frantz@univ-angers.fr\\

\noindent (F.Nicoleau) Laboratoire de Math\'ematiques Jean Leray, UMR CNRS 6629.
Nantes Universit\'e  F-44000 Nantes \\
         \emph{Email adress}: francois.nicoleau@univ-nantes.fr

		\end{document}